%% file: survey-final.tex
\documentclass[12pt,letterpaper]{article}
\usepackage[T1]{fontenc}
\usepackage[usenames,dvipsnames]{xcolor}
\usepackage{stmaryrd}
\usepackage{xfrac,footnote}
\usepackage{amsmath,amssymb,amsthm}
\usepackage{mathpazo}
\usepackage{fullpage}
\usepackage{bm}
\usepackage{float}
\usepackage{todonotes}
\usepackage{lipsum,complexity}
\usepackage{setspace}
\usepackage{scrextend}
\usepackage[shortlabels]{enumitem}
\usepackage[ruled,linesnumbered]{algorithm2e}

\usepackage[colorlinks=true]{hyperref}
\usepackage{url}
\input{thmmacros}

\input{lazymacros} 
\input{bibmacros}

\newcommand{\ehref}[1]{\href{mailto:#1}{#1}}
\renewcommand{\E}{\mathsf{E}}

\newcommand{\qedh}{\hfill\square}
\newcommand\inner[2]{\langle #1, #2 \rangle}
\newcommand{\outdim}{{\sf OuterDim}}
\newcommand{\indim}{{\sf InnerDim}}
\newcommand{\colsp}{{\sf colspace}}

\newtheorem{obs}{Observation}
\newcommand{\ssd}{{\sf SSD}}
\newcommand{\spar}{{sparsity}}
\newcommand{\rk}{{rank}}
\newcommand{\CSP}{{\sf CSP}}
\newcommand{\perm}{{perm}}
\newcommand{\rep}{{rep}}

\title{ Recent Progress on Matrix Rigidity \\
}
\author{	C. Ramya\thanks{\ehref{c.ramya@tifr.res.in}. Research supported by a fellowship of the DAE, Government of India.}\\
	Tata Institute of Fundamental Research, Mumbai, India.}

\begin{document}
	
\maketitle

\begin{abstract}
The concept of {\em matrix rigidity} was introduced by Valiant(independently by Grigoriev) in the context of computing linear transformations.  A matrix is {\em rigid} if it is {\em far} (in terms of Hamming distance) from any matrix of {\em low} rank. Although we know rigid matrices exist, obtaining {\em explicit} constructions of rigid matrices have remained a long-standing open question.  This decade has seen tremendous progress towards understanding matrix rigidity. In the past, several matrices such as Hadamard matrices and Fourier matrices were conjectured to be rigid. Very recently, many of these matrices were shown to have low rigidity. Further, several explicit constructions of rigid matrices in classes such as $\E$ and $\P^{\NP}$ were obtained recently. Among other things, matrix rigidity has found striking connections to areas as disparate as communication complexity, data structure lower bounds and error-correcting codes. In this survey, we present a selected set of results that highlight recent progress on matrix rigidity and its remarkable connections to other areas in theoretical computer science.
\end{abstract}

\tableofcontents

\section{Introduction}
\input{intro}

\section{Past Progress on Matrix Rigidity}
\label{sec:past-progress}

\input{prelims}

\section{Explicit constructions of Rigid Matrices}
\label{sec:lower-bounds}
In this section, we review semi-explicit constructions of rigid matrices starting with constructions in class  $\E^\NP$ proceeding towards constructions in $\P^{\NP}$. We use the term {\em semi-explicit} to broadly refer to matrices that require worse than polynomial time to construct them.

\subsection{Rigidity of Random Toeplitz matrices}
\input{goldreich-tal}

\subsection{Construction of rigid matrices in sub-exponential time}

\input{kumar-volk}

\subsection{Construction of rigid matrices based on probabilistically checkable proofs}
\input{rigid-pcp}

\input{elimination-rigid}

\section{Upper Bounds on Matrix Rigidity}
\label{sec:upper-bounds}
In this section, we survey some of the recent developments on mathematical techniques involved in proving the non-rigidity of some of the matrix families that were previously conjectured to be rigid.

\input{walsh-hadamard}

\input{function-matrix}

\input{generalized-hadamard}

\section{Matrix rigidity via static data structure lower bounds}
\label{sec:data-structures}
\input{data-structure-lb}

\section{Matrix rigidity and error-correcting codes}
\label{sec:codes}
\input{rigidity-codes}

\section{Discussion and Open problems}
\label{sec:discussion}
\input{discussion}

\paragraph*{Acknowledgements} I am grateful to Ramprasad Saptharishi for introducing to me the concept of matrix rigidity. I thank Ramprasad Saptharishi, Anamay Tengse and Prerona Chatterjee for numerous technical discussions on the various papers presented in this article. I thank Prahladh Harsha for providing several clarifications on the results in subsection \ref{subsec:pcp-constructions}.

\bibliographystyle{plainurl}
\bibliography{refbib}

\end{document}

%% file: thmmacros.tex
\usepackage{amsthm}
\usepackage{thmtools,thm-restate}

\numberwithin{equation}{section}
\declaretheoremstyle[bodyfont=\it,qed=\qedsymbol]{noproofstyle}

\declaretheorem[name=Observation,numbered=no]{observation*}

\declaretheorem[numberlike=equation]{theorem}

\declaretheorem[name=Theorem,numbered=no]{theorem*}

\declaretheorem[numberlike=equation]{lemma}
\declaretheorem[name=Lemma,numbered=no]{lemma*}

\declaretheorem[numberlike=equation]{corollary}
\declaretheorem[name=Corollary,numbered=no]{corollary*}

\declaretheorem[name=Proposition,numbered=no]{proposition*}

\declaretheorem[numberlike=equation]{claim}
\declaretheorem[name=Claim,numbered=no]{claim*}

\declaretheorem[name=Conjecture,numbered=no]{conjecture*}

\declaretheorem[numberlike=equation]{question}
\declaretheorem[name=Question,numbered=no]{question*}

\declaretheoremstyle[bodyfont=\it,qed=$\lozenge$]{defstyle} 

\declaretheorem[numberlike=equation,style=defstyle]{definition}
\declaretheorem[unnumbered,name=Definition,style=defstyle]{definition*}

\declaretheorem[unnumbered,name=Example,style=defstyle]{example*}

\declaretheorem[numberlike=equation,style=defstyle]{notation}
\declaretheorem[unnumbered,name=Notation=defstyle]{notation*}

\declaretheorem[unnumbered,name=Construction,style=defstyle]{construction*}

\declaretheorem[numberlike=equation,style=defstyle]{remark}
\declaretheorem[unnumbered,name=Remark,style=defstyle]{remark*}


%% file: lazymacros.tex
\renewcommand{\phi}{\varphi}
\renewcommand{\epsilon}{\varepsilon}


%% file: bibmacros.tex
\usepackage{nth}
\usepackage{intcalc}
\usepackage{etoolbox}
\usepackage{xstring}

\usepackage{ifpdf}
\ifpdf
\else
\usepackage[quadpoints=false]{hypdvips}
\fi

\newcommand{\ECCCurl}[2]{http://eccc.hpi-web.de/report/\ifnumcomp{#1}{>}{93}{19}{20}#1/#2/}

\newcommand{\shortECCC}[2]{\texttt{\href{http://eccc.hpi-web.de/report/\ifnumcomp{#1}{>}{93}{19}{20}#1/#2/}{eccc:TR#1-#2}}}

\newcommand{\parseECCC}[1]{
\StrSubstitute{#1}{TR}{}[\tmpstring]%
\IfSubStr{\tmpstring}{/}{ 
\StrBefore{\tmpstring}{/}[\ecccyear]%
\StrBehind{\tmpstring}{/}[\ecccreport]%
}{
\StrBefore{\tmpstring}{-}[\ecccyear]%
\StrBehind{\tmpstring}{-}[\ecccreport]%
}%
\shortECCC{\ecccyear}{\ecccreport}}

%% file: intro.tex
The concept of {\em matrix rigidity} was introduced by Valiant \cite{Val77} in the context of computing linear transformations by arithmetic circuits and was also studied independently by Grigoriev in \cite{Gri76}.

The rigidity of a matrix $A\in\mathbb{F}^{n\times n}$ for rank $r$ over $\mathbb{F}$ (denoted by $R_A^{\mathbb{F}}(r)$) is the minimum number of entries to be changed in $A$ so that rank of matrix $A$ becomes $r$. More formally,  $$R_A^{\mathbb{F}}(r) \triangleq \min_{C}\{ \spar(C) \mid C\in\mathbb{F}^{n\times n},~ \rk(A+C)\leq r \}$$ 
where $\spar$ of a matrix $C$ denotes the number of non-zero entries in $C$.

A matrix is {\em rigid} if it is far in terms of Hamming distance from any low rank matrix. {\em Matrix rigidity} is an interesting and intriguing concept in that sense that it intertwines a combinatorial property such as the sparsity of a matrix with an algebraic property namely the rank of a matrix.


For instance, the rigidity of an $n\times n$ identity matrix $I_n$ for rank $r$ is exactly $(n-r)$. Trivially, for any matrix $A\in\mathbb{F}^{n\times n}$ and for any $r \leq n$, the rigidity of $A$ is at most $n^2$. In fact, it is not difficult to observe that for any matrix $A\in\mathbb{F}^{n\times n}$ and for any $r \leq n$, $R_A^{\mathbb{F}}(r) \leq (n-r)^2$. Moreover, over finite fields {\em most} matrices have high rigidity (rigidity close to the upper bound).  Further, over infinite fields, for every choice of $n$ there exists an $n\times n$ matrix $A$ such that $R_A(r)=(n-r)^2$ for any $r$. Although the existence of rigid matrices is quite straight-forward, the major goal is to prove a super-linear lower bound on the rigidity of explicit $n\times n$ matrices. We say a sequence of matrices $\{A_n\}_{n\in \mathbb{N}}$ is {\em explicit} if there exists a deterministic algorithm that on input $n$ (in unary) outputs $A_n$ in time $\poly(n)$. The following question was posed by Valiant in  \cite{Val77} and has remained a tantalizing open problem:

\begin{question}
\label{que:qu1}
Does there exist an explicit sequence of matrices $(A_n)_{n\in\mathbb{N}}$ with entries in $\mathbb{F}$ such that $R_{A_n}^{\mathbb{F}}(\epsilon n) = \Omega(n^{1+\delta})$ for some $\epsilon,\delta >0$?
\end{question}

As mentioned earlier Question \ref{que:qu1} has connections to arithmetic circuits computing linear transformations. The study of linear transformations are central to linear algebra. Linear transformations such as the Discrete Fourier Transform, Fast Fourier Transform are of practical importance. A {\em linear circuit} is a directed acyclic graph C where every gate is either an input gate or computes a linear combination of its inputs. The {\em size} of a linear circuit is the number of edges in it and the {\em depth} of a linear circuit is the length of the longest path from the input to the output gate. Valiant \cite{Val77} observed that if
any linear transformation is computable by a {\em small}-size {\em small}-depth linear circuit then the
corresponding transformation matrix does not have {\em high} rigidity. In other words, for any $A\in\mathbb{F}^{n\times n}$ if $R_A(\epsilon n)\geq n^{1+\delta}$ for some $\epsilon,\delta>0$ then any linear circuit computing the transformation $A:x\mapsto A\cdot x$ must have either size $\Omega(n\log \log n)$ or depth $\Omega(\log n)$. Thus, rigidity lower bounds imply super-linear size lower bounds on linear circuits of logarithmic depth.  This brings us to the following question, a variant of Question \ref{que:qu1}:

\begin{question}
\label{que:rigidity-Valiant}
Does there exist an explicit sequence of matrices $(A_n)_{n\in\mathbb{N}}$ with entries in $\mathbb{F}$ such that $R_{A_n}^{\mathbb{F}}\left(\frac{n}{\log\log n}\right) = \Omega(n^{1+\delta})$ for some $\delta >0$?
\end{question}

The earliest works on matrix rigidity were due to Valiant\cite{Val77} and Razborov\cite{Raz89}.

The connections between communication complexity of boolean functions and matrix rigidity were first explored by Razborov\cite{Raz89}. Whenever we think of matrices in the  communication complexity setting the most natural candidates are communication matrices of boolean functions. Consider the two-party communication model with two parties {\em Alice} and {\em Bob} who want to jointly compute a boolean function  $f:\{0,1\}^{2n}\rightarrow \{0,1\}$ where the input is partitioned between the two parties. For any boolean function $f:\{0,1\}^{n} \times \{0,1\}^n\rightarrow \{0,1\}$ , the {\em communication matrix} $M_f$ is a $2^n\times 2^n$ where the rows and columns are indexed by strings in $\{0,1\}^n$ and $M_f[x,y]=f(x,y)$ for all $x,y\in \{0,1\}^n$.

Razborov in \cite{Raz89,Wun12} considered the complexity class $\PH^{cc}$, the communication complexity analogue of the polynomial hierarchy (see \cite{GPW18} for a formal definition of $\PH^{cc}$) and showed that for any function $f:\{0,1\}^n\rightarrow \{0,1\}$ in $\PH^{\cc}$, $R_{M_f}(2^{(\log n/\delta)^c})\leq \delta\cdot 2^{2n}$ where $\delta >0$ is arbitrary constant and $M_f$ is the $2^n\times 2^n$ communication matrix. Thus, lower bounds on the rigidity of explicit matrices immediately imply communication complexity lower bounds, a long-standing open question. This leads us to the following question which is quite similar to that of Question \ref{que:qu1} except for the parameters:

\begin{question}
\label{que:rigidity-Razborov}
For $N=2^n$, does there exist an explicit sequence of matrices $(A_N)_{N\in\mathbb{N}}$ with entries in $\mathbb{F}_2$ such that $R_{A_N}^{\mathbb{F}_2}(2^{(\log n/\delta)^c})\geq \delta\cdot 2^{2n}$ for some $\delta >0$?
\end{question}

Although we have not been able to obtain decisive answers to any of these questions, there has been considerable progress towards understanding Questions \ref{que:qu1}, \ref{que:rigidity-Valiant} and \ref{que:rigidity-Razborov} in the recent years. In fact, several interesting matrix families {\em were} conjectured to be rigid:

\begin{quote}    
"Many candidate matrices are conjectured to have rigidity as high as in Valiant’s question. Examples include Fourier transform matrices, Hadamard matrices, Cauchy matrices, Vandermonde matrices, incidence matrices of projective planes, etc." 
\hfill ---Page 15, \cite{Lok09}.
\end{quote}

In this article, we survey some of the recent developments  on the non-rigidity of some of the matrix families conjectured above. In particular, we review the following results:
\begin{itemize}
\item Non-rigidity of {\em Walsh-Hadamard matrix} by Alman and Williams\cite{AW17};
\item Non-rigidity of {\em generalized Hadamard matrices} due to Dvir and Liu\cite{DL19};
\item Non-Rigidity of certain matrices associated with functions over finite fields\cite{DE17}; and
\item Non-Rigidity of {\em Fourier} and {\em circulant} matrices\cite{DL19}.
\end{itemize}

Even though we are currently far away from answering Questions \ref{que:qu1}, \ref{que:rigidity-Valiant} and \ref{que:rigidity-Razborov}, several {\em semi-explicit} constructions of rigid matrices were obtained quite recently. In this regard, we survey the following results:

\begin{itemize}
\item Rigidity of Random Toeplitz matrices by Goldreich and Tal\cite{GT18};
\item Sub-exponential time constructions of rigid matrices\cite{KV19}; and
\item Explicit rigid matrices in the class $\P^{\NP}$ based on constructions of probabilistically checkable proofs($\PCP$s)\cite{AC19,BHPT20}.
\end{itemize}

However, the parameters in the above mentioned results are different from each other. The first two of the above mentioned results are towards answering Question \ref{que:qu1} while the construction in \cite{AC19} is in the spirit of answering Question \ref{que:rigidity-Razborov}.

Despite consistent efforts in obtaining rigid matrices, answering Question \ref{que:qu1} seems to be a distant dream. This difficulty is justified by understanding connections between explicit constructions of rigid matrices and other {\em hard} problems in theoretical computer science such as explicit constructions of error-correcting codes, communication complexity lower bounds as well as data structure lower bounds. In this regard, we discuss in detail the following recent connections between matrix rigidity and 
data structure lower bounds as well as linear codes:
\begin{itemize}
\item Proving data structure lower bounds is a fundamental open problem in theoretical computer science. A major goal  has been to understand {\em time-space tradeoffs}. That is, in the static setting how does one optimize space such that data structure queries can be answered quickly. In \cite{DGW19}, the authors show that a super-logarithmic lower bound on the query time of a linear data structure with linear space implies an answer to Question \ref{que:rigidity-Valiant} where the matrix of high rigidity is constructible in the class $\P^{\NP}$. Though 
constructions of rigid matrices in $\P^{\NP}$ are now available to us via $\PCP$s.

\item In the theory of error-correcting codes, {\em linear codes} are particularly useful. One can verify that {\em asymptotically good codes} yield generator matrices of high rigidity. We review a result by Dvir \cite{Dvir11} which states that if the generating matrix of a locally decodable code is not rigid, then it
defines a locally self-correctable code with rate close to one.
\end{itemize}

Before we delve into proving upper and lower bounds on matrix rigidity, let us investigate the computational complexity of computing the rigidity of a given matrix. Consider the problem ${\sf RIGID}(A,\mathbb{F},s,r)$ of deciding if $R_A^{\mathbb{F}}(r)\leq s$ given a matrix $A\in\mathbb{F}^{n\times n}$ and $s,r\in \mathbb{Z}^{+}$. Note that we can guess a matrix $S\in \mathbb{F}^{n\times n}$ of sparsity at most $s$ and test if $\rk(A-S)\leq r$.

\begin{itemize}
    \item Over finite fields, this problem is in the class $\NP$ and in fact ${\sf RIGID}(A,\mathbb{F}_q,s,r)$ is known to be $\NP$-complete\cite{Des07}.
    \item Given a matrix $A\in\mathbb{R}^{n\times n}$ let $m=\rk(A)$. We can brute force over all matrices of sparsity at most $s$ and test if there is a setting of these $s$ entries to real numbers such that the $rank(S)\leq r-m$.  This is in ${\sf PSPACE}$ as computing the minimum rank of a pattern matrix is the class $\exists\mathbb{R}$ (existential theory of reals). Hence,  ${\sf RIGID}(A,\mathbb{R},s,r)$ is in ${\sf PSPACE}$ (where the underlying computational model can handle real numbers).
 \item ${\sf RIGID}(A,\mathbb{Q},s,r)$ is not known to be decidable.
\end{itemize}

In the parameterized regime, ${\sf RIGID(A,\mathbb{F}_q,s,r)}$ is known to be fixed parameter tractable when $\mathbb{F}=\mathbb{F}_q$. The computational complexity of several variants of ${\sf RIGID(A,\mathbb{F},s,r)}$ has been studied extensively in \cite{MS10}.

\paragraph*{Organization of the article.} Rest of the article is organized as follows. In Section \ref{sec:past-progress}, we begin with some basic facts on rigidity. The goal of Section \ref{sec:past-progress} is to understand certain important classical progress made towards understanding rigidity so that successive sections are more accessible to the reader. For details on statements and proofs in Section \ref{sec:past-progress}, we refer the reader to an excellent survey by Satya Lokam \cite{Lok09} and references therein. In Section \ref{sec:lower-bounds}, we review recent explicit constructions of matrices that achieve rigidity to a large extent possible.  As mentioned earlier, several well-known families of matrices were recently ruled out from having rigidity and we survey these results in Section \ref{sec:upper-bounds}. In Sections \ref{sec:data-structures} and \ref{sec:codes} we investigate the connections between rigid matrices, static data structure lower bounds and error-correcting codes. We conclude with some open problems in Section \ref{sec:discussion}.


%% file: prelims.tex
To begin with, we prove a straight-forward upper bound on the rigidity of any matrix.

\begin{lemma}
\label{lem:rigid-upperbound}
Let $A$ be an $n\times n$ matrix with entries from $\mathbb{F}$. For any $r\leq n$, $R_A^{\mathbb{F}}(r)\leq (n-r)^2$. 
\end{lemma}
\begin{proof}
Let $A\in \mathbb{F}^{n\times n}$. If $\rk(A)\leq r$ then $R_{A}(r)=0$. Therefore, assume $\rk(A)>r$. Then there exists $r$ linearly independent rows $R_1,\ldots,R_r$ and $r$ linearly independent columns $C_1,\ldots,C_r$ in $A$. The rows and columns of $A$ can be 
permuted such that $R_1,\ldots,R_r$ and $C_1,\ldots,C_r$ are the first $r$ rows and columns of $A$ respectively (denoted by sub-matrix $A_{11}$ in Figure \ref{fig:rigidity-upperbound}). Let 
$R_1',\ldots,R_{n-r}'$
and $C_1',\ldots,C_{n-r}'$ be rows and columns of $A_{21}$ and $A_{12}$ respectively. Observe that there exists constants $\alpha_{i_1},\ldots,\alpha_{i_r}$ in $\mathbb{F}$ such that for all $i\in [n-r]$,  $R_i'=\alpha_{i_1}R_1+\ldots+ \alpha_{i_r}R_r$.  
\begin{figure}[H]
\centering
\includegraphics[scale=0.7]{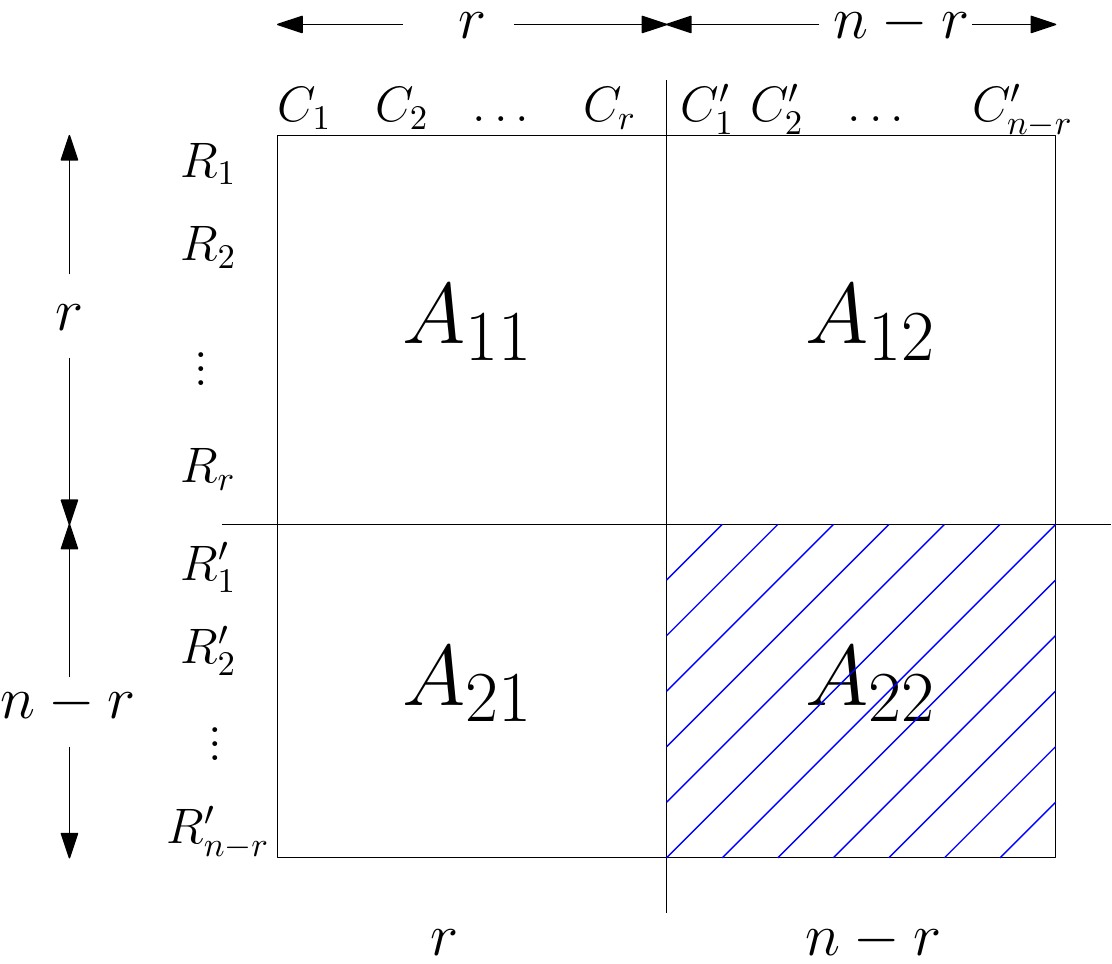}
\caption{Rigidity upper bound on a matrix $A$}
\label{fig:rigidity-upperbound}
\end{figure}

Now, by altering the entries of the sub-matrix $A_{22}$ based on the values of $\alpha$'s, we can ensure that every row of $A$ is a linear combination of the rows of the sub-matrix $[A_{11} ~\vert~ A_{12}]$ implying that
rank of $A$ is $r$. Since $A_{22}\in \mathbb{F}^{(n-r)\times (n-r)}$, we get $R_A^{\mathbb{F}}(r)\leq (n-r)^2$.\\
\end{proof}

In fact,  over finite fields, for {\em most} matrices the above rigidity upper bound is tight.

\begin{lemma}
Let $\mathbb{F}_q$ be a finite field. The fraction of $n \times n$ matrices over $\mathbb{F}_q$ with rigidity at most $(n-r)^2/\log n$ for rank $r$ is $O(1/n)$.
\end{lemma}

\begin{proof}
 For any matrix $A\in \mathbb{F}_q^{n\times n}$, $R_A(r)\leq s$ if $A=S+L$ where $\rk(L)\leq r$ and $S$ has at most $s$ non-zero entries. Therefore, we first count the number of matrices  over $\mathbb{F}_q$ of rank at most $r$ and of sparsity at most $s$. The number of matrices over $\mathbb{F}_q$ of rank at most $r$ is at most $\binom{n}{r}\cdot q^{2r(n-r)}$. The number of matrices over $\mathbb{F}_q$ of sparsity at most $s$ is at most $\binom{n^2}{s}\cdot q^s$. Thus, the number of $n \times n$ matrices over $\mathbb{F}_q$ with rigidity at most $s$ for rank $r$ is at most $q^{2nr-r^2+s+2s\log_q n + n\log_q 2}$. When $r\leq n-c_1\sqrt{n}$ and $0\leq s < c_2(n-r)^2/\log n$ for some constants $c_1, c_2$, we have $q^{2nr-r^2+s+2s\log_q n + n\log_q 2} \leq O(1/n)\cdot q^{n^2}$. \\
\end{proof}

As a corollary, the fraction of $n \times n$ matrices over $\mathbb{F}_q$ with rigidity at most $\frac{(n-r)^2}{\log n}$ for rank $r$ approaches 1. Hence, almost all matrices over $\mathbb{F}_q$ have rigidity $\Omega\left(\frac{(n-r)^2}{\log n}\right)$ for rank $r$. Similarly, over fields of infinite characteristic one can show that  for every choice of $n$ there exists an $n\times n$ matrix $A$ such that $R_A(r)=(n-r)^2$ for any $r$.

 The best known lower bounds on rigidity for explicit matrices over finite fields is an $\Omega(\frac{n^2}{r}\log \frac{n}{r})$ for $\log^2 n \leq r \leq n/2$ due to Friedman \cite{Fri93}. The best known lower bounds on rigidity for explicit matrices is an $\Omega(\frac{n^2}{r}\log \frac{n}{r})$ for $\log^2 n \leq r \leq n/2$ due to  Shokrollahi, Spielman, and Stemann \cite{SSS97}. The lower bounds in \cite{SSS97} apply to any {\em totally regular matrix} and use the following combinatorial approach called the {\em untouched minor argument}: \\

\noindent\textbf{\textit{The untouched minor argument.}} 
\begin{quote}
Consider a matrix $A$ almost all of whose minors have  rank $\Omega(r)$. Then, even after changing a {\em few} entries in $A$, there is at least one minor in $A$ that is "untouched" and the rank of $A$ remains $\Omega(r)$. Thus, in order to reduce the rank of $A$ to less than $r$, every minor in $A$ must be altered requiring a large number of entries of $A$ to be changed.    
\end{quote}

Matrices all of whose minors are full-rank are called {\em totally regular} matrices. A standard example of a totally regular matrix is the Cauchy matrix $C=\{ c_{ij}\}_{i,j\in [n]}, c_{ij}=\frac{1}{x_i+y_j}$ for $2n$ distinct
elements $x_1,\ldots,x_n,y_1,\ldots,y_n \in \mathbb{F}$.

\begin{theorem}
\label{thm:totally-regular}
Let $M$ be any totally regular matrix and $\log^2 n \leq r \leq n/2$. Then, $R_{M}(r)=\Omega(\frac{n^2}{r}\log \frac{n}{r})$. 
\end{theorem}

\noindent \textit{Proof Sketch.} Let $M$ be any totally regular matrix. For the sake of contradiction, assume that $R_{M}(r)=o(\frac{n^2}{r}\log \frac{n}{r})$. Then rank of $M$ can be reduced to $r$ by altering $o(\frac{n^2}{r}\log \frac{n}{r})$ entries of $M$. The entries of $M$ can be viewed as a bipartite graph $G_M=(U,V,E)$ with $|U|=|V|=n$ where $(u,v)\in E(G_M)$  if and only if  entry $M_{u,v}$ was not altered to reduce the rank of $M$. Intuitively as $G_M$ has many edges, it is likely that $G_M$ has a reasonably large complete bipartite subgraph. If fewer than $\frac{n^2}{r}\log \frac{n}{r}$ entries were changed in $M$, then $G_M$ has at least $n^2-(\frac{n^2}{r}\log \frac{n}{r})$ edges when $r\leq n/2$. In order to show this, we appeal to the {\em Zarankiewicz problem} in extremal graph theory that counts the maximum number of edges in any bipartite graph that forbids a reasonably large complete bipartite subgraph. If $r\geq \log^2 n$ then any bipartite graph with at least $n^2-(\frac{n^2}{r}\log \frac{n}{r})$ edges has a complete bipartite subgraph $K_{r+1,r+1}$. This immediately implies that there is an $(r+1)\times (r+1)$ sub-matrix $M'$ of $M$ that remains untouched. As $M$ be any totally regular, $\rk(M')= r+1$ a contradiction. Hence, $R_{M}(r)=\Omega(\frac{n^2}{r}\log \frac{n}{r})$. 

$\qedh$ 

However, the untouched minor argument has its own limitations and cannot be improved to obtain better lower bounds so as to answer Question \ref{que:rigidity-Valiant}(see Section 2.2.1 in \cite{Lok09} for more details).

Recall that the major goal here is to prove an $\Omega(n^{1+\delta})$ lower bound on the rigidity of $n\times n$ matrices for rank $\epsilon n$ for some $\epsilon,\delta >0$. In fact, we now demonstrate $\Omega(n^2)$ lower bounds on the rigidity of certain matrices for rank $\Omega(n)$ from \cite{Lok06,Lok00}.

\begin{theorem}
\label{thm:ssd-rigidity}
Let $p_{11},p_{12},\ldots,p_{nn}$ be $n^2$ distinct primes. Let $P$ be the $n\times n$ matrix given by $P_{ij}=\sqrt{p_{ij}}$. Then, $R_{P}(r)\geq n(n-16r)$. 
\end{theorem}

\noindent {\em Proof Sketch.} The proof is an algebraic argument based on the {\em Shoup-Smolensky dimension} of a matrix. The Shoup-Smolensky dimension of a matrix $P$ of order $nr$ (denoted by $\ssd_{nr}(P)$) is the dimension of the vector space spanned by the set of products of $nr$ distinct elements of $P$. Observe that the elements of $P$ are algebraically independent. Hence, $\ssd(P)$ of order $nr$ is the number of polynomials in $n^2$ variables of degree $nr$ which is at most $\binom{n^2}{nr}$. A lower bound of $\binom{n^2-R_P(r)}{nr}$  on the Shoup-Smolensky dimension of the matrix $P$ follows from the fact that entries of $P$ are algebraically independent and is not very difficult to observe. 

$\qedh$

Note that the matrix $M$ in Theorem \ref{thm:totally-regular} is explicit while matrix $P$ in Theorem \ref{thm:ssd-rigidity} is not. Hence, on one hand
from Theorem \ref{thm:totally-regular} we have explicit matrices that are not-so rigid and on the other hand in Thorem \ref{thm:ssd-rigidity} we have matrices such as the ones constructed from distinct primes that are highly rigid but not explicit.

On first thoughts, it is intriguing to note Valiant's claim that if we answer Question \ref{que:rigidity-Valiant} in the affirmative then the linear transformation corresponding to matrix $A_n$ cannot be computed by linear circuits of size $O(n)$ and depth $O(log n)$. Similar to the proof of Theorem \ref{thm:totally-regular} Valiant's argument from \cite{Val77} that we now outline is also   graph-theoretic .

\begin{theorem}
\label{thm:valiant-lin-ckt}
If $A\in\mathbb{F}^{n\times n}$ has a linear circuit of size $O(s)$ and depth $O(d)$ then $R_A(\frac{s\epsilon}{\log d})\leq n\cdot 2^{O(d/2^\epsilon)}$ for every $\epsilon>0$.
\end{theorem}

\noindent\textit{Proof Sketch.} The proof is based on a graph-theoretic argument that by removing a few edges the length of every path in a directed acyclic graph can be reduced by a factor of 2. That is, from any directed acyclic graph having $s$ edges and (every) path length bounded by $d$, by removing at most $s/\log d$ edges we can ensure that every path has length at most $d/2$. Let ${\cal C}$ be a linear circuit of size $s$ and depth $d$ computing the matrix $A\in\mathbb{F}^{n\times n}$. By repeating the above mentioned edge removal process $\epsilon$ times, ${\cal C}$ has at most $s\epsilon/\log d$ edges and every path in ${\cal C}$ has length at most $d/2^{\epsilon}$. As ${\cal C}$ is a linear circuit the linear function computed by output gate of ${\cal C}$ is a linear combination of the removed edges and the input gates. This implies that $A=S+L$ where $\rk(L)\leq \frac{s\epsilon}{\log d}$ and every row in $S$ has at most $ d/2^\epsilon$ many non-zero entries.

$\qedh$

By setting $s=n\log\log n$ and $d=\log n$ in Theorem \ref{thm:valiant-lin-ckt}, we can immediately conclude that for any $A\in\mathbb{F}^{n\times n}$ if $R_A(\epsilon n)= \Omega(n^{1+\delta})$ for some $\epsilon,\delta>0$ then any linear circuit computing the transformation $A:x\mapsto A\cdot x$ must have either size $\Omega(n\log \log n)$ or depth $\Omega(\log n)$. In essence, a positive answer to Question \ref{que:rigidity-Valiant} implies super-linear size lower bounds on linear circuits of logarithmic depth.

%% file: goldreich-tal.tex
Observe that a random matrix is {\em rigid} with {\em high} probability. Goldreich and Tal\cite{GT18} showed that in order to obtain rigid matrices it is enough to look only inside the space of random Toeplitz matrices.

Let $a_{-(n-1)},\ldots,a_{n-1}$ be $2n-1$ elements in $\mathbb{F}$.
A {\em Toeplitz matrix} $T\in\mathbb{F}^{n\times n}$ is given by $T_{ij}=a_{j-i}$ for all $i,j\in [n]$. 
A {\em Hankel matrix} $H\in\mathbb{F}^{n\times n}$ is given by $H_{ij}={a_{i+j}}$ where $a_2,\ldots,a_{2n}$ are in $\mathbb{F}$. 
The matrices $T$ and $H$ mentioned below are examples of  $3\times 3$ Toeplitz and Hankel matrices respectively:
\begin{center}
    $T= \begin{bmatrix}
a_0 & a_1 & a_2 \\
a_{-1} & a_0 & a_1 \\
a_{-2} & a_{-1} & a_0
\end{bmatrix}$
and $H= \begin{bmatrix}
a_2 & a_3 & a_4 \\
a_3 & a_4 & a_5 \\
a_4 & a_5 & a_6
\end{bmatrix}$.
\end{center}
A matrix $T\in\mathbb{F}_2^{n\times n}$ is a {\em random} Toeplitz (resp., Hankel) matrix $T_{ij}=a_{j-i}$ (resp., $H_{ij}={a_{i+j}}$) where $a_{-(n-1)},\ldots,a_{n-1}$ (resp., $a_2,\ldots,a_{2n}$) are bits in $\{0,1\}$ chosen independently and uniformly at random. Goldreich and Tal\cite{GT18} show that with high probability, a random Toeplitz matrix (resp., Hankel matrix) in $\mathbb{F}_2^{n\times n}$ is rigid. Observe that a Hankel matrix is the mirror image of a Toeplitz matrix. Hence, rigidity of Hankel matrices translates directly to rigidity of Toeplitz matrices. In this section, we prove the following result from \cite{GT18}: 

\begin{theorem}
\label{thm:toeplitz}
Let $H\in \mathbb{F}_2^{n\times n}$ be a random Hankel matrix.  For every $r\in [\sqrt{n},n/32]$, ${\cal R}_{H}(r)=\Omega\left(\frac{n^3}{r^2\log n}\right)$ with probability $1-o(1)$.
\end{theorem}


\noindent{\em Proof Sketch.} The high-level idea is to come up with a procedure ${\sf TEST}_{s,r}(H)$ which when given as input a Hankel matrix $H\in\mathbb{F}_2^{n\times n}$ does the following:

\begin{itemize}
    \item[(1)] If $H=S+L$ with $\spar(S)\leq s$ and $\rk (L)\leq r$ then reject $H$.  
    \item[(2)] If $H$ is  a random matrix then accept $H$ with probability $1-o(1)$.
\end{itemize}
If we succeed in obtaining such a test then for a random Hankel matrix $H\in\mathbb{F}_2^{n\times n}$, ${\sf TEST}_{s,r}(H)$ accepts $H$ with probability $1-o(1)$. Then with probability $1-o(1)$, ${\cal R}_H(r)=\Omega(\frac{n^3}{r^2\log n})$ when $r\in [\sqrt{n},n/32]$ and $s=\frac{n^3}{160 r^2\log n}$.

$\qedh$

The design of ${\sf TEST}_{s,r}(H)$ depends on the following simple observation that if $H$ is not rigid then there is a {\em super-sparse} sub-matrix of $H$ that witnesses the non-rigidity of $H$:

\begin{obs}
\label{obs:php}
Let $H\in \mathbb{F}^{n\times n}$ be a Hankel matrix such that $H=S+L$ for some $S,L\in \mathbb{F}_2^{n\times n}$ with $\spar(S)\leq s$ and $\rk (L)\leq r$. Then for every $2r\times 2r$ sub-matrix $H'$ of $H$ there exists  $S',L'\in\mathbb{F}_{2}^{2r\times 2r}$ such that $H'=S'+L'$ and $\spar(S')\leq \frac{s}{(n/2r)^2}$ and $\rk(L')\leq r$.
\end{obs}


Based on Observation \ref{obs:php}, we design ${\sf TEST}_{s,r}(H)$:

\IncMargin{2.5em}
\begin{algorithm}[H]
  \SetKwInOut{Input}{Input}
  \Input{Hankel matrix $H\in\mathbb{F}_2^{n\times n}$}
  \BlankLine
  {\em Partition} $H$ into $(n/2r)^2$ many matrices of dimension $2r\times 2r$ each. \footnotemark
 \BlankLine
  Set $s'= \frac{s}{(n/2r)^2}$. \\
  \For{every such sub-matrix $H'$ of $H$}{ 
  \For{every $s'$-sparse matrix $S'$ in $\mathbb{F}_2^{n\times n}$}{ 
    \If{$\rk(H'-S')\leq r$}{\texttt{reject $H$}\label{alg1:line-reject}} 
    }}
    \BlankLine
    \texttt{Accept $H$}
\end{algorithm}
\footnotetext{An arbitrary partition of $H$ may not work. We need to carefully partition $H$ so that the probability bounds work.}
\DecMargin{2.5em}


If the given Hankel matrix $H$ in $\mathbb{F}_2^{n\times n}$ is not rigid then by Observation \ref{obs:php}, line (\ref{alg1:line-reject}) of Algorithm is reached for some $s'$-sparse sub-matrix $S'$ and ${\sf TEST}_{s,r}(H)$ rejects $H$.  Now, it remains to show that ${\sf TEST}_{s,r}(H)$ accepts a random Hankel matrices with high probability.
 
To complete the proof we show that on input $H\in\mathbb{F}_2^{n\times n}$ that is a random Hankel matrix, ${\sf TEST}_{s,r}(H)$ rejects $H$  with probability $o(1)$. 
\begin{eqnarray}
  \Pr_{H}[{\sf TEST}_{s,r}(H) \text{~rejects~} H] &=  \Pr[\exists H' \exists S'~\spar(S')\leq s' \text{~s.t.~} \rk(H'-S')\leq r] \nonumber \\
&\leq \left(\frac{n}{2r}\right)^2 \binom{(2r)^2}{\leq s'} \Pr[\rk(H'-S')\leq r]  \label{eqn:eqn2}
\end{eqnarray}
Now, for a moment assume that  $\Pr[\rk(H'-S')\leq r]$ is quite low (i.e., $\Pr[\rk(H'-S')\leq r] \leq 2^{-n/16}$). Plugging this into Equation (\ref{eqn:eqn2}), we get: $$\Pr_{H}[{\sf TEST}_{s,r}(H) \text{~rejects~} H] \leq \left(\frac{n}{2r}\right)^2 \cdot \binom{(2r)^2}{\leq s'}\cdot 2^{-n/16}.$$ 
When $s=\frac{n^3}{160r^2\log n}$, $s'= \frac{n}{40 \log n}$ and as $\sqrt{n} \leq r \leq n/32$ $\Pr_{H}[{\sf TEST}_{s,r}(H) \text{~rejects~} H]$ is $o(1)$.  So condition (2) of the proof outline is satisfied by ${\sf TEST}_{s,r}$.

In the rest of this subsection we will show that for  $\Pr[\rk(H'-S')\leq r] \leq 2^{-n/16}$ when $H$ is carefully partitoned. There are several ways of partitioning $H$ into $2r\times 2r$ sub-matrices. For instance, one straight-forward way would be to tile of $H$ by matrices of dimension $2r\times 2r$ (see area shaded in solid grey in Figure \ref{fig:hankelpartition}). However, $H'$ has only $4r$ elements chosen independently and u.a.r. from $\mathbb{F}_2$. Since, we know that a random matrix in $\mathbb{F}^{n\times n}$ has high rank with high probability, intuitively we want $H'$ to see a large number of random bits so that $\rk(H'-S')$ is {\em low} with {\em low} probability. By using a cleverer partitioning of the Hankel matrix $H$, we can obtain sub-matrices $H'$ that see $\Theta(n)$ random bits and show that $\Pr[\rk(H'-S')\leq r] \leq 2^{-n/16}$.

Partition $H$ into $(n/2r)^{2}$ many sub-matrices of dimension $2r\times 2r$ each such that each sub-matrix $H'$ has $2r$ consecutive columns and $2r$ rows that are at a distance $n/2r$ apart as shown in
Figure \ref{fig:hankelpartition}. As $H$ is Hankel, every row in $H'$ sees $n/2r$ random bits and $H'$ sees $\Theta(n)$ random elements in $\mathbb{F}_2$. Such a matrix is said to be an {\em $n/2r$-Hankel} matrix.

\begin{figure}[H]
\centering
\includegraphics[scale=0.7]{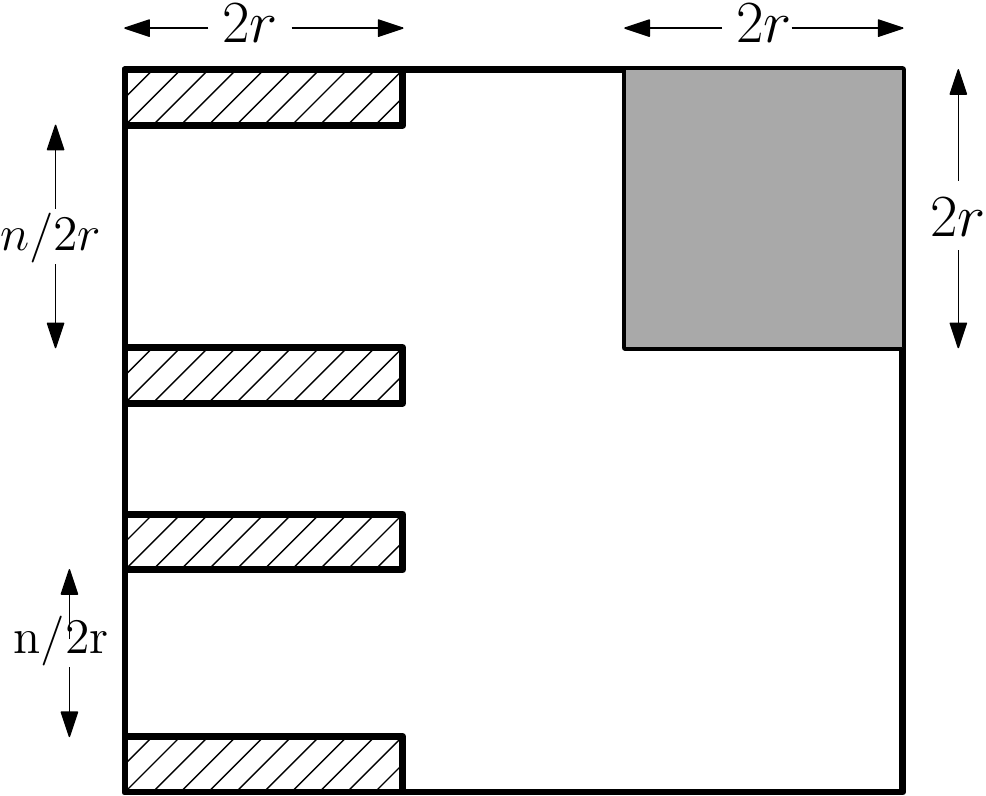}
\caption{Patition of $H$ into $2r\times 2r$ sub-matrices}
\label{fig:hankelpartition}
\end{figure}

Let $R_1,\ldots, R_{2r}$ be rows of $H'-S'$. If $\rk(H'-S')\leq r$, then there exists a basis $B=\{R_{i_1},\ldots,R_{i_r}\}$ such that any row in $H'-S'$ is spanned by a linear combination of the row-vectors in $B$. Let $J$ be the set of rows in $H'-S'$ that are not in $B$. In fact, if $\rk(H'-S')\leq r$, by a greedy procedure we can compute a basis $B$ such that for every row $R_j$ in $J$, $R_j\in span\{ R_k\mid R_k \in B, k<j \}$. Let $E$ denote the event that for every row $R_j$ in $J$, $R_j\in  span\{ R_k\mid R_k \in B, k<j \}$ (i.e., $R_j$ is in the linear span of rows above it). Then,
\begin{align*}
\Pr_{H}[\rk(H'-S')\leq r] &\leq \Pr[\exists B \text{~event~} E \text{~holds~}] \\
&\leq  \binom{2r}{\leq r} \Pr[\text{for a fixed $B$ event $E$ holds}] \\
&\leq 2^{2r}\cdot \Pr[\text{for a fixed $B$ event $E$ holds}].
\end{align*}
That is, for a fixed set $B$ of rows, we want to estimate the probability that every row not in $B$ is spanned by rows in $B$ occurring above it in the matrix $H'-S'$.
Let $J'$ be the set of rows not in the basis of $H'-S'$ that are sufficiently far apart. Let $J' = \{R_{j_1},\ldots,R_{j_t}\}\subseteq J$ be the set of rows such that the distance between $R_{j_p}$ and $R_{j_{p+1}}$ is at least $(2r)^2/n$ for all $p\in [t-1]$.  Note that $J' \geq |J|/(4r^2/n) = n/4r$. Now, if event $E$ holds, then  every row in $J'$ is spanned by $\{ R_k\mid R_k \in B, k<j \}$.  Let $E_{\ell}$ be the event that row $R_{j_{\ell}}$ in $J'$ is spanned by $\{ R_k\mid R_k \in B, k<j_{\ell} \}$. Now suppose for any $\ell \in [t]$, $\Pr[E_{\ell} \mid E_{1}, E_2,\ldots,E_{\ell-1} ] = 2^{-r}$, we get the following:

\begin{align*}
    \Pr_{H}[\rk(H'-S')\leq r] &\leq 2^{2r}\cdot \Pr[\text{for a fixed $B$, event $E_1\cap E_2 \cap \cdots \cap E_{t}$ holds}] \\
    &\leq 2^{2r}\cdot (\Pr[E_{\ell} \mid E_{1}, E_2,\ldots,E_{\ell-1} ])^{t} \\
    &\leq 2^{2r}\cdot 2^{-rt} \\
    &\leq 2^{-n/16}. 
\end{align*}
The last inequality follows as $r\leq n/32$ and $t\geq n/4r$. In the remaining part of this subsection, we show that for a fixed basis $B$, $\Pr[E_{\ell} \mid E_{1}, E_2,\ldots,E_{\ell-1} ] = 2^{-r}$ for any $\ell \in [t]$. For this, we will refer to the following figure:

\begin{figure}[H]
\centering
\label{fig:hankel-lincomb}
\includegraphics[scale=0.7]{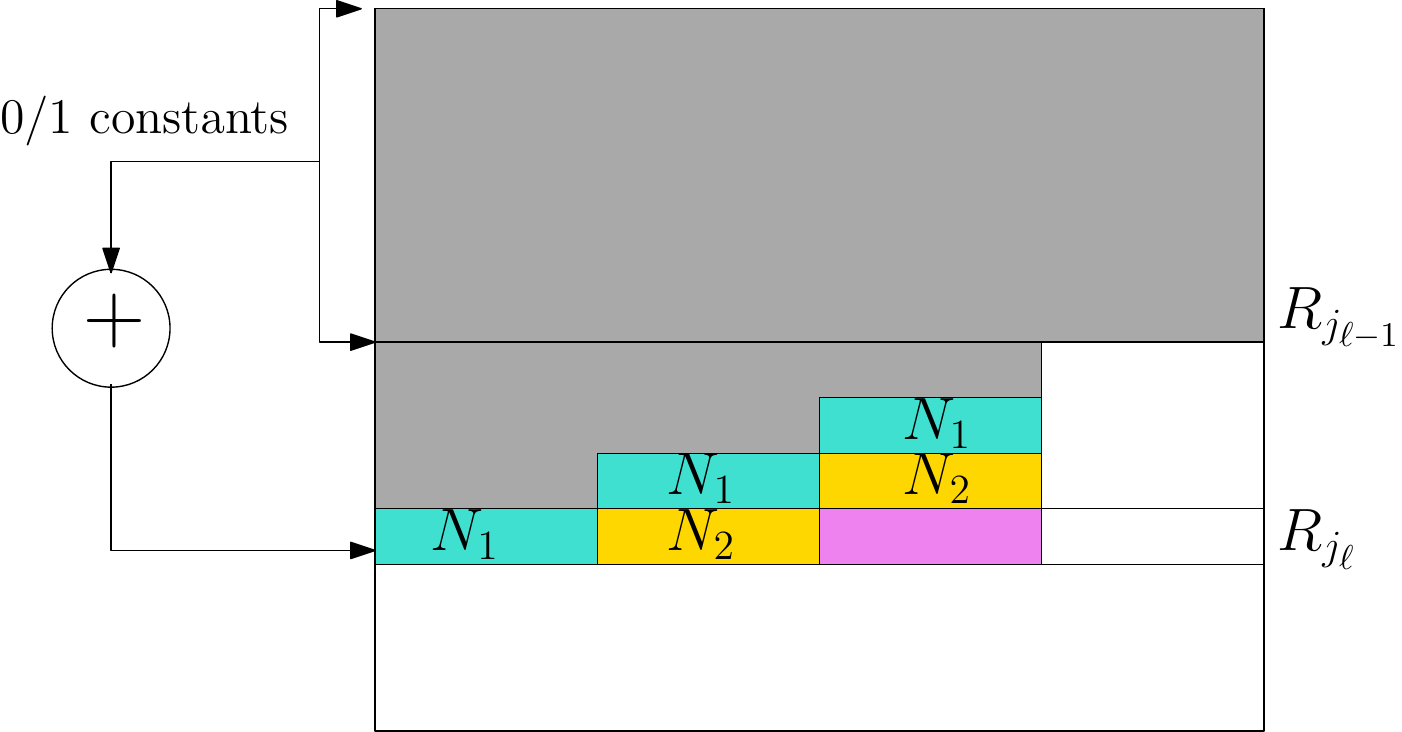}
\caption{Estimating $\Pr[E_{\ell} \mid E_{1}, E_2,\ldots,E_{\ell-1}]$}
\end{figure}

$\Pr[E_{\ell} \mid E_{1}, E_2,\ldots,E_{\ell-1}]$ is the probability that there exists a linear combination of the rows in $\{ R_k\mid R_k \in B, k<j_{\ell} \}$ such that $R_{\ell} = \sum_{k<j_{\ell}} \alpha_k R_k$. As the $\alpha_i's$ are from $\mathbb{F}_2$, there are at most $2^{|B|}\leq 2^r$ many linear combinations. Now, we need to estimate the probability that for a fixed linear combination of rows in $\{ R_k\mid R_k \in B, k<j_{\ell} \}$, $R_{\ell} = \sum_{k<j_{\ell}} \alpha_k R_k$. Once $\alpha_1,\ldots ,\alpha_{\ell}$ are fixed we can determine the elements in the block $N_1$. Block $N_1$ along with $\alpha_1,\ldots,\alpha_{\ell}$ completely determine $N_2$ . This way, once the linear combination is fixed, $R_{\ell}$ is a fixed row-vector in $\{0,1\}^{2r}$. Therefore, $\Pr[E_{\ell} \mid E_{1}, E_2,\ldots,E_{\ell-1}] \leq 2^r\cdot 2^{-2r} \leq 2^{-r}$ for any $\ell\in [t]$.

\begin{remark}
\label{rem:toeplitz}
For $r=o\left(\frac{n}{\log n\log \log n} \right)$, Theorem $\ref{thm:toeplitz}$ yields asymptotically better lower bound than the current best rigidity lower bound of $\Omega(\frac{n^2}{r}\log\frac{n}{r})$ for rank $r$. 
\end{remark}

\begin{remark}
A random $n\times n$ Toeplitz matrix can be constructed by using $2n$ random bits. Hence, Theorem $\ref{thm:toeplitz}$ gives an explicit construction of rigid matrices in the complexity class $\E^{\NP}$.
\end{remark}

%% file: kumar-volk.tex
Having constructed rigid matrices in the class $\E^{\NP}$, in this section we discuss the following result of \cite{KV19} which gives an explicit family of rigid matrices  constructible in sub-exponential time.

\begin{theorem}
\label{thm:rigid-subexp}
Let $\mathbb{F}_q$ be a finite field and $\mathbb{E}$ be an extension of $\mathbb{F}_q$ of degree at most $\exp(O(n^{1-1/2d}\log n))$. There exists a family of matrices $(A_n)_{n\in \mathbb{N}}$ constructible in time $\exp(n^{1-\Omega(1/d)})$ such that any linear circuit over $\overline{\mathbb{F}_q}$  of depth $d$ computing $A_n$ has size at least $\Omega(n^{1+1/2d})$. \end{theorem}

Here, $\overline{\mathbb{F}_q}$ denotes the algebraic closure of $\mathbb{F}_q$. The results in \cite{KV19} work for any field $\mathbb{F}$. In this article, we only consider the case when $\mathbb{F}$ is a finite field. \\

For any matrix $A\in \mathbb{F}^{n\times n}$, if $A = S +L$ where $\spar(S)\leq s$ and $\rk(L)\leq r$ then the linear transformation $x \mapsto A\cdot x$ can be computed by a linear circuit of depth 2 and size $2nr +s$. Hence, the following corollary of Theorem \ref{thm:rigid-subexp} which gives an explicit family of rigid matrices  constructible in sub-exponential time is not very difficult to observe.

\begin{corollary}
Let $\mathbb{F}$ be any field. There exists a family of matrices $(A_n)_{n\in \mathbb{N}}$ constructible in time $2^{o(n)}$ such that ${\cal R}_{A_n}(n^{0.5-\epsilon})=\Omega(n^{1.24})$.
\end{corollary}

Rest of this subsection is devoted to the proof of Theorem \ref{thm:rigid-subexp}. Following the standard template for proving arithmetic circuit lower bounds, the proof of Theorem \ref{thm:rigid-subexp}, proceeds by obtaining a {\em complexity measure} that is low for matrices computable by low-depth linear circuits of small size while obtaining explicit matrices for which the measure is large. Here, we use {\em Shoup-Smolensky dimension} of matrices as  a complexity measure.

The elements of any extension $\mathbb{E}$ of field $\mathbb{F}$ are univariate polynomials over $\mathbb{F}$ of appropriate degree and can be viewed as a vector of coefficients. Now, we formally define the Shoup-Smolensky dimension of a matrix:

\begin{definition}[Shoup-Smolensky dimension]
Let $\mathbb{F}$ be any field and $\mathbb{E}$ an extension of $\mathbb{F}$. Let
$M\in \mathbb{E}^{n\times n}$. For any $t\in \mathbb{N}$, $P_t(M) =\bigg\{\prod\limits_{(a,b)\in T} M_{ab}\mid T\in \binom{[n]\times [n]}{t}\bigg\} $ is the set of  products of $t$ distinct entries of $M$. The {\em Shoup-Smolensky dimension} of $M$ of order $t$ (denoted by $\ssd_t(M)$) is the dimension of the space spanned by the set $P_t(M)$ over $\mathbb{F}$.
\end{definition}

The Shoup-Smolensky dimension of $M$ of order $t$ is denoted by $\ssd_t(M)$ and is precisely $\dim_{\mathbb{F}}(span(P_t(M))$. First, we show that the Shoup-Smolensky dimension of matrices computable by small-size
and small-depth circuits is fairly low. 

\begin{lemma}
\label{lem:lin-ckt-ssd}
Let $M\in \mathbb{E}^{n\times n}$ be computable by linear circuit ${\cal C}$ of size $s$ and depth $d$. Then, for any $t\leq n^2/4$ such that $s\geq dt$, $\ssd_t(M)\leq {(e(2s/dt))}^{dt}$. 
\end{lemma}

\begin{proof}
Let matrix $M\in \mathbb{E}^{n\times n}$ be computable by linear circuit ${\cal C}$ of size $s$ and depth $d$ with layers $L_1,\ldots,L_{d+1}$. Then $M=P_1\cdots P_d$ where $P_i$ is the adjacency matrix of the graph ${\cal C}$ between layers $L_i$ and $L_{i+1}$. Then,

\begin{equation}
    M_{ij} = (P_1\cdots P_d)_{ij} 
    =\sum\limits_{k_1,\ldots,k_{d-1}} \left[(P_1)_{i,k_1}\cdot \prod\limits_{\ell=2}^{d-1}(P_\ell)_{k_\ell -1,k_\ell} \cdot (P_d)_{k_{d-1},j} \right] \label{eq:eq1}
\end{equation}

As ${\cal C}$ has size $s$, the total number of non-zero entries in all of $P_1,\ldots,P_d$ is at most $s$. From Equation (\ref{eq:eq1}), each entry of $M\in \mathbb{E}^{n\times n}$ is a sum of monomials of degree at most $d$ in the entries of matrices $P_1,\ldots,P_d$. Hence, every element of $P_t(M)$ is a sum of monomials of degree at most $dt$ in at most $s$ entries. Thus, 
\begin{align*}
    \ssd_t(M) &\leq \binom{s+dt}{dt}\\ 
    &\leq \left(\frac{e(s+dt)}{dt}\right)^{dt} \\
    &\leq  e^{dt} \left( 1+\frac{s}{dt} \right)^{dt} \leq {(e(2s/dt))}^{dt}
\end{align*}
as $s\geq dt$.
\end{proof}

Now, we want to construct $G\in \mathbb{E}^{n\times n}$ whose  Shoup-Smolensky dimension is large. For any $t\in \mathbb{N}$, clearly $\ssd_t(M)\leq \binom{n^2}{t}$ and we want $\ssd_t(G) \geq \binom{n^2}{t}$. Now, a simple way to achieve the maximum possible dimension for $\ssd_t(G)$ is to consider $G_{ij}=y^{e_{ij}}$ where the sum of any $t$ elements in $ \{e_{11},\ldots,e_{nn}\}\subseteq \mathbb{ N}$ of size $n^2$ is always distinct. As every element in $P_t(G)$ is the product of $t$ entries of $G$, $\ssd_t(G) \geq \binom{n^2}{t}$.

Recall that in the end we want to construct a family of  matrices in sub-exponential time. For this, we will require $G\in\mathbb{E}^{n\times n}$ to be constructed in sub-exponential time (i.e., time $n^{O(t)}$ which is sub-exponential when $t=n^{1-1/2d}$). This in turn implies that the entries of $G$ should be monomials of degree $n^{O(t)}$ and that every entry should be constructed in time $n^{O(t)}$. In summary, for any $t\in\mathbb{N}$, we require a set $S \subseteq \mathbb{N} $ satisfying the following conditions:
\begin{enumerate}
\item $|S|=n^2$ and every subset of $S$ of size $t$ has a distinct sum;
\item $S$ can be constructed in time $n^{O(t)}$; and
\item The maximum value of any element in $S$ is at most $n^{O(t)}$.
\end{enumerate}

To begin with, consider the following natural candidate set $S'=\{1,2,\ldots, 2^{n^2-1}\}$ for the set $S$. Clearly, $|S'|=n^2$ and every subset of $S'$ of size $t$ has a distinct sum but $S'$ does not satisfy condition $(3)$. A natural next step is to go modulo a prime $p$ so that set $S= \{ a \mod p \mid a\in S'\}$ satisfies conditions (1)-(3). To ensure (1), intuitively we want $p$ to be quite large. To ensure (2), we need $p$ to be not too large so that we can search for such a $p$ and construct $S$ in time $n^{O(t)}$.

In particular, we want a prime $p$ such that for any two sets $T,T'\subseteq S$ with $|T|=|T'|$, $\sigma_t = \sum\limits_{a\in T}a$ and $\sigma_{t'} = \sum\limits_{a\in {T'}}a$ are different. That is, $p$ does not divide $\prod\limits_{\substack{T,T'\subseteq S\\ |T|=|T'|}} (\sigma_t-\sigma_{t'})$ which is at most $(2^{n^2})^{n^{O(t)}}$ as every element of $S$ is at most $2^{n^2}$ and there are $n^{O(t)}$ subsets of $S$ of size $t$. Thus, by the {\em prime number theorem}, there are at most $\log ((2^{n^2})^{n^{O(t)}})$ distinct primes dividing $\prod\limits_{\substack{T,T'\subseteq S\\ |T|=|T'|}} (\sigma_t-\sigma_{t'})$. This proves the existence of such a prime $p$ and hence the existence of such a set $S \subseteq \mathbb{N}$ for any $t\in\mathbb{N}$ satisfying conditions (1)-(3). With  set $S$ in hand, we now complete the proof of Theorem \ref{thm:rigid-subexp}. \\

\noindent{\em Proof of Theorem $\ref{thm:rigid-subexp}$.}
Let $\mathbb{F}_q$ be any finite field.
Let $t=n^{1-1/2d}$ and $S=\{e_{11},\ldots,e_{nn}\}$ be the set constructed above. For the matrix
$G_n \in (\mathbb{F}[y])^{n\times n}$ given by $G_{ij}=y^{e_{ij}}$ where each $e_{ij}\in S$, we have $\ssd_t(G)\geq\binom{n^2}{t}$. So far, we have constructed a matrix $G_n\in (\mathbb{F}[y])^{n\times n}$. Since we want to obtain a matrix $G\in (\mathbb{E})^{n\times n}$ we need to project $y$ to some value preserving the Shoup-Smolensky dimension. For any $D$, an irreducible polynomial $g(z)$ of degree $D+1$ over $\mathbb{F}_q$ can be constructed in deterministic time $\poly(D,|\mathbb{F}_q|)$\cite{Sho88}. Let $\alpha$ be the root of $g(z)$ that is in $\mathbb{E} \triangleq \mathbb{F}_q[z]/ \langle g(z)\rangle$. Define $A_n \triangleq G_n\vert_{y=\alpha} $. Clearly, by the properties of the set $S$ constructed, any element of $P_t(A_n)$ is $\alpha^m$ where $m\leq t\cdot n^{O(t)}$ and every element of $P_t(A_n)$ corresponds to a distinct power of $\alpha$. Thus, by fixing $D=2\cdot t\cdot n^{O(t)}$, as $\{1,\alpha,\alpha^2,\ldots,\alpha^D\}$ are linearly independent over $\mathbb{F}_q$, $\ssd_t(A_n) = \ssd_t(G_n)\geq \binom{n^2}{t}$. Now, if $A_n$ is computable by a depth $d$ size $s$ linear circuit then by Lemma \ref{lem:lin-ckt-ssd},
\begin{align*}
   {(e(2s/dt))}^{dt}  \geq \binom{n^2}{t}
\end{align*}
If $s<n^{1+1/2d}/2$, the above equation contradicts binomial estimates. Hence $s=\Omega(n^{1+1/2d})$.

$\qedh$

%% file: rigid-pcp.tex
\label{subsec:pcp-constructions}

There have been recent constructions of semi-explicit rigid matrices based on a striking connection between rigid matrices and probabilistically checkable proofs. Informally, a {\em probabilistically checkable proof}(PCP) for a language $L$ is a proof or a certificate for membership of $x$ in $L$ such that by {\em probabilistically} querying very {\em few} locations of the proof, if $x\in L$ the verifier can always be convinced of this fact while if $x\not\in L$ then with {\em high} probability the verifier will reject the proof. For a more formal definition and a huge body of work revolving around PCPs see \cite{Pra07} and references therein.

Along these lines, there have been two results one due to Alman and Chen\cite{AC19} and the other by Bhangale et.\ al in \cite{BHPT20} both of which are geared towards constructing $\PCP$s with nice properties that aid the construction of semi-explicit rigid matrices. We begin by stating the following construction from \cite{AC19}:

\begin{theorem}
\label{thm:rigid-pcp1}
There exists a matrix $A \in \mathbb{F}_2^{n\times n}$ constructible in $\P^{\NP}$ such that there exists a $\delta>0$ for all$\epsilon>0$ with $R_{A_n}^{\mathbb{F}_q}(2^{(\log n)^{1/4-\epsilon}}) \geq \delta\cdot n^2$.
\end{theorem}

Very recently, Bhangale et.\ al in \cite{BHPT20} obtain the following strengthening of the parameters in the above theorem:

\begin{theorem}
\label{thm:rigid-pcp2}
There is a constant $\delta \in (0,1)$  such that there is an $\FNP$ machine  that for infinitely many $n$ on input $1^{n}$ 
outputs an  matrix $W_n$ in $\mathbb{F}_2^{n\times n}$ such that $R_{W_n}(2^{\log n/\Omega(\log \log n)}) \geq \delta\cdot n^2$. 
\end{theorem}

\begin{remark}
If you are wondering what the class $\FNP$ is, it is the function version of the class $\NP$. A relation $R(x,y)$ is in $\FNP$ if there exists a non-deterministic poly-time Turing machine $M$ that on input $x$ outputs $y$ such that $R(x, y) = 1$ or rejects when no such $y$ exists.
\end{remark}

Alman and Chen also provide a strengthening of the parameters (i.e., $R_{A_n}^{\mathbb{F}_q}(2^{(\log n)^{1-\epsilon}}) \geq \delta\cdot n^2$) in Theorem \ref{thm:rigid-pcp1} by assuming that $\NQP \subset \P/\poly$ and using the {\em easy witness lemma}. However, note that the statement of Theorem \ref{thm:rigid-pcp2} is an unconditional strengthening of Theorem \ref{thm:rigid-pcp1}. In this sub-section, we will see a proof sketch of Theorem \ref{thm:rigid-pcp2} and carefully delineate the connections between rigid matrices and probabilistically checkable proofs.

We begin with the three main hammers need to prove Theorem \ref{thm:rigid-pcp2}:
\begin{enumerate}
\item There exists a unary language $L\in \NTIME(2^n)\setminus \NTIME(2^n/n)$. This is essentially the non-deterministic time hierarchy theorem from \cite{Zak83}. \label{unary-ndth}
\item A faster algorithm to compute the sparsity of a given low-rank matrix $M$ developed in \cite{AC19}. Observe that given any $n\times n$ matrix $M$ the sparsity can be computed in time $n^2$. However if $M$ has {\em low} rank $r$ then it admits a product decomposition $M=A\cdot B$ with $A,B$ having dimensions $n\times r$ and $r\times n$ respectively. The faster algorithm in \cite{AC19} when given as input the matrices $A,B$ computes the sparsity of the matrix $M=A\cdot B$ in time $n^{2-\Omega(1/\log r)}$ for all $r=n^{o(1)}$. \label{lowrank-sparsity}
\item $\NTIME(2^n)$ has \PCP s with "nice" properties. See Theorem \ref{thm:pcp-nice} for actual statement. 
\label{ntime-pcp}
\end{enumerate} 
 
Before we sketch the proof of Theorem \ref{thm:rigid-pcp2}, we review the correspondence between \PCP s and constraint satisfaction problems(\CSP s). The $\PCP$ verifier $V$ can be viewed as an  instance $\Phi$ which is a set of functions  $(\phi_1,\phi_2,\ldots,\phi_m)$ on a set $V$ of $n$ variables. Each $\phi_i:\{0,1,\ldots,t-1\}^q\rightarrow \{0,1\}$ is a  constraint or clause whose arity is $q$. The probabilistically checkable proof $\pi$ is an assignment $\bar{a} \in \{0,\ldots,t-1\}^n$ to the $n$ variables. The query complexity of the verifier $V$ is the arity $q$ of the constraints. The randomness complexity of the verifier $V$ is log of the number of constraints (i.e., $\log m$).  We say that an assignment $\bar{a}$ satisfies constraint $\phi_i$ if $\phi_i(\bar{a})=1$. Let ${\sf val}(\Phi)$ denote the maximum over all assignments $\bar{a}$ of the fraction of clauses satisfied by $\bar{a}$ (i.e., $\max\limits_{\bar{a}}(\sum_{i=1}^m \phi_i(\bar{a}))/m$). The soundness error $s$ of the \PCP\ is ${\sf val}(\Phi)$.

As a first step, let us try to construct a high rank matrix using the three ingredients mentioned above. Note that constructing a high rank matrix is a trivial problem.  The goal however is to construct a matrix that has high rank even when a few entries are perturbed. The overall idea is to show that the length $2^n$ witnesses (viewed as a $2^{n/2} \times 2^{n/2}$ matrix) for the unary language $L\in \NTIME(2^n)\setminus \NTIME(2^n/n)$ cannot all be of low rank and there will exist high rank matrices infinitely often. As $\NTIME(2^n)$ has a \PCP\ with some {\em nice} properties there exists a verifier $V$ that randomly queries locations in the $N\times N$ witness matrix $W$ and run a decision predicate to decide if $1^n\in L$ or not. Here $N=2^{n/2}$. \footnote{We consider $N=2^{n/2}$ for simplicity of the argument. The square of the proof length is $2^n\poly(n)$.}

Let us assume for the moment that the decision predicate used by the verifier is the \CSP\ instance $\Phi$ consisting of the $N^2$ constraints $\phi_1=x_1, \phi_2=x_2,\cdots , \phi_{N^2}=x_{N^2}$. A proof is an assignment of 0's and 1's to the variables $x_1,x_2,\cdots, x_{N^2}$ which can be viewed as an $N\times N$ matrix $W$. As $L\in \NTIME(2^n)$, for every $1^n\in L$ there exists a witness $W_n$ that certifies the membership of $1^n $ in $L$. We begin with the following claim that asserts that every witness matrix $W_n$ cannot be of low rank and there will exists high rank matrices infinitely often.

\begin{claim}
\label{claim:1-lin}
The $N\times N$ witness matrix $W_n$ corresponding to every $1^n\in L$ cannot have low rank.
\end{claim}
\noindent{\em Proof of Claim \ref{claim:1-lin}.} Suppose not, the $N\times N$ witness matrix $W_n$ corresponding to every $1^n\in L$ is always of low rank(say rank $r$). Then $W_n=A\cdot B$ for some $A \in \{0,1\}^{N\times r},B\in \{0,1\}^{r\times N}$. Consider the following algorithm for $L$ in $NTIME(2^n/n)$:

\begin{algorithm}[H]
\caption{Algorithm for $L$ in $\NTIME(2^n/n)$.}
\label{algo:ntime-1lin}
  \SetKwInOut{Input}{Input}
  \SetKwInOut{Output}{Output}
  \Input{$1^n$}
  \Output{Decide if $1^n\in L$ or not}
  \BlankLine
Guess low-rank representation $A \in \{0,1\}^{N\times r}$ and $B\in \{0,1\}^{r\times N}$ for $W=A\cdot B$.
\BlankLine
Compute the $\spar(W)$ using $A,B$ and fast algorithm in (\ref{lowrank-sparsity}).
\BlankLine
Accept $1^n$ iff $\spar(W)>s\cdot N^2$. \label{algo:ntime-1lin-accept}
\end{algorithm}
  
  Now, we need to argue that the above algorithm correctly decides $L$ using the \PCP\ verifier and also investigate its running time. Note that by (\ref{ntime-pcp}), there is a \PCP\ verifier for $L$ with soundness error $s$ and decision predicate $\Phi$.(Although the $\PCP$ verifier has interesting properties we will not need them at the moment.) Recall that we assumed the clauses of $\Phi$ are just  variables. Hence, the number of clauses satisfied in $\Phi$ by any witness $W_n$ is exactly the $\spar(W_n)$.

\begin{align*}
1^n \in L &\Rightarrow \exists \text{ proof $\pi$ such that \PCP\ verifier accepts $\pi$ with prob. $>s$} \\
&\Rightarrow \exists \text{ proof $\pi$ such that fraction of clauses satisfied in $\Phi$ is $>s$} \\
&\Rightarrow \exists \text{ proof $\pi$ such that number of clauses satisfied in $\Phi$ is $>s\cdot N^2$} \\
 &\Rightarrow \exists W_n \in \{0,1\}^{N\times N} \text{ such that }\spar(W_n)>s\cdot N^2. \\
  &\Rightarrow \text{Algorithm \ref{algo:ntime-1lin} accepts in line \ref{algo:ntime-1lin-accept}.} 
  \end{align*}
\begin{align*}
1^n \not\in L &\Rightarrow \forall \text{ proofs $\pi$  \PCP\ verifier accepts $\pi$ with prob. $<s$}\\
&\Rightarrow \forall \text{ proofs $\pi$ the fraction of clauses satisfied in $\Phi$ is $<s$} \\
&\Rightarrow \forall \text{ proofs $\pi$ the number of clauses satisfied in $\Phi$ is $<s\cdot N^2$} \\
 &\Rightarrow \forall W_n \in \{0,1\}^{N\times N} \spar(W_n)<s\cdot N^2. \\
 &\Rightarrow \text{Algorithm \ref{algo:ntime-1lin} rejects in line \ref{algo:ntime-1lin-accept}.}
\end{align*}

We will use the assumption that  every witness $W_n$ corresponding to $1^n\in L$ has low rank to argue about the running time. By using non-determinism to guess the low rank matrices $A,B$ and by using the fast algorithm to computing the sparsity of a low rank matrix we can ensure that $L\in \NTIME(2^n/n)$ for a suitable choose of rank $r=2^{n/\Omega(\log n)}$ which is a contradiction to the fact that $L\in \NTIME(2^n)\setminus \NTIME(2^n/n)$.

$\qedh$

Note that one simplifying assumption is that every clause of the $q\CSP$ corresponding to the $\PCP$ is just a variable(also known as ${\sf MAX\mbox{-}1\mbox{-}LIN}$). Now, let us relax this assumption a bit by assuming that the $\CSP$ corresponding to the $\PCP$ verifier is a set of $M^2$ clauses on $N^2$ variables where each of the form $(x_a \oplus x_b)$. As there are $M^2$ clauses without loss of generality we can assume that every clause in $\Phi$ is indexed by two variables $i,j \in [M]$. Similarly as there are $N^2$ variables we assume that every variable is 
indexed by two variables $a_1,a_2 \in [N]$. Let $c_{ij}$ be a clause for some $i,j\in [M]$ then $c_{ij} = (x_{a_1,a_2} \oplus x_{b_1,b_2})$ where $a_1,a_2,b_1,b_2\in [N]$. Let us call such an instance where every clause satisfies the above property as a ${\sf MAX\mbox{-}2\mbox{-}LIN}$ instance. In the previous case when every clause was a variable, we had that the number of clauses in $\Phi$ satisfied by a witness $W$ is exactly the sparsity of the witness viewed as a matrix.  In the case when each clause is an $\sf XOR$ of two variables we need to use $W$ to relate the number of clauses in $\Phi$ satisfied by a witness $W$  and the sparsity of the a low rank matrix. For this purpose, we define two matrices $Q_1, Q_2\in \{0,1\}^{M\times M}$
by:
\begin{align}
Q_1[i,j] &= W[a_1,a_2] \label{eq:q1-matrix}\\
Q_2[i,j] &= W[b_1,b_2] \label{eq:q2-matrix}
\end{align}
where $c_{ij} = (x_{a_1,a_2} \oplus x_{b_1,b_2})$ is a clause in $\Phi$. That is, $Q_1[i,j]$ and $Q_2[i,j]$ contain the assignment (according to witness $W$) to the first and second variables of the clause $c_{ij}$. Now, it is easy to observe that the number of clauses in the ${\sf MAX\mbox{-}2\mbox{-}LIN}$ instance $\Phi$ satisfied by a witness $W$ is the sparsity of the matrix $(Q_1+Q_2) \mod 2$.

\begin{obs}
Let $L\in \NTIME(2^n)\setminus \NTIME(2^n/n)$. Assume the $N\times N$ witness matrix $W_n$ corresponding to every $1^n\in L$ has low rank. Let $Q_1,Q_2\in \{0,1\}^{M\times M}$ be matrices obtained from $W$ as given in Equations \ref{eq:q1-matrix} and \ref{eq:q2-matrix}. If $(Q_1+Q_2)$ has a low-rank representation then (by guessing the low-rank representation for $(Q_1+Q_2)$) we can follow the outline of the Algorithm \ref{algo:ntime-1lin} to show $L\in \NTIME(2^n/n)$  which is a contradiction. This implies that $W$ has high rank.
\end{obs}

The question that remains is that if $W$ has a low-rank representation then does $(Q_1+Q_2)$ have a low-rank representation? The answer to this question is {\em yes} if the ${\sf MAX\mbox{-}2\mbox{-}LIN}$ instance \CSP\ instance $\Phi$ mentioned above satisfies a specific property. 

Suppose there exists matrices $A_1,A_2,B_1,B_2$ such that
\begin{align}
Q_1=A_1\cdot W \cdot A_2 \label{eq:a1}\\
Q_2 = B_1\cdot W \cdot B_2 \label{eq:a2}
\end{align}
are satisfied. Now, observe that if $W=A\cdot B$ then 
\begin{align}
Q_1+Q_2 &= A_1\cdot W \cdot A_2 + B_1\cdot W \cdot B_2   
\nonumber \\
&= A_1\cdot (A\cdot B) \cdot A_2 + B_1\cdot (A\cdot B) \cdot B_2 \nonumber \\
&= \underbrace{\begin{bmatrix}
A_1 & B_1 \end{bmatrix} \cdot \begin{bmatrix}A & 0 \\
0 & A
\end{bmatrix}}_{\tilde{A}} \cdot
\underbrace{\begin{bmatrix}
0 & B \\
B & 0
\end{bmatrix} \cdot
\begin{bmatrix}
A_2 \\
B_2
\end{bmatrix}}_{\tilde{B}} \label{eq:q1-q2} \\
&= \tilde{A}\cdot \tilde{B} \nonumber
\end{align}

That is, if $W$ has a low(rank $r$) representation admitting a decomposition $W=A\cdot B$ then $Q_1+Q_2$  has a representation $\tilde{A}\cdot \tilde{B}$ with $\rk(\tilde{A}\cdot \tilde{B}) \leq 2r$. Now, by using the same algorithmic strategy as before we can construct an $\NTIME(2^n/n)$ algorithm for $L$:

\begin{algorithm}
\label{algo:ntime-2lin}
\caption{Algorithm for $L$ in $\NTIME(2^n/n)$.}
  \SetKwInOut{Input}{Input}
  \Input{$1^n$}
  \SetKwInOut{Output}{Output}
  \Output{Decide if $1^n\in L$ or not}
  \BlankLine
Guess low-rank representation $A \in \{0,1\}^{N\times r}$ and $B\in \{0,1\}^{r\times N}$ for $W=A\cdot B$. \label{algo:ntime-2lin-guess}
\BlankLine
Use \PCP\ verifier for $L$ to compute matrices $A_1, A_2, B_1, B_2$ of appropriate dimensions. \label{algo:ntime-2lin-compute}
\BlankLine
Using Equation (\ref{eq:q1-q2}) compute matrices $\tilde{A},\tilde{B}$.
\BlankLine
Calculate the sparsity of $\tilde{A}\cdot\tilde{B}$.
\BlankLine
Accept if and only if $\spar(\tilde{A}\cdot\tilde{B}) > s\cdot N^2$
 \label{algo:ntime-2lin-accept}
\end{algorithm}

By an argument similar to previous case, we can conclude that the above algorithm correctly decides $L$. However, we have the following few caveats. We will address them one by one.
\begin{enumerate}[leftmargin=0.5cm]
\item {\em How to compute matrices $A_1, A_2, B_1, B_2$ in time $2^{\gamma n}$ for some suitably chosen $\gamma >0$?} \\ 
\noindent Since $L\in \NTIME(2^n)$ there exists a \PCP\ verifier for $L$ which is used by \cite{BHPT20} to device a procedure that given a row-index $i$ of $A_1$ (respectively $ A_2, B_1, B_2$) computes the non-zero column entries of $i^{th}$ row in time $2^{\gamma n}$. Now, using the algorithm from \cite{AC19} to compute sparsity of $\tilde{A}\cdot \tilde{B}$ we can ensure that the above algorithm is in $\NTIME(2^n/n)$. 
\item {\em We have shown that if there exists matrices $A_1, A_2, B_1, B_2$ such that Equations \ref{eq:a1} and \ref{eq:a2} hold then the witness matrix $W$ must be of high rank infinitely often. But what does it mean to say that there exists matrices $A_1, A_2, B_1, B_2$ such that Equations \ref{eq:a1} and \ref{eq:a2} are satisfied? What structural requirement does this impose on the ${\sf MAX\mbox{-}2\mbox{-}LIN}$ instance  $\Phi$?} \\
It is not very difficult to note that the existence of $A_1, A_2, B_1, B_2$ such that Equations \ref{eq:a1} and \ref{eq:a2} hold is the same as placing the restriction that for any clause $c_{ij}=(x_{a_1(i,j),a_2(i,j)} \oplus x_{b_1(i,j),b_2(i,j)})$
\begin{align}
a_1(i,j) &= a_1(i) \text{ and } a_2(i,j) = a_2(j) \label{eq:rect1}  \\
b_1(i,j) &= b_1(i) \text{ and } b_2(i,j) = b_2(j) \label{eq:rect2}  
\end{align}
Any \CSP\ instance $\Phi$ satisfying Equations \ref{eq:rect1} and \ref{eq:rect2} is said to be {\em rectangular}. In fact, rectangularity can be extended to arbitrary q\CSP s. A q\CSP\ is {\em rectangular} if the $(i,j)^{th}$ constraint in the \CSP\ on $N^2$ variables and $M^2$ clauses involves $q$ variables $x_{t_1(i,j)},\ldots, x_{t_q(i,j)}$ then  for any $i\in [q]$the function $t_i:[M]\times [M] \rightarrow [N] \times [N]$ is a product of functions $a_i:[M]\rightarrow [N] $ and  $b_i:[M]\rightarrow [N] $. Furthermore, we say a \PCP\ is {\em rectangular} if the corresponding $\CSP$ is rectangular.

\item {\em Recall that we set out to prove that $W$ is a rigid matrix but we have shown that $W$ has high rank. What are the "nice" properties of the \PCP\ that enable us to ensure that even if a {\em few} entries of $W$ are changed the rank of matrix $W$ remains high?}\\
We have a unary language $L\in \NTIME(2^n)$ and a \PCP\  verifier that decides $L$ by using the corresponding \CSP\ instance $\Phi$. That is, if $1^n \in L$ then there is a witness $W_n$(assignment to the $N^2$ variables in $\Phi$) that satisfies at least $c$(say $75\%$) of the clauses in $\Phi$. Similarly, if $1^n \not\in L$ then for any witness $W_n$(assignment to the $N^2$ variables in $\Phi$) satisfies at most $s$(say $51\%$) of the clauses in $\Phi$. Note that $1< c <s <0$. Now, we want to show that $W$ when viewed as an $N\times N$ matrix has high rigidity. \\

\noindent{\em Proof Sketch of Theorem \ref{thm:rigid-pcp2}.} Let $L\in \NTIME(2^n)$ be a unary language such that $L\not\in \NTIME(2^n/n)$. For the sake of contradiction assume that for every $1^n \in L$ the witness $W$ is {\em close}(say $\delta$-close) to a a low -rank matrix $W'$. That is, by changing at most $2\delta$ entries in the matrix $W$ we can get the matrix $W'$. Now as $W'$ has low-rank we can follow the outline in Algorithm \ref{algo:ntime-2lin} by guessing the low-rank representation of $W'$(instead of guessing the low-rank representation of $W$) in line \ref{algo:ntime-2lin-guess}. In order to argue that the algorithm correctly decides $L$, we use the completeness and soundness error corresponding to the \PCP\ verifier for $L$. \\
If $1^n\in L$ then there is a witness $W_n$ that satisfies at least $c{-}2\delta$ fraction of the clauses in $\Phi$ and when $1^n \not\in L$ any witness $W_n$ satisfies at most $s$ fraction of the clauses in $\Phi$. By setting $\delta = (c-s)/3$, we get that when $1^n\in L$ then there is a witness $W_n$ that satisfies at least $ > s\cdot N^2$  clauses in $\Phi$ and when $1^n \not\in L$ any witness $W_n$ satisfies at most $s\cdot N^2$ clauses in $\Phi$. This gap in the number of clauses satisfied by $W_n$  can be used by the Algorithm \ref{algo:ntime-2lin} in line \ref{algo:ntime-2lin-accept} to distinguish between the {\sf YES} and {\sf NO} instances.

Since the \PCP\ verifier  randomly queries the locations in the proof it is possible that the verifier queries exactly the $2\delta$ locations in $W'$ in which the proofs $W$ and $W'$ differ. Note that if every proof location is equally likely to be queried by the \PCP\ verifier then the probability that the verifier queries exactly the {\em wrong} $2\delta$ locations in $W'$ is {\em small}. A \PCP\ whose verifier has such a property is said to be {\em smooth}.

For every step in the Algorithm \ref{algo:ntime-2lin} to yield the desired outcome observe that we have to prove the existence of {\em short, efficient, smooth, rectangular}\footnote{In \cite{BHPT20}, the authors actually prove the existence of short, efficient, smooth, almost-rectangular \PCP s with randomness-oblivious property.} \PCP s for $\NTIME(2^n)$. These are the "nice" properties that we expect the \PCP\ for for $\NTIME(2^n)$ to have. The existence of such \PCP s for $\NTIME(2^n)$ is the major contribution of \cite{BHPT20}. We state this formally in Theorem \ref{thm:pcp-nice} below without giving the proof.

From the above discussion, by choosing $r=2^{n/\Omega(\log n)}$ we can ensure that $L\in \NTIME(2^n/n)$ which is a contradiction. Hence, $W$ is rigid infinitely often and the algorithm runs in $\FNP$.

 $\qedh$

\begin{theorem}
\label{thm:pcp-nice}
Let $L$ be a language in $\NTIME(2^n)$. For every constants $s\in (0,1/2)$ and $\tau \in (0,1)$, there
exists a constant-query, smooth and $\tau$-almost rectangular \PCP\ for $L$ over the Boolean alphabet with soundness error $s$, proof length at most $2^n · poly(n)$ and verifier running time at most $2^{O(\tau n)}$.
\end{theorem}
The above theorem is a very informal statement of the \PCP\ construction in \cite{BHPT20}. Exact statement mentioning all the parameters of the \PCP\ can be found in Theorem 8.2 of \cite{BHPT20}. We do not include a proof of the above theorem here and refer the interested readers to Section 4 through 8 of \cite{BHPT20}. 

\item {\em Why should the predicate corresponding to the \PCP\ verifier be a ${\sf MAX\mbox{-}2\mbox{-}LIN}$  predicate?} \\
For example the \CSP\ could be {\sf MAXCUT}(which is also an example of a ${\sf MAX\mbox{-}2\mbox{-}LIN}$ instance) in directed graphs. For a discussion in the case when the decision predicate of the \PCP\ verifier(which is equivalent to \CSP\ instance) is {\sf MAXCUT} see Section 1.3 in \cite{BHPT20}. One interesting observation is that  Algorithm \ref{algo:ntime-2lin} for $L$ is in $\NTIME(2^n/n)$ even if predicate corresponding to the \PCP\ verifier is a ${\sf MAX\mbox{-}q\mbox{-}LIN}$  predicate where $q$ is a constant. This is because in the case of ${\sf MAX\mbox{-}q\mbox{-}LIN}$ every clause if the XOR of $q$ variables and similar to Equations \ref{eq:a1} and \ref{eq:a2} there exists $q$ matrices $Q_1,\ldots,Q_k$ such that $Q_1=A_1\cdot W \cdot A_2,  Q_2 = B_1\cdot W \cdot B_2 , Q_3 = C_1\cdot W \cdot C_2$ and so on till $Q_k$. As long as $q$ is a constant these matrices can be computed using the $\PCP$ verifier in line \ref{algo:ntime-2lin-compute} of Algorithm \ref{algo:ntime-2lin} and $L\in \NTIME(2^n/n)$. The argument then proceeds similar to ${\sf MAX\mbox{-}2\mbox{-}LIN}$.

Now, all we need to do to answer Question $4$ is to reduce from an arbitrary ${\sf MAX\mbox{-}q\mbox{-}CSP}$ to a ${\sf MAX\mbox{-}q'\mbox{-}LIN}$ for some constant $q'$ preserving the gap between $(c,s)$ where $c,s$ are the completeness and soundness guarantee. Note that we have a $\PCP$ verifier for $L$ whose predicate is a ${\sf MAX\mbox{-}q\mbox{-}CSP}$ instance that verifies the proof $\tilde{A}\cdot \tilde{B}$. 

Those familiar with Hastad's 3-bit $\PCP$ for $\NP$, recall that for every $\delta>0$ and any language $L\in \NP$ there exists a $\PCP$ verifier $V$ that reads 3 bits of proof $\pi$ and chooses locations $(i_1,i_2,i_3)$ and bit $b\in \{0,1\}$ according to some distribution and accepts iff $(\pi_{i_1}\oplus \pi_{i_2}\oplus \pi_{i_3}) = b $. Further, $V$ has completeness $1-\delta$ and soundness $1/2+\delta$.

Along similar lines, in our case $L\in \NEXP$ we want to compute acceptance probability of verifier for $\tilde{A}\cdot \tilde{B}$. In \cite{BHPT20}, the authors carefully design matrices $\tilde{A}_1,\ldots,\tilde{A}_{q'},\tilde{B}_1,\ldots,\tilde{B}_{q'}$ such that the acceptance probability of verifier $V$ for  $\tilde{A}\cdot \tilde{B}$ is at most the acceptance probability of verifier $V$ for $(\tilde{A}_1\cdot \tilde{B}_1)\oplus \cdots \oplus (\tilde{A}_{q'}\cdot \tilde{B}_{q'})$. Relating the acceptance probability of the ${\sf MAX\mbox{-}q\mbox{-}CSP}$ instance  to the acceptance probability of the ${\sf MAX\mbox{-}q'\mbox{-}LIN}$ instance requires Fourier analysis.

\end{enumerate}

%% file: elimination-rigid.tex
\subsection{Construction of rigid matrices: An algebraic geometry perspective}

In \cite{KLPS14}, the authors attempt to construct rigid matrices by using a approach based on {\em algebraic geometry} that we discuss in this subsection. The rigid matrices demonstrated in \cite{KLPS14} have the same shortcomings as that of \cite{Lok00,Lok06} in the sense that these matrices are not as explicit as we want them to be although their rigidity matches the upper bound in Lemma \ref{lem:rigid-upperbound}. However,  the  construction of rigid matrices based on ideas from elimination theory is quite insightful.

\begin{theorem}
\label{thm:rigid-alggeo}
Let $p_{11},\ldots,p_{nn}$ be $n^2$ distinct primes greater than $n^{4n^2}$ and $\zeta_{ij}$ be the primitive root of unity of order $p_{ij}$ (i.e., $\zeta_{ij}= e^{2\pi i/p_{ij}}$).  Let  $A \in \mathbb{K}^{n\times n}$ be the matrix given by $A[i,j]=\zeta_{ij}$ where $\mathbb{K}=\mathbb{Q}(\zeta_{ij},\ldots,\zeta_{ij})$. Then $R_A^{\mathbb{K}}(r)= (n-r)^2$.
\end{theorem}

\noindent {\em Proof Sketch of Theorem \ref{thm:rigid-alggeo}.} The proof involves the following observations as basic building blocks:

\begin{enumerate}
\item[(1)] The set of $n\times n$ matrices of rigidity at most $s$ for rank $r$ have dimension $n^2 - (n-r)^2 +s$ when viewed as an algebraic variety.
\item[(2)] By using (1), prove the existence of a non-zero polynomial $g$ of not-so-large degree in the {\em elimination ideals} associated with matrices with rigidity at most $s$.
\item[(3)] As the matrix $A$ has as entries primitive roots of unity of high order, $A$ cannot satisfy any polynomial $g$ with such a degree upper bound(i.e., $g(A)\neq 0$).
\end{enumerate}

Before we briefly describe each of the steps outlined 
above we will need the following notation:

\begin{notation}
\begin{itemize}
\item For any $n\times n$ matrix $A\in \mathbb{F}^{n\times n}$, ${\sf Supp}(A)$ denotes the positions $(i,j)$ in $A$ where there are non-zero entries.
\item Let {\em pattern} $\pi$ denote a subset of positions $\{(i,j) \mid i,j\in [n]\}$ in the matrix $A$. For any pattern $\pi$ let $S(\pi)$ be set of $n\times n$ matrices $A$ over $\mathbb{F}$ that are supported only on positions in $\pi$ (i.e., ${\sf Supp}(A)\subseteq \pi$). 
\item For a fixed pattern $\pi$ denote by ${\sf RIG}(n,r,\pi,\mathbb{F})$  the set of matrices in $\mathbb{F}^{n\times n}$ such that their rank can be reduced to $r$ by changing only the locations indexed by $\pi$. We will drop $\mathbb{F}$ when the field is clear from the context for ease of notation.
\end{itemize}
\end{notation}

\paragraph*{Step (1):}Let $\pi$ be a fixed pattern of size $s$. By definition of matrix rigidity, for every matrix $A\in {\sf RIG}(n,r,\pi)$ there exists a matrix $C_\pi\in \mathbb{F}^{n\times n}$ with ${\sf Supp}(C_\pi)\subseteq \pi$ and $\rk(A+C_\pi) =r$. The first observation to make is that both these conditions - ${\sf Supp}(C_\pi)\subseteq \pi$ and $\rk(A+C_\pi) =r$ can be expressed via polynomial equations(support can be expressed via simple linear equations  and rank being $r$ can be expressed by $(r+1)\times(r+1)$ minors of $A+C_\pi$ being 0). That is, ${\sf RIG}(n,r,\pi)$ is solution of a system of finitely many polynomial equations in variables $x_1,\ldots,x_{n^2},t_1,\ldots,t_s$. Hence, ${\sf RIG}(n,r,\pi)$ is an affine algebraic variety and so is ${\sf RIG}(n,r,\leq s) = \bigcup_{\pi : |\pi|=s}$ the set of $n\times n$ matrices of rigidity at most $s$ for rank $r$. Thus, it makes sense to talk about the dimension of ${\sf RIG}(n,r,\leq s)$ as an affine algebraic variety. Now, we analyse upper and lower bounds on the dimension of ${\sf RIG}(n,r,\leq s)$.

\paragraph*{Upper bound on $\dim({\sf RIG}(n,r,\leq s))$:}
Clearly, $\dim({\sf RIG}(n,r,\leq s))\leq n^2$. By the definition of rigidity, there is a natural map $\Phi$ from the  product of rank $r$ $n\times n$ matrices and $S(\pi)$ to ${\sf RIG}(n,r,\pi))$(i.e., $\Phi((A,C_\pi)) = A+C_\pi$).

As mentioned earlier the set of rank $r$ $n\times n$ matrices as well as $S(\pi)$ form an affine algebraic variety. Note that $\dim(S(\pi))=s$ for any pattern $\pi$ of size $s$. Also, by an argument similar to Lemma \ref{lem:rigid-upperbound} dimension of the variety corresponding to rank $r$ $n\times n$ matrices is $n^2-(n-r)^2$. Putting this all together, $\dim({\sf RIG}(n,r,\pi)) \leq n^2-(n-r)^2 +s$ since $\Phi$ is surjective. 
\paragraph*{Lower bound on $\dim({\sf RIG}(n,r,\leq s))$:}
First, let us try to understand {\em elimination ideals} associated with matrices of low rigidity. For any pattern $\pi$ with $|\pi|=s$, let $T_{\pi}$ denote the $n\times n$ matrix with variables $y_1,\ldots,y_s$ as entries in the $s$ positions indexed by $\pi$. It is clear  hat for any $n\times n$ matrix $X$  with entries $x_{1},\ldots,x_{n^2}$, the fact that $\rk(X+T_{\pi})=r$ is the same as saying that all $(r+1)\times (r+1)$ minors of the matrix $X+T_{\pi}$ vanish. Then denoting by $I(n,r,\pi)$ the ideal generated by the $(r+1)\times (r+1)$ minors of the matrix $X+T_{\pi}$, we get that $I(n,r,\pi)\subseteq \mathbb{F}[x_1,\ldots,x_{n^2},y_1,\ldots,y_s]$. It is not difficult to observe that ${\sf RIG}(n,r,\pi) = \psi(\mathbb{V}(I(n,r,\pi))$ where $\psi$ is a projection map representing the projection of $s$ variables $y_1,\ldots,y_s$. Let us define the {\em elimination ideal} $EI(n,r,\pi)$ as  the ideal $EI(n,r,\pi) \triangleq I(n,r,\pi) \cap \mathbb{F}[x_1,\ldots,x_{n^2}]$. By Closure Theorem of elimination theory[], $\psi(\mathbb{V}(I(n,r,\pi)) = \mathbb{V}(EI(n,r,\pi))$. Hence, $\dim({\sf RIG}(n,r,\pi)) =\dim(\mathbb{V}(EI(n,r,\pi)))$ and $\dim({\sf RIG}(n,r,\leq s)) = \max_{k\leq s,\pi}\dim(\mathbb{V}(EI(n,r,\pi)))$. The authors in \cite{KLPS14} demonstrate a pattern $\pi$ of size $k\leq s$ such that $\dim(\mathbb{V}(EI(n,r,\pi))) \leq n^2-(n-r)^2 +s$ thhus obtaining a lower bound on $\dim({\sf RIG}(n,r,\leq s))$.

\paragraph*{Step (2):}From Step 1, proving that a matrix $A$ has rigidity $(n-r)^2$ for rank $r$ is the same as showing that $A\not \in {\sf RIG}(n,r,\leq (n-r)^2-1)$. This is in essence the same as proving that $A\not\in {\sf RIG}(n,r,\pi)$ for any pattern $\pi$ with $|\pi|= (n-r)^2-1$. Given that ${\sf RIG}(n,r,\pi) = \mathbb{V}(EI(n,r,\pi))$, we want to show $A \not\in \mathbb{V}(EI(n,r,\pi))$ for any pattern $\pi$ with $|\pi|= (n-r)^2-1$. In other words, our goal is to the existence of a non-zero fairly-low degree polynomial $g\in EI(n,r,\pi)$ such that $g(A)\neq 0$. But what if $EI(n,r,\pi) =\langle 0\rangle$? To rule this out, observe that $\dim(\mathbb{V}(EI(n,r,\pi))) < n^2$ for any  pattern $\pi$ with $|\pi|= (n-r)^2-1$. Hence $EI(n,r,\pi) \neq (0)$ by Hilbert's Nullstellensatz.

In particular, the authors in \cite{KLPS14} use the effective Nullstellensatz theorem  of \cite{DFGS91} which is as follows: 
 
\begin{theorem}
 Let $Z = \{  z_1,\ldots,z_m\}$ and $I = \langle F_1,\ldots,f_p \rangle \subseteq\mathbb{F}[Z]$ such that the maximum degree of any of the $f_i$'s is $d$. Let $Z'$ be a subset of $\ell$ $Z$ variables. If $I\cap \mathbb{F}[Z'] \neq \langle 0 \rangle$ then there exists a non-zero polynomial $g\in I\cap \mathbb{F}[Z']$ such that $g=\sum_{i\in [p]} f_i g_i$ where $g_i \in \mathbb{F}[Z']$ and $\deg(g_i f_i) \leq d^p(d^p+1)$.
 \end{theorem} 
 
In our setting, $Z =\{x_1,\ldots,x_{n^2},y_1,\ldots,y_s \} , Z' =\{x_1,\ldots,x_{n^2}\}, I = I(n,r,\pi)$ and $d\leq r+1$ for sufficiently large $n,r$. Then there exists a polynomial $g$ in $I(n,r,\pi)\cap \mathbb{F}[x_1,\ldots,x_{n^2}]$ of degree less than $n^{4n^2}$
where $\pi$ is a pattern of size $(n-r)^2-1$. This shows that there is a polynomial $g\in EI(n,r,\pi)$ of degree $<n^{4n^2}$.

\paragraph*{Step (3):} Let  $A \in \mathbb{K}^{n\times n}$ be the matrix in the statement of the theorem, $A[i,j]=\zeta_{ij}$ where $\mathbb{K}=\mathbb{Q}(\zeta_{ij},\ldots,\zeta_{ij})$. It is not very difficult to $g(A)\neq 0$ as the entries of the matrix $A$ are algebraic.

This completes the proof of Theorem \ref{thm:rigid-alggeo}.

\begin{remark}
Along the lines of analysing the degree of the polynomial $g$ in the above result, Kumar and Volk in \cite{KV19} reduce the upper bound on the degree of such a polynomial to $\poly(n)$. That is, there is a polynomial $P$ on $n^2$ variables of degree at most $\poly(n)$ such that any matrix $M\in\mathbb{F}^{n\times n}$ with $R_{M}(n/100)\leq n^2/100$ satisfies $P(M) = 0$.
\end{remark}

%% file: walsh-hadamard.tex
To begin with, let us focus on the upper bounds on the rigidity of $2^n\times 2^n$ {\em Walsh-Hadamard matrix} proved by Alman and Williams in \cite{AW17}.

\subsection{Non-rigidity of Walsh-Hadamard matrices}
The {\em Walsh-Hadamard matrix} $H_n$ is a $2^n\times 2^n$ matrix whose rows and columns are indexed by vectors in $\{0,1\}^n$ (in lexicographic order).  The entries of $H_n$ are given by $H_n[x,y] = (-1)^{\inner{x}{y}}$ for any $x,y\in \{0,1\}^n$ where $\inner{x}{y}$ denotes the inner product of vectors $x$ and $y$.

Alman and Williams in \cite{AW17} prove the following upper bounds on the rigidity of {\em Walsh-Hadamard matrix}:

\begin{theorem}
\label{thm:walsh-hadamard}
Let $\mathbb{F}$ be any field. For every $\epsilon\in (0,1/2)$, ${\cal R}_{H_n}(2^{n(1-f(\epsilon))})  \leq 2^{n(1+\epsilon)}$ where $f(\epsilon)=\Theta(\epsilon\log(1/\epsilon))$.
\end{theorem}

\noindent{\em Proof Sketch of Theorem \ref{thm:walsh-hadamard}.} The idea behind the proof is to approximate $H_n$ by a sparse polynomial(say $M'$) so that we can obtain a trivial upper bound on the rank of $M'$ and $H_n$ is close to $M'$ implying that rigidity of $H_n$ is low.

In order to provide more clarity, we delineate the approach to prove Theorem \ref{thm:walsh-hadamard} which has two broad steps:
\begin{enumerate}
    \item Approximate $H_n$ by the truth table matrix of a {\em sparse} polynomial. That is, construct a sparse polynomial $p:\{0,1\}^{2n}\rightarrow \mathbb{R}$ such that the $2^n\times 2^n$ matrix $M_p$ given by $M_p[x,y]\triangleq p(x+y)$ (for all $x,y\in \{0,1\}^n$) agrees with $H_n$ on {\em most} entries. If $p$ has sparsity $2^{n-\Omega(\epsilon^2 n)}$ then $\rk(M_p) \leq 2^{n-\Omega(\epsilon^2 n)}$.  Although this way $M_p$ has low rank, $H_n$ does not agree with $M_p$ on all entries. The construction of $p$ is only such that $H_n$ agrees with $M_p$ on those where $\inner{x}{y}\in [2\epsilon n, (1/2+\epsilon)n]$ for some $\epsilon\in (0,1/2)$. Thus, $\spar(H_n-M_p)$ could be large (which we tackle in Step 2). 
    \item Construct a matrix $M'$ from $M_p$ that agrees with $H_n$ on far more entries than that of $M_p$ but has rank comparable to that of $M_p$. Obtain $M'$ from $M_p$ such that $M'[x,y]=H_n[x,y]$ whenever:
\begin{enumerate}
    \item[(i)]  $M'[x,y]=M_p[x,y]$ (i.e., $\inner{x}{y}\in [2\epsilon n, (1/2+\epsilon)n]$); or
    \item[(ii)]  one of $x$ or $y$ has a large fraction of $1$'s (i.e., when $|x|\not\in [(1/2-\epsilon)n,(1/2+\epsilon)n]$  or $|y|\not \in [(1/2-\epsilon)n,(1/2+\epsilon)n]$).
\end{enumerate}
Now, $M'$ disagrees with $H_n$ only on entries indexed by elements in
\begin{align*}
D=\{  (x,y)\in \{0,1\}^n\times \{0,1\}^n &\mid  |x|\in [(1/2-\epsilon)n,(1/2+\epsilon)n],   \\
&  |y|\in [(1/2-\epsilon)n,(1/2+\epsilon)n], \\
& \inner{x}{y}\not\in [2\epsilon n, (1/2+\epsilon)n].\}
\end{align*} 
\end{enumerate}
To complete the proof of Theorem \ref{thm:walsh-hadamard}, we estimate the size of $D$. Observe that for any $(x,y)$ in $\{0,1\}^n\times \{0,1\}^n$ such that $|x|\in [(1/2-\epsilon)n,(1/2+\epsilon)n],|y|\in [(1/2-\epsilon)n,(1/2+\epsilon)n]$, the inner product $\inner{x}{y}$ has value at most $n(1/2+\epsilon)$. In order to estimate $|D|$, it suffices to count for a given $x\in \{0,1\}^n$ with $|x|\in [(1/2-\epsilon)n,(1/2+\epsilon)n]$, the number of vectors $y\in \{0,1\}^n$ with $|y|\in [(1/2-\epsilon)n,(1/2+\epsilon)n]$ and $\inner{x}{y}< 2\epsilon n$ which is given by $\sum\limits_{i=(1/2-\epsilon)n}^{(1/2+\epsilon)n}\sum\limits_{j=0}^{2\epsilon n} \binom{|x|}{i}\binom{n-|x|}{i-j} $. Using the fact that $\epsilon\in (0,1/2)$ and $(1/2-\epsilon)n \leq |x|,|y|\leq (1/2+\epsilon)n$, $|D|\leq 2^{O(n\epsilon\log(1/\epsilon))}$.

$\qedh$

In the rest of this section we discuss in detail the steps outlined above. To discuss Step 1, consider the following lemma that uses multivariate polynomial interpolation over integers to construct a sparse polynomial agreeing with the matrix $H_n$ on many entries.

\begin{lemma}
\label{lem:poly-interpolate}
Let $\mathbb{F}$ be a field and $\epsilon\in (0, 1/2)$. There exists a $2n$-variate multilinear polynomial $p(\bar{x},\bar{y})$ of sparsity at most $2^{n-\Omega(\epsilon^2 n)}$ such that for any $x,y\in \{0,1\}^n$ with $\inner{x}{y}\in [2\epsilon n, (1/2+\epsilon)n]$,
$$p(\bar{x},\bar{y}) = H_n[\bar{x},\bar{y}] = (-1)^{\inner{\bar{x}}{\bar{y}}}.$$
\end{lemma}

\begin{proof}
First, we construct an $n$-variate polynomial $q(z_1,\ldots,z_n)$ with integer coefficients such that for any $\bar{a} \in \{0,1\}^n$ with $|\bar{a} | \in [2\epsilon n, (1/2+\epsilon)n]$, $q(\bar{a} )= (-1)^{|\bar{a} |}$. 
\begin{claim}
\label{claim:hadamard-polynomial}
Given integers $c_1,c_2,\ldots,c_r$, there exists an $n$-variate polynomial $q(z_1,\ldots,z_n)$ of degree $r-1$ that agrees with $c_1,\ldots,c_r$ on boolean inputs of hamming weight $k+1,\ldots,k+r$ for any $k\leq n-r$. That is, $q(\bar{a} )=c_i$ for any $\bar{a} \in \{0,1\}^n$ such that $|\bar{a} |=k+i$.  
\end{claim}
\noindent\textit{Proof of Claim $\ref{claim:hadamard-polynomial}$.} Let us consider the most natural construction of such an $n$-variate polynomial $q:\{0,1\}^n \rightarrow \mathbb{Z}$ of degree $r-1$. Then,
\begin{equation}
   q(z_1,\ldots,z_n)=\sum\limits_{i=1}^{r-1}b_i\sum\limits_{\substack{\alpha\in \{0,1\}^n \\ |\alpha|=i}}\prod\limits_{j=1}^n x_j^{\alpha_j},
\end{equation}
where $b_0,\ldots,b_{r-1}$ are in $\mathbb{Z}$. Observe that for any $\bar{a} \in \{0,1\}^n$, $q(\bar{a})= \sum\limits_{i=1}^{r-1}b_i \binom{|\bar{a}|}{i}$. Whenever $|\bar{a}| = k+i$, we want $q(\bar{a})=c_i$. That is, for every $i\in [r-1]$, we need
\begin{equation*}
    b_0\binom{k+i}{0} +  b_1\binom{k+i}{1} +\cdots +  b_{r-1}\binom{k+i}{r-1} = c_i
\end{equation*}
This gives us the following matrix equation:
\begin{equation*}
    \begin{bmatrix}
 \binom{k+1}{0} & \binom{k+1}{1} & \cdots & \binom{k+1}{r-1} \\
 \binom{k+2}{0} & \binom{k+2}{1} & \cdots & \binom{k+2}{r-1}
 \\
 \vdots & \vdots & \vdots &\vdots \\
 \binom{k+r}{0} & \binom{k+r}{1} & \cdots & \binom{k+r}{r-1}
 \end{bmatrix}\cdot \begin{bmatrix}
b_0 \\
b_1 \\
\vdots \\
b_{r-1}
 \end{bmatrix} = \begin{bmatrix}
c_1 \\
c_2 \\
\vdots \\
c_{r}
 \end{bmatrix}
\end{equation*}
It is not difficult to note that the $r\times r$ matrix in the above  equation with binomial coefficients as entries is invertible. Hence there exists a vector $b=(b_0 b_1 \cdots b_{r-1})$ satisfying the above matrix equation which in turn completes the proof of Claim \ref{claim:hadamard-polynomial}.

$\qedh$

Now, we proceed with the proof of Lemma \ref{lem:poly-interpolate}.
Observe that by setting  $r=(1/2-\epsilon)n+1, k=2\epsilon n -1$ and $c_i=(-1)^{k+i}$( $k+r = (1/2+\epsilon)n$) in Claim \ref{claim:hadamard-polynomial}, we get the polynomial $q(z_1,\ldots,z_n)$ over $\mathbb{Z}$ of degree $(1/2-\epsilon)n$ satisfying the required properties. (By taking coefficients of $q$ modulo an appropriate $m$, we can get a polynomial $q$ over a field $\mathbb{F}$ satisfying the required properties.) In the remaining part of the proof, we show how to use the above lemma to obtain the $2n$-variate polynomial $p(x_1,\ldots,x_n,y_1,\ldots,y_n)$ as required. A semi-natural candidate for $p(x_1,\ldots,x_n,y_1,\ldots,y_n) \triangleq q(x_1y_1,\ldots,x_n y_n) $. For any $x,y \in \{0,1\}^n$ such that $\inner{x}{y}\in [2\epsilon n , (1/2+\epsilon)n]$ we immediately have $p(\bar{x},\bar{y}) = q(\bar{z})$ where $|\bar{z}|=\inner{\bar{x}}{\bar{y}}$. Hence, by construction of $q$ in Claim \ref{claim:hadamard-polynomial}, $p(\bar{x},\bar{y}) = (-1)^{\inner{x}{y}}$. Since we are only interested in $x,y$ in $\{0,1\}^n$, we can make $p$ multilinear by setting $x_i^2=x_i,y_i^2=y_i$ for all $i\in [n]$. In $p$ the variables $x_i$ and $y_i$ are tied together whenever they occur. Thus, we can view $p(x_1,\ldots,x_n,y_1,\ldots,y_n)$ as a multilinear polynomial of degree $(1/2-\epsilon)n$ on $n$ variables $w_1,\ldots,w_n$ where each $w_i=x_iy_i$ and $\spar(p) = \binom{n}{(1/2-\epsilon)n+1} \leq 2^{n-\Omega(\epsilon^2 n)}$.   
\end{proof}

Now, we move on to Step 2 of the proof of Theorem \ref{thm:walsh-hadamard} outlined above. To obtain $M'$ from $M_p$, correct those rows in $M_p$ indexed by $\{x\in \{0,1\}^n \mid |x|\not\in [(1/2-\epsilon)n,(1/2+\epsilon)n]$
and those columns in $M_p$ indexed by 
$\{y\in \{0,1\}^n \mid |y|\not\in [(1/2-\epsilon)n,(1/2+\epsilon)n]$. Since we want $\rk(M')$ to be comparable to $\rk(M_p)$ the idea is to construct  a matrix $M'$ such that $M' \triangleq M_p - (M_1 +\cdots +M_t)$ where $\rk(M_i)=1$ for all $i\in [t]$ and  $t=2^{n-\Omega(\epsilon^2 n)}$. 
For every row $r$ indexed by $x\in \{0,1\}^n$ with $|x|\not \in [(1/2-\epsilon)n,(1/2+\epsilon)n]$ let matrix $M_r\in\mathbb{F}^{2^n\times 2^n}$ be given by $M_r[x,y]=H_n[x,y]$ for all $y\in \{0,1\}^n$ and every other row of $M_r$ be zero. Similarly, for every column $c$ indexed by $y\in \{0,1\}^n$ with $|y|\not \in [(1/2-\epsilon)n,(1/2+\epsilon)n]$, let matrix $M_c\in\mathbb{F}^{2^n\times 2^n}$ be given by $M_c[x,y]=H_n[x,y]$ for all $x\in \{0,1\}^n$ and every other column of $M_c$ be zero. Observe that each such matrix $M_r$(resp., $M_c$) has rank 1. Note that
\begin{align*}
t &= 2\cdot |\{ v\in \{0,1\}^n \mid |v|\not\in [(1/2-\epsilon)n,(1/2+\epsilon)n] \}| \\ 
&= 2 \cdot \Bigg[\sum\limits_{i=0}^{(1/2-\epsilon)n}\binom{n}{i} + \sum\limits_{i=(1/2+\epsilon)n}^{n}\binom{n}{i}\Bigg] =4\cdot \sum\limits_{i=0}^{(1/2-\epsilon)n}\binom{n}{i} \leq 4\cdot n \cdot 2^{n-\Omega(\epsilon^2 n)}
\end{align*}

Therefore, $\rk(M')\leq \rk(M_p) + 4\cdot n \cdot 2^{n-\Omega(\epsilon^2 n)} \leq  5\cdot n \cdot 2^{n-\Omega(\epsilon^2 n)} $ as $\rk(M_p) \leq 2^{n-\Omega(\epsilon^2 n)}$ from Step 1. Further, on every row $M'$ differs from $H_n$ on at most $|D| \leq 2^{n\epsilon\log(1/\epsilon)}$ entries. Hence, ${\cal R}_{H_n}(2^{n(1-f(\epsilon))})  \leq 2^{n(1+\epsilon)}$. \\

$\qedh$

%% file: function-matrix.tex
Now, we review a recent result of Dvir and Edelman \cite{DE17} on the non-rigidity of certain matrices(based on functions over finite fields) using the {\em Croot-Lev-Pach Lemma}.

\subsection{Non-Rigidity of Function Matrices}
Let $\mathbb{F}_q$ be any finite field. For any function $f:\mathbb{F}_q^n\rightarrow \mathbb{F}_q$ let $M_f$ be the $q^n\times q^n$ matrix given by $M_f[I,J]=f(I+J)$ for any $I,J$ in $\mathbb{F}_q^n$. In the following subsection we discuss a result from \cite{DE17} proving an upper bound on the rigidity of the {\em function matrix} $M_f$\footnote{The term {\em function matrix} used here is non-standard terminology and is used here to denote that there is a specific function associated with these matrices.}.

\begin{theorem}
\label{thm:function-matrix}
Let $f:\mathbb{F}_q^n\rightarrow \mathbb{F}_q$ be any function. For any $\epsilon>0$ and $n$ sufficiently large, there exists an $\epsilon ' >0$ such that ${\cal R}_{M_f}(q^{n(1-\epsilon')})\leq q^{n(1+\epsilon)}$. 
\end{theorem}

The above theorem says that for any function $f:\mathbb{F}_q^n\rightarrow \mathbb{F}_q$ and any $\epsilon >0$, the matrix $M_f$ has rigidity at most $q^{n(1+\epsilon)}$ for rank $q^{n(1-\epsilon')}$ where the rank is over $\mathbb{F}_q$.

\noindent\textit{Proof Sketch of Theorem \ref{thm:function-matrix}.} The proof is extremely elegant and involves the following two steps:
\begin{enumerate}
    \item {\em Approximate} $f:\mathbb{F}_q^n\rightarrow \mathbb{F}_q$ by a polynomial $p:\mathbb{F}_q^n\rightarrow \mathbb{F}_q$ of {\em low} degree ($d=(1-\delta)n(q-1)$) for some $\delta>0$. By approximating function $f$ by polynomial $p$, we mean $|\{ x\in \mathbb{F}_q^n \mid p(x)\neq f(x) \}| \leq q^{n\epsilon}$.
    \item Show that for any polynomial $p:\mathbb{F}_q^n\rightarrow \mathbb{F}_q$ of sufficiently {\em low} degree($d$ being $(1-\delta)n(q-1)$), $\rk(M_p)\leq q^{n(1-\epsilon')}$ for some $\epsilon'>0$ depending on $\delta$ and $\epsilon$.
\end{enumerate}
From Steps $1$ and $2$, we can infer that $M_f= S+L$ where 
$S=M_f-M_p$ and $L=M_p$. From Step 1 function $f$ and polynomial $p$ differ on at most $q^{n\epsilon}$ many inputs implying that $S$ has at most $q^{n\epsilon}$ non-zero entries in every row and column. Hence, $\spar(S)\leq q^{n(1+\epsilon)}$. From Step 2, $\rk(L)\leq q^{n(1-\epsilon')}$ for some $\epsilon'>0$. Thus , ${\cal R}_{M_f}(q^{n(1-\epsilon')})\leq q^{n(1+\epsilon)}$. $\qedh$

\paragraph{}
We now delve into the details of Steps $1$ and $2$. The set of all functions $\{f \mid f:\mathbb{F}_q^n\rightarrow \mathbb{F}_q\}$ denoted by $F(q,n)$ is a vector space of dimension $q^n$ with basis $\{ x_1^{a_1}x_2^{a_2}\cdots x_n^{a_n}\mid 0\leq a_i\leq q-1 \}$. Let $F_d(q,n)$ be the set of polynomials of degree $d$ in $\mathbb{F}_q[x_1,\ldots,x_n]$. $F_d(q,n)$ is a subspace of $F(q,n)$ with basis $\{ x_1^{a_1}x_2^{a_2}\cdots x_n^{a_n} \mid 0\leq a_i\leq q-1,  \sum_{i} a_i=d \}$. Any function $f:\mathbb{F}_q^n\rightarrow \mathbb{F}_q$ can be viewed as a vector $v$ in $F(q,n)$.

To begin with, we show that for any $d\leq n$ by changing the vector $v$ in $F(q,n)$ (corresponding to the function $f$) on $\dim(F(q,n))-\dim(F_d(q,n))$ many coordinates, we can obtain a vector $u$ in $F_d(q,n)$ (corresponding to a polynomial $p$ of degree $d$). To complete the proof of Step $1$, we obtain a lower bound of $q^n-q^{n\epsilon}$ on $\dim(F_d(q,n))$ when $d=(1-\delta)n(q-1)$.

Let $m\triangleq\dim(F(q,n))$ (i.e., $m=q^n$) and $r\triangleq\dim(F_d(q,n))$. As $F_d(q,n)$ is a subspace of $F(q,n)$, there is an $m\times r$ matrix $M$ of rank $r$ such that $F_d(q,n)$ is the image of the linear transformation defined by $M$. The aim here is to construct for any vector $v\in F(q,n)$, a vector $u$ in $F_d(q,n)$ that differs from $v$ on $m-r$ coordinates. In other words, for any vector $v$ in $F(q,n)$, we want to construct a vector $u$ agreeing with $v$ on $r$ coordinates such that $u$ is in the image of the transformation defined by $M$ (i.e., $u=My$ for some $y\in \mathbb{F}_q^r$). As $\rk(M)=r$, there exists row-vectors $R_{i_1},\ldots,R_{i_r}$ that span the row-space of $M$. A natural attempt would be do construct a partial vector $\bar{u}$ by setting $u_{i_j}\triangleq v_{i_j}$ for every $j\in [r]$. To set the remaining coordinates of vector $u$, observe that the matrix $A$ with rows $R_{i_1},\ldots,R_{i_r}$ has full rank implying that there is a unique $x$ satisfying $Ax=\bar{u}$. As rows of $A$ span rows of $M$, the remaining coordinates of $u$ can be fixed using matrix $A$ and vector $x$. This implies that vector $u$ is in $F_d(q,n)$.

Now, for $d=(1-\delta)n(q-1)$, we want a lower bound of $q^n-q^{n\epsilon}$  on $\dim(F_d(q,n))$. For this, we want to bound the size of the set $\{ x_1^{a_1}x_2^{a_2}\cdots x_n^{a_n} \mid 0\leq a_i\leq q-1,  \sum a_i=d \}$. Let $m=x_1^{a_1}\cdots x_n^{a_n}$ be a monomial of degree $d$. Consider the map $\varphi:  x_1^{a_1}\cdots x_n^{a_n} \mapsto x_1^{(q-1)-a_1}\cdots x_n^{(q-1)-a_n}$. Clearly, $\varphi$ is a bijection and $\deg(\varphi(m))\leq n(q-1)-d \leq \delta n(q-1)$ when $\deg(m)\geq (1-\delta)n(q-1)$. Hence, estimating $\dim(F_d(q,n))$ is the same as estimating $|\{ x_1^{a_1}x_2^{a_2}\cdots x_n^{a_n} \mid 0\leq a_i\leq q-1,  \sum a_i \leq \delta n (q-1) \}|$. Now, by multilinearizing the monomial $x_1^{a_1}x_2^{a_2}\cdots x_n^{a_n}$ by $x_{11}x_{12}\cdots x_{1a_1}x_{21}\cdots x_{2a_2}\cdots x_{n1}x_{n}\cdots x_{na_n}$, it suffices to count the number of multilinear monomials of degree $\delta n (q-1)$ in $n(q-1)$ variables. Therefore, $\dim(F_d(q,n)) = \binom{n(q-1)}{\delta n (q-1)} = 2^{n(q-1)H(\delta)}$ where $H$ is the binary entropy function. By choosing $\delta$  to be a small enough, we get $\dim(F_d(q,n)) \leq q^{n\epsilon}$ where $\delta$ is a function of $q$ and $\epsilon$.

Now, we move on to Step $2$. We use the {\em Croot-Lev-Pach Lemma} to obtain an upper bound on the rank of the matrix $M_p$ where the degree of $p$ if small enough. 

\begin{remark}[The Cap Set Problem]
Consider the space $\mathbb{Z}_3^n$. The {\em cap set problem} is to understand the maximum size of a {\em cap set}, a subset $A$ of $\mathbb{Z}_3^n$ that does not contain pairwise distinct elements $a,b$ and $c$ that lie in a line (i.e., $a+b=2c$). That is, we want to find the size of the largest set $A$ in  $\mathbb{Z}_3^n$ that does not contain an arithmetic progression of the form $\{x,x+r,x+2r\}$ for some $r>0$.  A trivial upper bound on $|A|$ is that of $3^n$. By using the {\em polynomial method} Croot, Lev and Pach in \cite{CLP17} showed that over  $\mathbb{Z}_4^n$, any cap set $A$ has size at most $4^{cn}$ where $c\approx 0.926$. For more on this problem, see blog posts \cite{Tao17,Gowers17} and references therein.
\end{remark}

We now state and prove the {\em Croot-Lev-Pach Lemma} completing the proof of Step 2.

\begin{lemma}
\label{lem:clp}
 Let $p$ be a polynomial in $F_d(q,n)$ and $M_p$ be the $q^n\times q^n$ matrix given by $M_p[I,J] =p(I+J)$ for all $I,J\in \mathbb{F}_q^n$. Then, $\rk(M_p)\leq 2\cdot\dim(F_{d/2}(q,n))$.
\end{lemma}

\begin{proof}
Let $p$ be a polynomial in $F_d(q,n)$, so $\deg(p)\leq d$. We show that for any $x,y\in \mathbb{F}_q^n$, $p(x+y)=\sum_{i=1}^{R}f_i(x)\cdot g_i(y)$, where $R\leq 2\cdot \dim(F_{d/2}(q,n))$ and the polynomial $f_i$ (respectively $g_i(y)$) is independent of what $x\in \mathbb{F}_q^n$ (respectively $y\in \mathbb{F}_q^n$) is.
This immediately implies that $M_p = \sum_{i=1}^{R} M_i$ where each $M_i$ is the outer-product of two vectors in  $\mathbb{F}_q^t $ ($t=q^n$) and $\rk(M_i)=1$. Therefore, $\rk(M_p)\leq R \leq 2\cdot\dim(F_{d/2}(q,n))$. As the polynomial $p$ is in $F_d(q,n)$, there exists coefficients $\alpha_{I,J}\in \mathbb{F}_q$ (depending on $p$) such that for any $x,y\in \mathbb{F}_q^n$,
\begin{align*}
    p(x+y)&= \sum_{\substack{I,J\subseteq [n]\\ |I|+|J|\leq d}}\alpha_{I,J} x^I y^J
\end{align*}
where for any $I,J\subseteq [n]$, $x^I = \prod_{i\in I}x_i$ and  $y^J = \prod_{j\in J}y_j$. For every $I,J\subseteq [n], |I|+|J|\leq d$, we have either $|I|\leq d/2$ or $|J|\leq d/2$. Then,
\begin{equation}
    p(x+y) = \sum_{\substack{I\subseteq [n]\\ |I|\leq d/2}}x^I \left(\sum_{\substack{J\subseteq [n]\\ |J|\leq d-|I|}} \alpha_{I,J} y^J\right) + \sum_{\substack{J\subseteq [n]\\ |J|\leq d/2}}y^J \left(\sum_{\substack{I\subseteq [n]\\ d/2 < |I|\leq d-|J|}} \alpha_{I,J} x^I \right)
\end{equation}
Let $m$ be the number of subsets of $\{1,\ldots,n\}$ of size at most $d/2$. Note that $m= \sum_{i=0}^{d/2} \binom{n}{i} = \dim(F_{d/2}(q,n))$. Let $\{S_1,\ldots, S_m \}$ be subsets of $\{1,\ldots,n\}$ of size at most $d/2$. We now define vectors $\bar{f}$ and $\bar{g}$ in $\mathbb{F}^{2m}$ as follows:
\begin{itemize}
    \item for $i\in [m]$, $f_i(x)=x^{S_i}$; and $g_i(y) = \sum\limits_{\substack{J\subseteq [n]\\ |J|\leq d-|S_i|}} \alpha_{S_i,J} y^J$.
    \item for $i\in [m]$, $f_{i+m}(x) = \sum\limits_{\substack{I\subseteq [n]\\ d/2 < |I|\leq d-|S_i|}} \alpha_{I,S_i} x^I$; and $g_{i+m}(y) = y^{S_i}$
\end{itemize}
Clearly, $p(x+y)=<\bar{f},\bar{g}>$. 
Hence, $p(x+y)=\sum_{i=1}^{R}f_i(x)\cdot g_i(y)$, where $R\leq 2\cdot \dim(F_{d/2}(q,n))$ and $\rk(M_p)\leq 2\cdot\dim(F_{d/2}(q,n))$.
\end{proof}

Now, we need to estimate $\dim(F_{d/2}(q,n))$. For this, we need upper bound on number of monomials in $F(q,n)$ of degree at most $d/2$ which is $q^n\cdot \Pr\limits_m[\deg(m)\leq d/2] \leq q^n\exp^{-\delta^2n/4}$ by Chernoff bound for $d=(1-\delta)n(q-1)$ when $m$ is a random monomial. Therefore, $\dim(F_{d/2}(q,n))\leq q^n\cdot q^{-\frac{n\delta^2}{4\log q}} \leq q^{n(1-\epsilon')}$ for some $\epsilon'>0$ thus completing the proof.

%% file: generalized-hadamard.tex
Dvir and Liu in \cite{DL19} showed upper bounds on the rigidity of the {\em generalized Hadamard matrices} which were conjectured to be rigid. In the following subsection, we survey this upper bound from \cite{DL19}.

\subsection{Non-rigidity of Generalized Hadamard Matrices}

For the whole of this subsection we will deal with a {\em weaker} notion of rigidity.

A matrix $A\in\mathbb{F}^{n\times n}$ has {\em  weak rigidity} at most $s$ for rank $r$  if the rank of matrix $A$ can be reduced to $r$ by changing at most $s$ entries in every row and every column of $A$. The weak rigidity of a matrix $A$ for rank $r$ is denoted by $WR_A(r)$\footnote{In \cite{DL19} the authors use the term regular rigidity. However for ease, we use the term weak-rigidity}. Observe that this is a weaker notion of matrix rigidity that we have seen so far as $R_A(r)\leq n\cdot s$ whenever $WR_A(r)\leq s$.

The {\em generalized Hadamard matrix} $H_{d,n}$ is a $d^n\times d^n$ matrix given by $H_{d,n}[I,J] = \omega^{I\cdot J}$ for $I,J\in \mathbb{Z}_d^n$ where $w=e^{\frac{2\pi i}{d}}$ is the $d^{th}$ root of unity.  One of the results of \cite{DL19} is that generalized Hadamard matrices are not weakly rigid over $\mathbb{C}$. Note that these results are stronger than just saying that generalized Hadamard matrices are not rigid over $\mathbb{C}$.

\begin{theorem}
\label{thm:gen-hadamard}
Let $d,n$ be positive integers. For any $\epsilon \in (0, 0.1)$ and  $n \geq \frac{d^2 (\log d)^2}{\epsilon^4}$, there exists an $\epsilon ' = \frac{\epsilon^4}{d^2\log d}$ such that ${ WR}_{H_{d,n}}\left(d^{n\left(1-\epsilon '\right)}\right)\leq d^{n\epsilon}$.
\end{theorem}

In order to proceed with the proof of Theorem \ref{thm:gen-hadamard} we need to introduce a few notations and make some preliminary observations.
For any $I\in \mathbb{Z}_d^n$ such that $I=(i_1,i_2,\ldots,i_n)$ we denote by $x^I$ the monomial $x_1^{i_1}x_2^{i_2}\cdots x_n^{i_n}$. Let $f:\mathbb{Z}_d^n \rightarrow \mathbb{C}$ be any function. We can associate with function $f$:
\begin{enumerate}
\item[(i)] a polynomial $P_f\in \mathbb{C}[x_1,\ldots ,x_n]$ given by $P_f\triangleq \sum_{I\in\mathbb{Z}_d^n} f(I)x^{I}$; and
\item[(ii)] a $d^n\times d^n$ matrix $M_f$ given by $M_f [I,J] \triangleq f(I+J)$ for $I,J\in \mathbb{Z}_d^n$.
\end{enumerate}
It is reasonable to expect interesting connections between the matrix $M_f$ and polynomial $P_f$ which we pen down in the following observation:

\begin{obs}
\label{obs:rank-roots}
Let $f:\mathbb{Z}_d^n \rightarrow \mathbb{C}$ be any function. If the polynomial $P_f$ has $r$ roots in the set $\{(\omega^{i_1},\ldots,\omega^{i_n})\mid (i_1,\ldots,i_n)\in \mathbb{Z}_d^n \}$ then $\rk(M_f) = d^n-r$.
\end{obs}

Proof of the above observation is based on a simple fact that the  matrix $H_{d,n}\cdot M_f\cdot H_{d,n}$ is a $d^n\times d^n$ diagonal matrix whose $[I,I]^{th}$ diagonal entry is given by $d^n\cdot P_f (\omega^{I})$ where $\omega^{I}$ denotes the tuple $(\omega^{i_1},\omega^{i_2},\ldots,\omega^{i_n})$ for any $I =(i_1,i_2,\ldots,i_n) \in \mathbb{Z}_d^n$.\footnote{The notation $\omega^{[I]}$ is more appropriate as $(\omega^{i_1},\omega^{i_2},\ldots,\omega^{i_n})$ is tuple in $\mathbb{Z}_d^n$. However we will use $\omega^{I}$ for ease of notation.} 
Now, with Observation \ref{obs:rank-roots} in hand, we sketch the proof of Theorem \ref{thm:gen-hadamard}. \\

\noindent\textit{Proof Sketch of Theorem \ref{thm:gen-hadamard}.} The proof proceeds in two steps: 
\begin{enumerate}
    \item  Rescale the rows and columns of $H_{d,n}$ to obtain a matrix $H_{d,n}'$ such that there exists a symmetric function $f:\mathbb{Z}_d^n\rightarrow \mathbb{C}$ with $M_f=H_{d,n}'$. The rows and columns of $H_{d,n}$ are uniformly rescaled in such a way that $ {WR}_{H_{d,n}}(r)= {WR}_{M_f}(r)$ for any $r$.
    \item For any $\epsilon\in (0,0.1)$ and any symmetric function $f:\mathbb{Z}_d^n\rightarrow \mathbb{C}$, by changing $f$ on at most $d^{n\epsilon}$ many values, obtain a symmetric function $f':\mathbb{Z}_d^n\rightarrow \mathbb{C}$ such that the matrix $M_{f'}$ has $\rk(M_{f'}) \leq  d^{n(1-\epsilon')}$ where $\epsilon' = \frac{\epsilon4}{d^2\log d}$. The upper bound on $\rk(M_{f'})$ follows from Observation \ref{obs:rank-roots} as the polynomial $P_{f'}$ has {\em many} roots in $\{(\omega^{i_1},\ldots,\omega^{i_n})\mid (i_1,\ldots,i_n)\in \mathbb{Z}_n^d \}$.
\end{enumerate}

\noindent {\em Proof of Step 1:} Let $H_{d,n}$ be the $d^n\times d^n$ generalized Hadamard matrix. Let $\mu\in \mathbb{C}$ be such that $\mu^2 = \omega$.
For every $I,J \in \mathbb{Z}_d^n$, the $d^n\times d^n$ matrix $H_{d,n}'$ is  obtained by multiplying every element of the $I^{th}$ row  by $\mu^{I\cdot I}$ and every element of the $J^{th}$ column by $\mu^{J\cdot J}$. Now, we define $f:\mathbb{Z}_d^n\rightarrow \mathbb{C}$ as: $f(I)=\mu^{i_1^2+i_2^2 +\cdots + i_n^2}$ for any $I=(i_1,i_2,\ldots,i_n)\in \mathbb{Z}_d^n$. Observe that $f$ is a symmetric function and for any $I,J \in \mathbb{Z}_d^n$ the matrix $M_f [I,J] = f(I+J) = \mu^{(i_1+j_1)^2+(i_2+j_2)^2 +\cdots + (i_n+j_n)^2} = H_{d,n}'[I,J]$. Note that the function $f$ is well-defined and ${ WR}_{H_{d,n}}(r)= { WR}_{M_f}(r)$ for any $r$. \\

Now, in the following step we will have to modify the function $f$ so that the polynomial $P_f$ satisfies the hypothesis of Observation \ref{obs:rank-roots}.


\noindent {\em Proof of Step 2:} Given any symmetric function $f:\mathbb{Z}_d^n\rightarrow \mathbb{C}$ by changing $f$ on a {\em ``small"} set $T$ of values in $\mathbb{Z}_d^n$, we want to construct a symmetric function $f':\mathbb{Z}_d^n\rightarrow \mathbb{C}$  such that the polynomial $P_{f'}$ vanishes on a {\em ``large"} set $S$ in $\{(\omega^{i_1},\ldots,\omega^{i_n})\mid (i_1,\ldots,i_n)\in \mathbb{Z}_d^n \}$. The sets $S$ and $T$ are defined as follows: 

\begin{enumerate}
    \item Let $m = \frac{n(1-\epsilon^2)}{d}$ and $S$ be the set of tuples $ (i_1, i_2,\ldots, i_n) \in \mathbb{Z}_d^n$ such that $i_1= i_2 =\cdots = i_m =0$; $i_{m+1} = i_{m+2} = \cdots = i_{2m} = 1$ and so on till $i_{(d-1)m+1}= i_{(d-1)m+2} = \cdots = i_{dm} = m$. 
    \item Let $T\subseteq \mathbb{Z}_d^n$ be the set of tuples with at least $n(1-\epsilon^2)$ many zeros.
\end{enumerate}

Having defined sets $S$ and $T$, the following lemma(which we will prove later) ensures that we can use Observation \ref{obs:rank-roots} to complete the proof of Theorem \ref{thm:gen-hadamard}.


\begin{lemma}
\label{lem:symmetric}
Let $f:\mathbb{Z}_n^d\rightarrow \mathbb{C}$ be any symmetric function. By changing $f$ on values in $T$, we can construct a symmetric function $f':\mathbb{Z}_n^d\rightarrow \mathbb{C}$ such that $P_{f'}(\omega^I) = 0$ for every $I\in S$.
\end{lemma}
Now, assuming Lemma \ref{lem:symmetric}, let us  complete the proof of Theorem \ref{thm:gen-hadamard}. Note that $M_f = (M_f- M_{f'}) + M_{f'}$ and we bound $\rk(M_{f'})$ and $\spar(M_f- M_{f'})$.  By Lemma \ref{lem:symmetric} $P_{f'}$ vanishes on the set $\{ \omega^I \mid I\in S\}$. Hence  $P_{f'}$ has $|S| = d^{n-dm} = d^{n\epsilon^2}$ many roots in $\{(\omega^{i_1},\ldots,\omega^{i_n})\mid (i_1,\ldots,i_n)\in \mathbb{Z}_d^n \}$ as $m=n(1-\epsilon^2)/d$. However, as $f'$ is a symmetric function, the polynomial $P_{f'}$ not only vanishes on $\{\omega^I \mid I\in S\}$ but also on $\omega^J$ for all tuples $J$ in $\mathbb{Z}_d^n$ that are obtained by permuting the entries of $I = (i_1,\ldots,i_n)$. That is, $P_{f'}(\omega^J)=0$ for all $J$ in $\perm(S)=\{ \perm(I) \mid I\in S\}$ where $\perm(I)$ denotes the set of distinct permutations are obtained by permuting the entries of $I = (i_1,\ldots,i_n)$.

Thus, by Observation \ref{obs:rank-roots}, $\rk(M_{f'})$ is exactly the number of tuples in $\mathbb{Z}_d^n$ that are not in $\perm(S)$ and estimating $\rk(M_{f'})$ amounts to estimating the size of $\mathbb{Z}_d^n \setminus \perm(S)$. A tuple $I\in \mathbb{Z}_d^n$ is in $\perm(S)$ iff every $a\in \{0,1,\ldots,d-1\}$ appears at least $m$ times. Then, $\rk(M_{f'})$ is given by the number of tuples in  $\mathbb{Z}_d^n$ such that there exists an $a\in \{0,1,\ldots,d-1\}$, $a$ appears less than $m$ times.  Let $\tau\in_{r} \mathbb{Z}_d^n, i\in \{0,\ldots,d-1\}$ and $X_i$ be a random variable that denotes the number of times $i$ appears in the tuple $\tau$. Then, $\Pr\left[X_i < \frac{(1-\epsilon^2)n}{d}\right] \leq e^{-\frac{2\epsilon^4 n}{d^2}}$ and $\Pr[\tau \not \in \perm(S)] \leq d\cdot e^{-\frac{2\epsilon^4 n}{d^2}}$. The expected size of  $\mathbb{Z}_d^n \setminus \perm(S)$ is at most $d^n\cdot d\cdot e^{-\frac{2\epsilon^4 n}{d^2}}$. Thus, when $n> \frac{d^2(\log d)^2}{\epsilon^4}$, the size of $\mathbb{Z}_d^n \setminus \perm(S)$ is $d^{n(1-\epsilon')}$ for $\epsilon' = \frac{\epsilon^4}{d^2\log d}$. This immediately implies that $\rk(M_{f'})\leq d^{n(1-\epsilon')}$ for $\epsilon' = \frac{\epsilon^4}{d^2\log d}$.

To upper bound $\spar(M_f- M_{f'})$, it is enough to estimate $|T|$ which is the number of tuples in $\mathbb{Z}_d^n$ with at least $n(1-\epsilon^2)$ many zeros. Let $\tau\in_{r} \mathbb{Z}_d^n$ and  $X$ be a random variable that denotes the number of zeros in $\tau$. Then, $\Pr[ \tau \in T ] = \Pr [X \geq n(1-\epsilon^2)] \leq e^{-D((1-\epsilon^2)||1/d)} \leq d^{-n(1-\epsilon)}$ when $\epsilon\in (0,0.1)$. The expected size of set $T$ is at most $d^n\cdot d^{-n(1-\epsilon)}$. This implies that $\spar(M_f- M_{f'}) \leq d^{n\epsilon}$. Thus, by changing $M_f$ on $|T|\leq d^{n\epsilon}$ values in every row, the rank of $M_f$ becomes $d^{n(1-\epsilon')}$ implying that ${ WR}_{H_{d,n}}\left(d^{n\left(1-\epsilon '\right)}\right)\leq d^{n\epsilon}$.
$\qedh$

\paragraph{}
We now turn to the proof of Lemma $\ref{lem:symmetric}$. Let $f:\mathbb{Z}_d^n\rightarrow \mathbb{C}$ be any symmetric function and $T\subseteq \mathbb{Z}_d^n$ be the set of tuples with at least $n(1-\epsilon^2)$ many zeros. As we want to change $f$ only on tuples in $T$, for all $J\not\in T$, $f'(J)=f(J)$. Also, as we want $f':\mathbb{Z}_d^n\rightarrow \mathbb{C}$ to be symmetric, we require that for every $j\in [k]$, for any $J,J'\in \perm(J_j)$, $f'(J)=f'(J')$. Since we do not know what values to change $f$ to on tuples in set $T$, the most natural approach would be to come up with a system of equations with these as the unknown variables and $P_{f'}(\omega^{I}) = 0$ for every $I\in S$ as the constraints and show that this system has a solution. 

\begin{align*}
    P_{f'}(\omega^{I}) &= 0 \text{~~for all $I\in S$} \\
    \sum\limits_{J\in T} f'(J)\omega^{I\cdot J} + \sum\limits_{J'\not \in T} f(J')\omega^{I\cdot J'} &= 0  \text{~~for all $I\in S$}
\end{align*}

Note that we require the new function $f'$ to be symmetric. Also, let us consider the equivalence classes obtained by permuting the tuples $S$ and $T$ and denote by $\rep(S)=\{I_1,\ldots,I_{\ell}\}$ and $\rep(T)=\{ J_1,\ldots,J_k\}$ the set obtained by picking one representation from each equivalence class of $S$ and $T$ respectively. Now, we define a system of linear equations with $\{f'(J_j) \mid j\in [k] \}$ as the unknowns labelled as $a_1,\ldots,a_k$:  
\begin{align*}
    \sum\limits_{j=1}^{k} a_j \sum\limits_{J\in \perm(J_j)}  \omega^{I\cdot J} + \sum\limits_{J'\not \in T} f(J')\omega^{I\cdot J'} &= 0 \text{~~for all $I\in S$}
\end{align*}
Since $f:\mathbb{Z}_d^n\rightarrow \mathbb{C}$ is a symmetric function, it suffices to consider the following set of linear equations:
\begin{align*}
    \sum\limits_{j=1}^{k} a_j \sum\limits_{J\in \perm(J_j)}  \omega^{I_i\cdot J} + \sum\limits_{J'\not \in T} f(J')\omega^{I_i\cdot J'} &= 0 \text{~~for all $I_i\in \rep(S)$} 
    \end{align*}
That is,
    \begin{align*}
    \sum\limits_{j=1}^{k} a_j \sum\limits_{J\in \perm(J_j)}  \omega^{I_i\cdot J} &=- \sum\limits_{J'\not \in T} f(J')\omega^{I_i\cdot J'}  \text{~~for all $I_i\in \rep(S)$}
\end{align*}
Let $M$ be the $\ell\times k$ coefficient matrix given by $M_{ij}=\sum\limits_{J\in \perm(J_j)}  \omega^{I_i\cdot J}$. In order to show that the above non-homogeneous system of linear equations has a solution, it is enough to show that the column space of $M$ has full rank. That is, for each $i=1,\ldots,\ell$, we require constants $a_1,\ldots ,a_k$ such that: 
\begin{align}
    \sum\limits_{j=1}^{k} a_j M_{ij} &\neq 0 \label{eq:eq5} \\
    \sum\limits_{j=1}^{k} a_j M_{i'j} &= 0 \text{~~for $i'\neq i$} \label{eq:eq6} 
\end{align} Fix $i'=i_0$ in Equations (\ref{eq:eq5}) and (\ref{eq:eq6}). We need $a_1,\ldots ,a_k$ such that 
\begin{align}
    \sum\limits_{j=1}^{k} a_j \sum\limits_{J\in \perm(J_j)}  \omega^{I_{i_0}\cdot J} &\neq 0 \label{eq:eq3} \\
    \sum\limits_{j=1}^{k} a_j \sum\limits_{J\in \perm(J_j)}  \omega^{I_i\cdot J} &= 0 \text{~~for $i\neq i_0$} \label{eq:eq4}
\end{align}
Clearly, from equations (\ref{eq:eq3}) and (\ref{eq:eq4}) this is equivalent to constructing an $n$-variate polynomial $P(x_1,\ldots,x_n) = \sum\limits_{j=1}^{k} a_j \sum\limits_{J\in \perm(J_j)}  x^J$ that vanishes on $\omega^{I_i}$ for any $i\in [\ell], i\neq i_0$ but does not vanish on $\omega^{I_{i_0}}$.

However, for any tuple $I= (i_1,\ldots,i_n)$ in $S$, the first $dm$ entries are fixed and let $I'$ be the sub-tuple $(i_{dm+1},\ldots,i_n)$ of $I$ and $I'(j)$ denote the $j^{th}$ entry of tuple $I'$. Thus, we want an $(n-dm)$-variate polynomial $Q(x_{dm+1} \ldots,x_n) = P(1, \ldots 1, \ldots, \omega^{d-1},\ldots,\omega^{d-1}, x_{dm+1} \ldots,x_n)$ that vanishes on $\omega^{I_i'}$ if and only if $i\neq i_0$ \footnote{Note that for every $i_0$, we get a polynomial $Q_{i_0}$ that is dependent on the tuple $I_{i_0}$. For ease of notation, we refer to the polynomial as $Q$ dropping the subscript $i_0$}.

Let $Q(x_{dm+1} \ldots,x_n) = \sum\limits_{I'\in \perm(I_{i_0}')}\left( \frac{x_{dm+1}^d-1}{x_{dm+1}-\omega^{I'(0)}} \right)\cdots \left(\frac{x_{n}^d-1}{x_n-\omega^{I'(n-dm)}}  \right)$. The proof of Theorem \ref{thm:gen-hadamard} is complete with the following claim:
\begin{claim} 
\label{claim:hadamard-claim1}
Let $Q(x_{dm+1} \ldots,x_n)$ be the polynomial defined above. Then,
\begin{itemize}
    \item[(i)] $Q(x_{dm+1} \ldots,x_n)$ vanishes on $\omega^{I_i'}$ if and only if $i\neq i_0$.
    \item[(ii)] $Q(x_{dm+1} \ldots,x_n) = P(1, \ldots 1, \ldots, \omega^{d-1},\ldots,\omega^{d-1}, x_{dm+1} \ldots,x_n)$.
\end{itemize}
\end{claim}
We do not include a proof of Claim \ref{claim:hadamard-claim1} here but it is not hard to prove the above properties of the polynomial $Q$.

The above discussion proves that for any symmetric function $f:\mathbb{Z}_d^n\rightarrow \mathbb{C}$, for every $\epsilon\in (0,0.1)$ and sufficiently large $n$, ${WR}_{M_f}(d^{n(1-\epsilon')})\leq d^{n\epsilon}$ for some $\epsilon'>0$ ($\epsilon'$ is a function of $d$ and $\epsilon$).

We now extend this rigidity upper bound to matrices $M_f$ corresponding to functions that are not symmetric:

\begin{theorem}
\label{thm:non-symmetric}
Let $f:\mathbb{Z}_d^n\rightarrow \mathbb{C}$ be any function. For any $\epsilon\in (0,0.1)$ and $n \geq \frac{d^2 (\log d)^2}{\epsilon^4}$, there exists an $\epsilon'= \frac{\epsilon^4}{d^2\log d}$ such that ${WR}_{M_f}(2d^{n(1-\epsilon')})\leq d^{2n\epsilon}$.
\end{theorem}

The proof of Theorem \ref{thm:non-symmetric} is immediate from Theorem \ref{thm:gen-hadamard}, the following property of Hadamard matrices (mentioned in Lemma \ref{lem:hadamard-prop} proof of which is straightforward from the definition of matrices $H_{d,n}, M_f$ and polynomial $P_f$) and a simple tool that reduces the task of proving non-rigidity of a matrix $B$ to proving non-rigidity of the matrix $A$ that diagonalizes it (mentioned in Lemma \ref{lem:diag}). 

\begin{lemma}
\label{lem:hadamard-prop}
Let $f:\mathbb{Z}_d^n\rightarrow \mathbb{C}$ be any function. Then $D=H_{d,n}\cdot M_f\cdot H_{d,n}$ is a $d^n\times d^n$ diagonal matrix with $D[I,I]=d^n\cdot P_f(\omega^{I})$ where $\omega$ is the $d^{th}$ root of unity.
\end{lemma}

\begin{lemma}
\label{lem:diag}
Let $B=A^{*}DA$ (respectively $B=ADA$) where $A^{*}$ is the conjugate transpose of $A$ and $D$ is a diagonal matrix. If ${WR}_A(r)\leq s$ then ${WR}_B(2r)\leq s^2$. 
\end{lemma}

\begin{proof}
If ${WR}_{A}(r)\leq s$ then $A=S+L$ where $\rk(L)\leq r$ and $S$ has at most $s$ non-zero entries in every row and column. Then,
\begin{align*}
B-S^{*}DS &= B-S^{*}DS+A^{*}DS -A^{*}DS     \\
&= A^{*}DA  -S^{*}DS+A^{*}DS -A^{*}DS & [\because B=A^{*}DA] \\
&= A^{*}D(A-S) + (A^{*}-S^{*})DS  \\ 
B &= S^{*}DS + A^{*}D(A-S) + (A^{*}-S^{*})DS
\end{align*}
where the matrix $S^{*}DS$ has at most $s^2$ non-zero entries in each row and column as  $S$ has at most $s$ non-zero entries in every row and column. Further, $\rk ( A^{*}D(A-S) + (A^{*}-S^{*})DS)\leq 2r$ as $\rk(A-S)\leq r$. Therefore, ${WR}_B(2r)\leq s^2$.
\end{proof}

The proof of Theorem \ref{thm:non-symmetric} is immediate though we sketch it here for the sake of completeness.

\noindent{\em Proof of Theorem \ref{thm:non-symmetric}.} By Lemma \ref{lem:hadamard-prop}, $M_f=H_{d,n}^{-1}\cdot D\cdot H_{d,n}^{-1}$. From Theorem \ref{thm:gen-hadamard},  we have $WR_{H_{d,n}}(d^{n(1-\epsilon')}) \leq d^{n\epsilon}$ for any $\epsilon\in (0,0.1)$ and some $\epsilon'$. This immediately implies that $WR_{M_f}(2d^{n(1-\epsilon')}) \leq d^{2n\epsilon}$ from Lemma \ref{lem:diag}.

\paragraph*{A brief note on non-rigidity of Fourier and Circulant matrices.}Although understanding the rigidity of generalized Hadamard matrices is of independent interest,Theorem \ref{thm:gen-hadamard} also acts as a building block in showing that {\em Fourier matrices} are also not rigid which is the main theorem of \cite{DL19}. As {\em Fourier matrix} $F_d$ is the $d\times d$ matrix $H_{d,1}$, the generalized Hadamard matrix $H_{d,n}=  \underbrace{F_d \otimes F_d \cdots \otimes F_d}_{\text{$n$ times}}$. Even 
though we don not include the proof of non-rigidity of Fourier matrices which is quite involved, among other basic blocks it uses Theorem \ref{thm:gen-hadamard} as well as the following interesting lemma which analyses the weak rigidity of tensor product of two matrices:

\begin{lemma}
\label{lem:tensor-product}
Let $A\in\mathbb{F}^{m\times m}$ and $B\in \mathbb{F}^{n\times n}$. Then for any $r_1\leq m$ and $r_2\leq n$,
$WR_{M}(r_1 n+r_2 m)\leq WR_{A}(r_1)\cdot WR_{B}(r_2)$ where $M=A\otimes B$.
\end{lemma}

\begin{proof}
Suppose $WR_{A}(r_1)\leq s_1$ and $WR_{A}(r_2)\leq s_2$ then there exists $S_1,S_2$ of appropriate dimensions such that $\rk(A+S_1)\leq s_1$ and $\rk(B+S_2)\leq s_2$. Now, we want to argue about the rank of $M + (S_1\otimes S_2)$:
\begin{align*}
M+ (S_1\otimes S_2)  &= (A\otimes B) + (S_1\otimes S_2) \\
& = (A\otimes B) + (S_1\otimes B) - (S_1\otimes B ) + (S_1\otimes S_2)\\
&= (A+ S_1)\otimes B - S_1\otimes(B+S_2)
\end{align*}
Thus, $\rk(M+(S_1\otimes S_2))=r_1n+r_2m$ and sparsity of 
$S_1\otimes S_2$ is $s_1s_2$.

\end{proof}

In \cite{DL19}, the authors also prove that {\em circulant matrices} are not rigid. Let $c_0,\ldots,c_{n-1}\in \mathbb{F}$. A matrix $C_n\in \mathbb{F}^{n\times n}$ is said to be {\em circulant} if 
\begin{center}
 $C_n= \begin{bmatrix}
c_0 & c_{n-1} & \cdots & c_2 & c_1 \\
c_1 & c_0 & c_{n-1} & \cdots & c_2 \\
\vdots & \vdots & \vdots & \vdots & \vdots \\
c_{n-1} & c_{n-2} & \cdots & c_1 & c_0 \\
\end{bmatrix}$
\end{center}   
   
Observe that circulant matrix is a special case of Toeplitz matrix. Dvir and Liu\cite{DL19} prove that for sufficiently large $n$, $C_n$ is not rigid. Hence, although rigidity lower bound of Toeplitz matrix in Theorem \ref{thm:toeplitz} is reasonable for much smaller $n$(as noted in Remark \ref{rem:toeplitz}) it is impossible to match the lower bound in Question \ref{que:rigidity-Valiant}.

\begin{remark}
In \cite{DL19} the matrix $M_f$ is given by $M_f [I,J] = f(I+J)$ for $I,J\in \mathbb{Z}_n^d$. However the argument also works for $M_f [I,J] = f(I-J)$ for $I,J\in \mathbb{Z}_n^d$ as the two definitions differ only upto permutation of rows/columns giving the same rigidity bounds. Further, Theorem $\ref{thm:non-symmetric}$ extends the results of \cite{DE17} to the field of complex numbers and the result of \cite{AW17} to arbitrary $d$ while the result in \cite{AW17} is for $d=2$. 
\end{remark}

%% file: data-structure-lb.tex
Given a database $X$ of $n$ elements $\{x_1,\ldots,x_n\}$, an  {\em $(s,t)$-data structure} for $X$ is a way to store $X$ into $s$ memory cells so that any query concerning $X$ can be answered effectively in time $t$. Let ${\cal Q}=\{q_1,\ldots,q_m\}$ be a set of $m$ queries on $X$ (usually $m=\poly(n)$). The time to answer a query is the number of cells accessed and computation on the accessed cells is for free. 

There are two trivial static data structures for any problem:
\begin{itemize}
    \item[(i)] Pre-compute answers to all queries in ${\cal Q}$ and store them in space $\poly(n)$ as $|\mathcal{Q}| =\poly(n)$. In this case, any query in ${\cal Q}$ can be answered in constant time.
    \item[(ii)] Store the entire database $X$ in memory using $n$ memory cells and for every query in ${\cal Q}$ compute the answer by performing a linear search on the memory (as query answer may depend on all inputs). In this case, both space and time are linear.  
\end{itemize}
In this regard, one major goal  is to understand {\em time-space tradeoffs}. That is, can we get better (sub-linear) upper bounds on the query time against linear space for static data structures? Standard counting arguments show that for most data structure problems either time is $|X|^{0.99}$ or space is $|Q|^{0.99}$. Further, there exists explicit static data structure problems such that any data structure that uses space $O(n)$ requires time  $\Omega(\log n)$ to answer queries in ${\cal Q}$ where $|{\cal Q}|=\poly(n)$ (see \cite{PTW10,Lar14} for details).  This brings us to the following question: 
 
 \begin{question}
\label{que:ds-lowerbound}
 Does there exist an explicit data structure problem $P$ such that any $(O(n),t)$-data structure for $P$ requires $t=\omega(\log n)$?
 \end{question}

The above question is quite challenging and this difficulty in proving explicit data structure lower bounds is justified as data structures correspond to circuits with arbitrary gates. See Figure \ref{fig:ds-circuit} for a pictorial representation of the following discussion. An $(s,t)$-data structure for a database $X$ containing  $n$ field elements $\{x_1,\ldots,x_n\}$ can be viewed as a depth-$2$ circuit whose leaf gates are elements of $X$. The  middle layer consists of $s$ gates of unbounded fan-in representing the $s$ memory cells and the top layer consists of $m$ gates representing queries $q_1,\ldots,q_m$ in ${\cal Q}$. As the data structure is allowed to take time $t$ on any query $q\in {\cal Q}$, the fan-in of the gates in the top layer are bounded by $t$. The mapping of elements in $X$ to $s$ memory cells can be viewed as a function $P:\mathbb{F}^n\rightarrow \mathbb{F}^s$ ($P$ stands for {\em pre-processing function}) and the memory cells associated with queries in top layer gates can be viewed as a  function $Q:\mathbb{F}^s\rightarrow \mathbb{F}^m$ ($Q$ stands for {\em query function}). Note that this correspondence between $(s,t)$-data structure for $X$ and an $m$-output unbounded top fan-in depth-$2$ circuit of width $s$ with arbitrary gates holds only when the queries in ${\cal Q}$ are non-adaptive.

\begin{figure}[H]
\begin{center}
\includegraphics[scale=0.9]{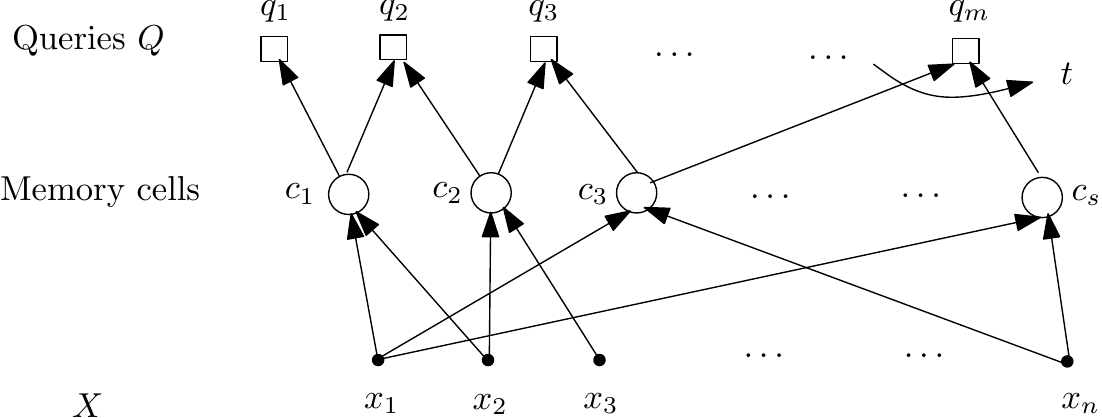}
\caption{Data structure viewed as a depth-$2$ circuit with arbitrary gates}
\centering
\label{fig:ds-circuit}
\end{center}
\end{figure}

\begin{remark}
Throughout this section, query time is measured  by the number of cells probed where each cell is capable of holding multiple bits. This measure was introduced by Yao in \cite{Yao81}. However, there is yet another interesting data structure model called the {\em bit-probe model} introduced in \cite{EF75} in which query time is measured by the number of bits accessed to answer the query. In this article we will work with the cell-probe model. 
\end{remark}

This correspondence between data structures and circuits with arbitrary gates hints that proving data structure lower bounds are considerably hard. Hence, it is reasonable to place certain restriction on the data structure to get better lower bounds. In this regard, Dvir et al.\ in \cite{DGW19} consider static data structures with the following restrictions:
\begin{itemize}
    \item The database $X=\{x_1,\ldots,x_n\}$ contains elements from $\mathbb{F}$.
    \item The data structure can perform only linear operations on the database $X$. That is,  $P:\mathbb{F}^n\rightarrow \mathbb{F}^s$ and  $Q:\mathbb{F}^s\rightarrow \mathbb{F}^m$ are linear functions.
\end{itemize}
In this case, the $m$ queries $\{q_1,\ldots,q_m\}$ in ${\cal Q}$ can be viewed as  $m$ rows $R_1,\ldots,R_m$ of a matrix $M\in \mathbb{F}^{m\times n}$. Whenever query $q_i$ is raised, the data structure returns the inner product $\inner{R_i}{X}$, an element in $\mathbb{F}$ (here $X= (x_1 ~ x_2 ~\cdots ~ x_n)$ is viewed as a vector). A data structure for the set of queries in ${\cal Q}$ using space $\leq s$ and query time $\leq t$ with $P,Q$ being linear functions is called an {\em $(s,t)$-linear data structure} for $M$.


In \cite{DGW19}, the authors demonstrate a connection between the answers to  Question \ref{que:rigidity-Valiant} and  Question \ref{que:ds-lowerbound}. In particular Dvir et al.\ prove the following theorem:

For the rest of this section, we will need a notion of rigidity weaker than matrix rigidity called {\em row-rigidity}. The row-rigidity of a matrix $M$ for rank $r$(denoted by $RR_M(r)$) is $s$ if the rank can be reduced to $r$ by changing at most $s$ entries in every row. The row-rigidity of a matrix is seemingly weaker than rigidity and stronger than weak rigidity. A matrix $M$ is {\em $t$-row sparse} if every row of $M$ has at most $t$ non-zero entries.

\begin{theorem}
\label{thm:rigidity-ds}
Let $\epsilon,\delta>0$ be constants. Let $M \in\mathbb{F}^{m\times n}$ be a matrix such that there is no $(\frac{n}{1-\epsilon},(\log n)^c)$ linear data structure for $M$. Then for some $n' \geq \alpha \cdot (\log n)^{c-1}$ there exists a matrix $M\in \mathbb{F}^{m\times n'}$ such that $RR_{M'}(\epsilon n')\geq (\log n)^{c-1}$. In fact, $M'$ is a sub-matrix of $M$ and when $M$ is explicit $M'$ is in $\P^{\NP}$.\end{theorem}

\begin{remark}
Although the above theorem relates data structure lower bounds to rigidity of rectangular matrices an analogous theorem also holds in the case of square matrices (see Theorem 2 in \cite{GPW18} for the exact statement). In fact, a query lower bound of $t$ on linear space data structure translates to row rigidity lower bound of $\frac{t}{\log n}$.  
\end{remark}

In the rest of this section, we provide the reader intuition as to why this connection between static linear data structure lower bounds and matrix rigidity is true and sketch the details of the proof. We begin with the following simple observation(whose proof intuitively follows from Figure \ref{fig:ds-circuit}):

\begin{obs}
\label{obs:ds-odim}
Let there be an $(s,t)$-linear data structure for $M\in\mathbb{F}^{m\times n}$. Then, $M=Q\cdot P$ where $Q\in\mathbb{F}^{m\times s}$ is  a $t$-row sparse matrix and $P\in\mathbb{F}^{s\times n}$.
\end{obs}

Now, we discuss a linear algebraic characterization of the existence of efficient data structures. Let $M\in\mathbb{F}^{m\times n}$ be such that $M=Q\cdot P$ where $Q\in\mathbb{F}^{m\times s}$ is  a $t$-row sparse matrix and $P\in\mathbb{F}^{s\times n}$. If we denote by $V$ the column space of matrix $M$ then there exists a subspace $U \triangleq \colsp(Q)$ of $\mathbb{F}^m$ such that $V\subseteq U$ and $U$ is a $t$-sparse vector space\footnote{A vector space $U\subseteq \mathbb{F}^m$ is $t$-sparse if it can be expressed as the column space of a matrix that is $t$-row sparse.} (as $Q$ is a $t$-row sparse matrix).  This leads us to the definition of the {\em outer-dimension} of a vector space. Informally, the outer dimension of a vector space $V$ is the dimension of the smallest $t$-sparse vector space containing (outer of) $V$. More formally, we define the {\em outer-dimension} of a vector space $V$ with respect to sparsity parameter $t$ (denoted by $\outdim_{V}(t)$) as $\min\limits_{U}\{\dim(U)\mid V\subseteq U, \text{$U$ is $t$-sparse}  \}$. In this article, for ease of notation we refer to $\outdim_{M}(t)$ to denote the outer-dimension of vector space $V$ where $V$ is $\colsp(M)$.

From the above discussion and Observation \ref{obs:ds-odim}, it is clear that if there is an $(s,t)$-linear data structure for $M$ then $\outdim_{M}(t)\leq s$. Now, consider the converse. If $\outdim_{M}(t)\leq s$ for some matrix $M\in \mathbb{F}^{m\times n}$ then by definition there exists $U\subseteq \mathbb{F}^m$ of dimension at most $s$ such that $V\subseteq U$ and $U$ is $t$-sparse (here, $V=\colsp(M)$). Let $Q\in\mathbb{F}^{m\times s}$ be such that $U$ is $\colsp(Q)$. As $V\subseteq U$, every column of $M$ can be expressed as a linear combination of the columns of $Q$. Hence $M=Q\cdot P$ where where $Q\in\mathbb{F}^{m\times s}$ is  a $t$-row sparse matrix and $P$ is a matrix in $\mathbb{F}^{s\times n}$. From the circuit view of data structures mentioned earlier this immediately gives an $(s,t)$ data structure for $M$. Hence, outer-dimension of a matrix $M$ characterizes the existence of an efficient linear data structure for $M$:

\begin{obs}
\label{obs:outdim}
There is an $(s,t)$-linear data structure for $M$ if and only if $\outdim_{M}(t)\leq s$.
\end{obs}

Recall that the goal is to understand the connection between matrix rigidity and data structures. Similar to the notion of {\em low} outer-dimension for efficient data structures, we give a linear algebraic characterization of rigid matrices. Let $M\in\mathbb{F}^{m\times n}$ be a matrix that is not row rigid (i.e., $RR_M(r)\leq t$). Then there exists matrices $S,L\in \mathbb{F}^{m\times n}$ such that every row of $S$ has at most $t$ non-zero entries and $\rk(L)\leq r$. Let $V\triangleq \colsp(M), U\triangleq\colsp(S)$ and $W\triangleq\colsp(L)$ and we have that $V=U+W$. Observe that $U$ is a $t$-sparse vector space and that $U+V \subseteq U+W$. Thus,
\begin{align*}
\dim(U+V) &\leq \dim(U+W)    \\
\dim(U)+\dim(V)-\dim(U\cap V) &\leq \dim(U)+\dim(W)-\dim(U\cap W) \\
&\leq \dim(U)+\dim(W) \\
\dim(U\cap V) &\geq \dim(V)-\dim(W) \\
&\geq \rk(M)-r
\end{align*}

Hence, whenever the row rigidity of a matrix $M$ for rank $r$ is at most $t$, there exists a $t$-sparse vector space $U$ that intersects $\colsp(M)$ in {\em a large number of dimensions}. This precisely leads us to the definition of {\em inner-dimension} of a vector space. The {\em inner-dimension} of a vector space $V$ with respect to sparsity parameter $t$ (denoted by $\indim_{V}(t)$) is defined as $\max\limits_{U}\{\dim(U\cap V)\mid \dim(U)\leq \dim(V), \text{$U$ is $t$-sparse}  \}$. In this article, for ease of notation we denote by $\indim_{M}(t)$ to denote the inner-dimension of vector space $V$ where $V$ is $\colsp(M)$. Before we move on, we make a remark on the complexity of computing the inner dimension of a given matrix (we will use this to prove Theorem \ref{thm:rigidity-ds}).

\begin{obs}
\label{obs:inner-dim}
Let ${\sf InnerDim}(M,d,t)$ denote the problem of deciding if $\indim_{M}(t)\geq d$. It is not very difficult to observe that ${\sf InnerDim}(M,d,t)$ is in $\NP$. Let $V=\colsp(M)$ and $\dim(V)=\rk(M)$. Given a witness $N$ in $\mathbb{F}^{m\times n}$ that is a $t$-row sparse matrix, the $\NP$ algorithm ${\cal A}$ verifies if $\dim(U\cap V) \geq d$ where $U=\colsp(N)$. That is, ${\cal A}$ computes $\dim(U)+\dim(V)-\dim(U+V) = \rk(M)+\rk(N)-\rk(NM)$ and test if this is at least $d$. This verification can be done in polynomial time implying that ${\sf InnerDim}(M,d,t)\in \NP$. \end{obs}

From the preceding discussion, it is clear that if $M$ is not a row rigid matrix then $M$ has a high inner-dimension. Apparently, the converse is also true.

Suppose $\indim_{M}(t)> \rk(M)- r$ for some $r$. Then by definition, there exists a $t$-sparse vector space $U \subseteq\mathbb{F}^m$ with $\dim(U)\leq \dim(V)$ and $\dim(U\cap V)> \rk(M) - r$ where $V$ is $\colsp(M)$. This means that there exists a subspace $W\subseteq \mathbb{F}^m$ with $\dim(W)<r$ such that $V=U+W$. As $U$ is a $t$-sparse vector space there is a a $t$-row sparse
matrix $A$ such that the columns of $A$ span the space $U$. Since $V=U+W$ there is a matrix $B$ of rank less than $r$ satisfying $M=AT+B$ for some $T\in GL(n,\mathbb{F})$. As $T$ is invertible, $MT^{-1}=A+BT^{-1}$ and the rank of $MT^{-1}$ can be reduced to $r$ by changing at most $t$ entries in each row. Thus, $RR_M(r)\leq  t$ as $\rk(MT^{-1})=\rk(M)$.

At the end of the above discussion on inner dimension of spaces we observe the following: 

\begin{obs}
\label{obs:inndim}
Let $M\in\mathbb{F}^{m\times n}$ be a matrix. $RR_M(r)> t$ if and only if $\indim_M(t) \leq \rk(M)- r$.
\end{obs}

In summary, there is no efficient $(s,t)$-linear data structure for $M$ if and only if $M$ has high outer-dimension and $M$ is a strongly rigid matrix if and only if $M$ has low inner-dimension. Hence, in order to prove Theorem \ref{thm:rigidity-ds}, it is enough to show that high outer-dimension of a matrix $M$ implies the existence of a sub-matrix of $M$ having low inner-dimension. \\

\noindent\textit{Proof Sketch of Theorem \ref{thm:rigidity-ds}.} We begin with the following claim that matrices having {\em large} outer-dimension have large enough sub-matrices of {\em small} inner-dimension.

\begin{claim}
\label{claim:indim-outdim}
Let $t,k\in \mathbb{Z}^{+}$ and $\epsilon\in (0,1)$ and $M\in \mathbb{F}^{m\times n}$. If $\outdim_M(tk+n\epsilon^k)\geq \frac{n}{1-\epsilon}$ then for some $n'\geq n\epsilon^k$
there exists an $m\times n'$ submatrix $M'$ of $M$ computable in $\P^{\NP}$ such that $\indim_{M'}(t)\leq \rk(M')-\epsilon n'$. 
\end{claim}

Let us complete the proof of Theorem \ref{thm:rigidity-ds} assuming Claim \ref{claim:indim-outdim}. Let $\epsilon,\delta>0$ be constants and $M\in \mathbb{F}^{m\times n}$. Suppose there is no $(\frac{n}{1-\epsilon},(\log n)^c)$ linear data structure for a matrix $M$ then by Observation \ref{obs:outdim} we know that $\outdim_{M}((\log n)^c) > \frac{n}{1-\epsilon}$. Observe that by setting $k=\frac{\log(n/t)}{\log(1/\epsilon)}$ and $t=\frac{(\log n)^{c-1}}{\log(1 /\epsilon)}-1$, we get $n\epsilon^k = n\epsilon^{\frac{\log(n/t)}{\log(1/\epsilon)}} = t$ and hence $tk+n\epsilon^k = (k+1)t$. This implies that $\outdim_{M}(tk + n\epsilon^k) > \frac{n}{1-\epsilon}$ for values of $t,k$ chosen above. Now, by Claim \ref{claim:indim-outdim} for some $n'\geq n\epsilon^k$
there exists an $m\times n'$ submatrix $M'$ of $M$ computable in $P^{\NP}$ such that $\indim_{M'}(t)\leq \rk(M')-\epsilon n'$. From Observation \ref{obs:inndim} we get $RR_{M'}(\epsilon n') \geq (\log n)^{c-1}$. $\qedh$ \\

Now, let us briefly sketch the proof of Claim \ref{claim:indim-outdim}. Let us begin by observing that matrices with {large} inner-dimension have a {\em decomposition property}  that can be obtained efficiently given access to an $\NP$ oracle. That is, given an $m\times n$ matrix $M \in \mathbb{F}^{m\times n}$ with $\indim_M(t) \geq \rk(M)- r$ we can obtain matrices $A\in \mathbb{F}^{m\times n},B\in \mathbb{F}^{n\times n},C\in \mathbb{F}^{r\times n},M'\in \mathbb{F}^{m \times r}$ such that $A$ is $t$-row sparse, $M'$ is a sub-matrix of $M$ and $M=A\cdot B+ M'\cdot C$. Over large enough finite fields $\mathbb{F}$, such a decomposition can be obtained in polynomial time given an oracle computing inner-dimension of a matrix. As ${\sf InnerDim}(M,d,t)\in \NP$ from Observation \ref{obs:inner-dim}, we have that this decomposition can be computed in $\P^{\NP}$.

Given the above decomposition property we will argue Claim \ref{claim:indim-outdim} that if all the {\em useful} sub-matrices of $M$ have {\em large} inner dimension then $M$ has {\em small} outer-dimension which is a contradiction.

That is, suppose 
$\outdim_M(tk+n\epsilon^k)\geq \frac{n}{1-\epsilon}$  and  $\indim_{M}(t)\leq \rk(M)-\epsilon n$(here $r=\epsilon n$). Then, by the decomposition property, $M= A\cdot B + M' \cdot C$ for some $A\in \mathbb{F}^{m\times n},B\in \mathbb{F}^{n\times n},C\in \mathbb{F}^{r\times n},M'\in \mathbb{F}^{m \times r}$ where $A$ is $t$-row sparse and $M'$ is a sub-matrix of $M$. Further, if $M'$ also does not have the requisite inner-dimension then by recursively applying the decomposition procedure we get: \begin{align*}
M &=  A\cdot B + M' \cdot C \\
&= AB+ A'B'C + M''C'C & [\because M''&=A'B' + M''C']\\
&= AB+ A'B'C + (A''B''+M'''C'')C'C & [\because M'''&=A''B'' + M'''C'']\\
&= AB+ A'B'C + A''B''C'C + M'''C''C'C \\
M &= \begin{bmatrix}
A & A' & A'' & M'''       
\end{bmatrix} \cdot 
\begin{bmatrix}
B \\
B' \\
B''C \\
C''C       
\end{bmatrix}
\end{align*}
assuming none of $M,M',M'',\dots$ and so on have {\em low} inner-dimension, Now after $k$ steps of the decomposition procedure we obtain:

$$M = \begin{bmatrix}
A_0 & A_1 & A_1 & \cdots & A_{k-1} & M_k       
\end{bmatrix} \cdot 
\begin{bmatrix}
N_0 \\
N_1 \\
\vdots \\
\vdots \\
 N_{k} 
\end{bmatrix}$$
where $A_0,\ldots ,A_{k-1}$ are all $t$-row sparse matrices, $M_k \in \mathbb{F}^{m\times n\epsilon^k}$ and $N_0,\ldots ,N_{k}$ are obtained from $B$'s and $C$'s accordingly. It is not difficult to observe that from the above decomposition we get $M=P\cdot Q$ where $P$ has at most $tk+n\epsilon^k$ non-zero entries in each row as the $k$ matrices $A_0,\ldots ,A_{k-1}$ have $t$ non-zero entries per row and $M_k$ has at most $n\epsilon^k$ columns. Further, note that each matrix $N_i$ has dimension $n\epsilon^i\times n$. Hence $Q \in \mathbb{F}^{s\times n}$ where $s=n(1+\epsilon+\epsilon^2 + \cdots +\epsilon^k)$ which is less than $\frac{n}{1-\epsilon}$ for any positive integer $k$ and $\epsilon\in (0,1)$. Thus, from the definition of outer-dimension $\outdim_M(tk+n\epsilon^k)< \frac{n}{1-\epsilon}$ which is a contradiction. \hspace*{108 mm} {\small (End of Claim \ref{claim:indim-outdim})}$\qedh$

%% file: rigidity-codes.tex
\paragraph*{}
{\em Coding theory} essentially deals with detecting and correcting errors in messages transmitted over a noisy channel thereby ensuring reliable communication. Suppose there are two parties {\em Alice} and {\em Bob} and Alice wants to send a message $m\in\{0,1\}^k$ to Bob. Alice encodes the message $m$ using an encoding function $E:\{0,1\}^k \rightarrow \{0,1\}^n$ and send the $c=E(m)$ in $\{0,1\}^n$ over a transmission channel that could potentially be noisy. Here the word $c=E(m)\in \{0,1\}^n$ is called the {\em codeword}. Let $C$ denote the set of all possible codewords in $\{0,1\}^n$.  Now, Bob receives a word $c'\in \{0,1\}^n$ called the {\em received word} and uses a decoding function $D:\{0,1\}^n\rightarrow \{0,1\}^k$ to obtain $m'=D(c')$. In an ideal channel with no noise, $c'=c$.

In other cases, if Bob is able to identify the codeword $c$ from the received word $c'$, then he can get hold of the message $m$ by using $D(c)$. One intuitive way to do this is by designing the encoding algorithm to repeat the message $m$ several times (here $n \gg k$). This redundancy in the codeword is captured in the value $k/n$ called {\em rate} of the code (denoted by $R(C)$). For any code $C$, $R(C)\leq 1$ and is inversely proportional to the actual redundancy. Further,  {\em distance} between two codewords is another important parameter which is the hamming distance (denoted by $\Delta$) between them. Observe that as the distance between two codewords increases, it is unlikely to confuse one codeword for another which intuitively helps detect errors in the received codeword. The relative distance $\delta(C)$ of a code $C$ is $d/n$ where $d= \min\limits_{c\in C} \Delta(0,c)$. An immediate question would be to understand the optimal trade-off between $R(C)$ and $\Delta(C)$. There is huge body of work revolving around this question and we refer the reader to \cite{Pra07} for more details. In this article, we will be interested particularly in {\em linear codes}.

An $[n,k,d]_{q}$ linear code is one where the set $C$ of codewords is a linear subspace of $\mathbb{F}_q^n$ of dimension $k$ and distance of the code is $d$. Observe that every codeword of a linear code can be obtained as a linear combination of the rows of an $n\times n$ {\em generator} matrix $G_C$. Now that we have associated matrices with codes, it is natural to ask how rigid the generator matrices of  codes are?

To begin with, we demonstrate a connection between coding theory({\em asymptotically good codes}) and matrix rigidity.

\subsection{Rigidity of generator matrices of asymptotically good codes}

Asymptotically good codes are family of codes whose rate and relative distance are both constant in the asymptotic sense.

\begin{definition}
\label{def:alg-codes}
A family of codes ${\cal C} = \{C_i\}_{i\geq 1}, C_i = [n_i,k_i,d_i]_{q}$  is said to be {\em asymptotically good} if there exists constants $R_0,\delta_0>0$ such that $\lim\limits_{n\rightarrow \infty} \frac{k_i}{n_i}\geq R_0$ and $\lim\limits_{n\rightarrow \infty} \frac{d_i}{n_i}\geq \delta_0$. 
\end{definition}

Using algebraic geometric codes\cite{TLP98}, we can prove the existence of asymptotically good error correcting codes.  We state the lemma about the existence of asymptotically good error correcting codes without giving a proof. For a proof see Theorem 2.81 in \cite{TLP98}. 

\begin{lemma} 
\label{lem:asym-good-code}
Let $\mathbb{F}_q$ be a finite field.  For infinitely many $n$, there exists $[2n,n,d]_{q}$ code with rate $1/2$ and relative distance at least $1-\frac{2}{\sqrt{q}-1}$.
\end{lemma}

For the above code let $G$ denote the generator matrix and  the generator matrix can be brought to the standard form $G_C=[ I_n \mid A]$  where $I_n$ is the $n\times n$ identity matrix and $A$ is a $n\times n$ matrix. In the following theorem, we prove that the matrix $A$ has high rigidity over $\mathbb{F}_q$:

\begin{theorem}
Let $A\in\mathbb{F}_q^{n\times n}$ be a  the matrix obtained from the standard form of the generator matrix of the $[2n,n,(1-\epsilon)n]$ code for $\epsilon = \frac{2}{\sqrt{q}-1}$ as in Lemma \ref{lem:asym-good-code}. Then $R_A(r) = \Omega(\frac{n^2}{r}\log\frac{n}{r})$ for $\epsilon n \leq r \leq n/2$.
\end{theorem}

\begin{proof}
Let $\epsilon n \leq r \leq n/2$ and $A'$ be a $2(r+1)\times 2(r+1)$ submatrix of $A$. We claim that $\rk(A') \geq r+1$. Suppose not, $\rk(A') < r+1$. Then, there exists a codeword of weight $n-(r+1) < n-\epsilon n$. Hence, minimum distance of the code is at most $n(1-\epsilon)$, a contradiction. This implies that every $2(r+1)\times 2(r+1)$ sub-matrix of $A$ has rank at least  $r+1$. Now, by following an argument similar to the untouched minor argument, we get the required lower bound.
\end{proof}

\begin{remark}
\label{rem:dvir}
Although matrices of high rigidity can be obtained from generator matrices of asymptotically
good linear codes, \cite{Dvir16} obtained a distribution ${\cal D}$ of matrices such that for $G \sim D$, $G$ generates a good linear code but with high probability $R_G(r) \leq  O(n^2/r)$ for any $r \leq  O(r \log(\frac{n}{r}))$.
\end{remark}

Next, we review a result of Dvir\cite{Dvir16} which states that if the generating matrix $G_C$ of any locally decodable code $C$ is not row rigid then there exists a locally self-correctable code $C'$ with rate of $C'$ is $\approx$ 1.

\subsection{Locally self-correctable codes and rigid matrices}

The focus of this subsection is to review the connections between {\em locally decodable codes} or {\em locally self-correctable codes} and matrix rigidity which is the main result of \cite{Dvir11}. Informally a {\em locally decodable code}(LDC) is an error-correcting code that enables probabilistically decoding a particular symbol of the message by querying a {\em small} number of locations of the corresponding codeword even when the codeword is corrupted in a {\em few} locations while a {\em locally self correctable code}(LCC) is an error-correcting code that enables probabilistically decoding bits of the codeword rather than the message which can be viewed as self-correcting the corrupted codeword. For any vector $v$, we denote by $w(v)$ the Hamming weight of the vector $v$. We give the formal definitions below:

\begin{definition}
[Locally decodable code.]
A $(q,\delta,\epsilon)$-LDC $C$ is a linear map $C:\mathbb{F}_p^n\rightarrow \mathbb{F}_p^m$ such that there is a randomized decoding algorithm $D:\mathbb{F}_p^m\times [n]\rightarrow \mathbb{F}_p$ that on input $(c+u,i)$ queries at most $q$ locations in $c+u$ and recovers with probability at least $1-\epsilon$, the $i^{th}$ bit of message $x$ from $c+u$  where $c =C(x)$ and $w(u)\leq \delta\cdot n$ (i.e.,  codeword $c$ is corrupted in at most $\delta\cdot n$  locations). 
\end{definition}

\begin{definition}[Locally self-correctable code.] A $(q,\delta,\epsilon)$-LCC is a linear map $C':\mathbb{F}_p^n\rightarrow \mathbb{F}_p^m$
such that there is a randomized (self-correcting) algorithm $D':\mathbb{F}_p^m\times [n]\rightarrow \mathbb{F}_p$ that on input $(c+u,i)$ queries at most $q$ locations in $c+u$ and recovers with probability at least $1-\epsilon$, the $i^{th}$ bit of codeword $c$ from $c+u$  where $u\in \mathbb{F}_p^n$ with $w(u)\leq \delta\cdot n$ (i.e.,  codeword $c$ is corrupted in at most $\delta\cdot n$  locations). 
\end{definition}

We say an error-correcting code $C$ is {\em explicit} if every entry of the generator matrix can be obtained in deterministic polynomial time. It is interesting to note the following  explicit constructions of locally decodable code from \cite{Dvir11} which will be useful for our purpose. We do not prove this construction here(for proof, see Corollary 3.3 in \cite{Dvir11})

\begin{theorem}
\label{thm:ldc-existance}
For any $\epsilon,a>0$, there exists an explicit family of codes  $C_n:\mathbb{F}_p^n\rightarrow \mathbb{F}_p^m$ such that $C_n$ is a $(n^a,\delta,\epsilon)$-LDC with $m=O(n)$ and $\delta=\delta(\epsilon)>0$.
\end{theorem}

We now state the main theorem of \cite{Dvir11} showing that if the generating matrix $G_C$ of any locally decodable code $C$ is not row rigid then there exists a locally self-correctable code $C'$ with dimension close to $n$. We first give a sketch of the proof and then move on to the details.

\begin{theorem}
\label{thm:rigidity-selfcorrectable-codes}
Let $C:\mathbb{F}_p^n\rightarrow \mathbb{F}_p^m$ be a $(q,\delta,\epsilon)$ locally decodable code whose generator matrix $G_C$ has $RR_{G_C}(r)\leq s$. Then for any $\rho >0$, there exists a $(q',\delta',\epsilon)$-LCC $C'$ (a subspace of $\mathbb{F}_p^n$) with 
$q'=qs, \delta' = (\rho\delta)/s$ and dimension of $C'$ being $n(1-\rho)-r$.
\end{theorem}

\noindent\textit{Proof Sketch.} Suppose $G_C$ has row rigidity at most $s$ for rank $r$ then $G_C=S+L$ where $\rk(L)$ is {\em low} and every row of $S$ has at most $s$ non-zero entries. Since $\rk(L)$ is {\em low} to construct an LCC $C'$ of sufficiently large dimension a natural candidate for $C'$ is the ${\sf nullspace}(L)$. When $C'={\sf nullspace}(L)$, dimension of $C$ is $n-\rk(L)$ which is {\em large}(as $\rk(L)$ is low). We need to ensure that $C'$ is $(q',\delta',\epsilon)$ locally self-correctable. That is, to decode the $i^{th}$ symbol of the codeword $c$ which is corrupted in at most $\delta'$ locations we need a decoding algorithm $D'$ that on input $D(c+v)$ ($w(v)\leq \delta'\cdot n$) outputs $c_i$ with probability $1-\epsilon$. Observe that for every $c\in C'$, $$C(c)=G_C\cdot c = S\cdot c + L\cdot c = S\cdot c$$
 as $L\cdot c =0$ for $C'={\sf nullspace}(L)$.

 Now, it is sufficient to invoke the local decoding algorithm for the LDC $C$ with $(C(c)+v',i)$  as input where $v'\triangleq S\cdot v$. Here, weight of $v'$ is {\em small} as matrix $S$ is $s$-row-sparse. the algorithm $D$ that locally decodes the LDC $C$ returns $c_i$ with probability $1-\epsilon$ by querying a small number of locations as $C(c)+v' = S(c+v)$.(For technical reasons we cannot quite work with the matrix $S$ but we will construct a slightly modified matrix $S'$ from $S$ obtained in Observation \ref{obs:col-sparse}.) 
 
 $\qedh$

We now explain all the details mentioned in the above proof idea. We will need the following simple observation that for any row sparse matrix, the columns can also be made fairly sparse without increasing the rank by much. The proof appeals to the intuition that if too many columns of an $s$ row-sparse matrix are {\em dense} then we can find a row that is not $s$-sparse.

\begin{obs}
\label{obs:col-sparse}
Let $\rho>0$ and $A\in\mathbb{F}^{m\times n}$ be 
any matrix with $RR_A(r)\leq s$ (i.e., $A=S+L$ where $\rk(L)\leq r$ and $S$ is $s$-row-sparse). Then, $A=S'+L'$ where $\rk(L')\leq r+\rho\cdot n$ and every column of $S'$ has at most $(s\cdot m)/
(\rho\cdot n)$ non-zero entries.  
\end{obs}

\noindent\textit{Proof of Observation \ref{obs:col-sparse}.} The number of non-zero entries in $A$ is at most $s\cdot n$. For any $\rho>0$, the number of columns with at least $(s\cdot m)/(\rho\cdot n)$ non-zero entries is at most $\rho\cdot n$. Let $C_{i_1},C_{i_2},\ldots,C_{i_j}, j\in [\rho\cdot n]$ be the columns in $S$ with at least $(s\cdot m)/(\rho\cdot n)$ non-zero entries. Let $S'$ be the matrix obtained by replacing columns $C_{i_1},C_{i_2},\ldots,C_{i_j}$ in $S$ with all zeros vectors. Let $L'$ be the matrix obtained by adding to the $i_j^{th}$ column of  $L$ the column vector $C_{i_j}$ for all $j\in [\rho\cdot n]$. Then $A=L'+S'$ where $\rk(L')\leq r+\rho\cdot n$ and every row of $S'$ has at most $s$ non-zero entries and every column of $S'$ has at most $(s\cdot m)/(\rho\cdot n)$ non-zero entries. 

$\qedh$

Now, we complete proof of theorem \ref{thm:rigidity-selfcorrectable-codes}.

\begin{proof}
Let $C:\mathbb{F}_p^n\rightarrow \mathbb{F}_p^m$ be a $(q,\delta,\epsilon)$-LDC and $G_C\in \mathbb{F}^{m\times n}$ be its generator matrix. Suppose $G_C$ has row rigidity at most $s$ for rank $r$ then by Observation \ref{obs:col-sparse}, $G_C=S'+L'$ where $
\rk(L')\leq r+\rho\cdot n$ and every row of $S'$ has at most $s$ non-zero entries and every column of $S'$ has at most $(s\cdot m)/(\rho\cdot n)$ non-zero entries. Let $C'\triangleq {\sf nullspace(L')}$. Then dimension of $C'$ as a subspace of $\mathbb{F}_p^n$ is $n-\rk(L') \geq n-(r+\rho\cdot n) = n(1-\rho)-r$.

It remains to show that $C'$ is a $(q',\delta',\epsilon)$-LCC where $q'=qs$ and $\delta'=\rho\delta/s$. In particular, we need a randomized algorithm $D':\mathbb{F}_p^n \times [n]\rightarrow \mathbb{F}_p$ that decodes (with probability at least $1-\epsilon$) a particular symbol of a codeword $c\in \mathbb{F}_p^n$ that is corrupted in at most $\delta'\cdot n$ locations by querying at most $q'$ locations of the corrupted codeword. Since $C$ is a $(q,\delta,\epsilon)$-LDC, we have at our disposal a randomized algorithm $D:\mathbb{F}_p^m \times [n]\rightarrow \mathbb{F}_p$ that decodes (with probability at least $1-\epsilon$) a particular symbol  of a message $x\in\mathbb{F}_p^n$
by querying at most $q$ locations in the corresponding codeword $C(x)$ which is corrupted in at most $\delta\cdot m$ locations. The main idea is to make $D'$ run $D$ on appropriate inputs.  Note that $D$
can correct message symbols only when the codeword is corrupted in at most $\delta m$ locations. The input to $D'$ is $(c+v,i)$ where $i\in [n], c\in \mathbb{F}_p^n$ and $v\in \mathbb{F}_p^n$ with $w(v)\leq \delta'\cdot n$. \\

The idea is to encode $c\in \mathbb{F}_p^n$ using the LDC $C$ and then use decoding algorithm $D$ on $C(c)$ to correct the $i^{th}$ bit of codeword $c_i$. Let $v'=S'\cdot v$ be a vector in $\mathbb{F}_p^m$. Observe the following:
\begin{itemize}
\item The weight of vector $v'$ is at most $\delta\cdot m$ as $w(v)\leq \delta'\cdot n$ and every column of $S'$ has at most $(s\cdot m)/
(\rho\cdot n)$ non-zero entries.
\item $D(C(c)+v',i)$ outputs $c_i$ (the $i^{th}$ bit of $c\in \mathbb{F}_p^n$) with probability $1-\epsilon$ by making at most $q$ queries to $C(c)+v'$.
\item For every $c\in C'$, $C(c)+v'= S'\cdot m + v' = S'\cdot (c+v)$. As $S'$ has at most $s $non-zero entries in every row, $D'$ makes at most $qs$ queries overall before returning $c_i$.
\end{itemize}
Thus, $C'$ is a $(q',\delta',\epsilon)$-LCC where $q'=qs$ and $\delta'=\rho\delta/s$. 
 \end{proof}

%% file: discussion.tex
This article is entirely based on the problem of matrix rigidity and its multiple connections to other central problems in theoretical computer science such as static data structure lower bounds, error-correcting codes and communication complexity. By now, the reader is probably convinced of the harsh $\mathbb{R}$eality of rigid matrices. Now, we mention a few open questions:

\begin{enumerate}
\item One of the foremost open problems is to answer Valiant's Question \ref{que:rigidity-Valiant} or even Razborov's Question \ref{que:rigidity-Razborov} by constructing explicit matrices of high rigidity.
We have thus far been able to obtain explicit constructions of rigid matrices in the class $\P^{\NP}$. 


\item One of the matrix families that we have not analysed so far is the {\em incidence matrices of projective planes} from the conjecture on Page 2. In \cite{DSW14}, the authors show that the {\em monotone rigidity} of  incidence matrices of projective planes is $\alpha n$ for rank $\alpha\sqrt{n}$(for some $\alpha>0$) where monotone rigidity means that only non-zero entries can be changed to reduce the rank of $A$. Obtaining upper or lower bounds on the rigidity of such matrices remains largely open.

\item On the computational front, what is the complexity of ${\sf RIGID}(A,\mathbb{Q},s,r)$?

\item The {\em matrix factorization problem} is seemingly the dual of matrix rigidity where the goal is to construct an explicit matrix that cannot be expressed as a product of sparse matrices. That is, we want an explicit matrix $A\in \mathbb{F}^{n\times n}$ such that if $A = A_1 \cdot A_2 \cdots A_d$ then $\spar(A_i) = \Omega(n^{1+\delta})$ for some $i \in [d]$ and $\delta > 0$. The best known lower bound for matrix factorization is  $\Omega(n\cdot \lambda_d(n))$ for some small-growing function $\lambda_d(n)$. In \cite{KV19}, authors obtain $\Omega(n^2)$ lower bounds for matrix factorization when $d = 2$ and $A_i$'s are symmetric or invertible matrices. It would be interesting to study the matrix factorization problem for other special matrices as well as in total generality.

\item In connection with error-correcting codes, can we obtain explicit constructions of good linear error-correcting codes whose generator matrices have low rigidity? A standard methodology is to use techniques from {\em derandomization toolkit} to derandomize the result of \cite{Dvir16} mentioned in Remark \ref{rem:dvir}.

\end{enumerate}

%% file: survey-final.bbl
\begin{thebibliography}{10}

\bibitem{AC19}
Josh Alman and Lijie Chen.
\newblock {Efficient Construction of Rigid Matrices Using an {NP} Oracle}.
\newblock In David Zuckerman, editor, {\em 60th {IEEE} Annual Symposium on
  Foundations of Computer Science, {FOCS}2019, Baltimore, Maryland, USA,
  November 9-12, 2019}, pages 1034--1055. {IEEE} Computer Society, 2019.
\newblock \href {https://doi.org/10.1109/FOCS.2019.00067}
  {\path{doi:10.1109/FOCS.2019.00067}}.

\bibitem{AW17}
Josh Alman and R.~Ryan Williams.
\newblock Probabilistic rank and matrix rigidity.
\newblock In Hamed Hatami, Pierre McKenzie, and Valerie King, editors, {\em
  Proceedings of the 49th Annual {ACM} {SIGACT} Symposium on Theory of
  Computing, {STOC} 2017, Montreal, QC, Canada, June 19-23, 2017}, pages
  641--652. {ACM}, 2017.
\newblock \href {https://doi.org/10.1145/3055399.3055484}
  {\path{doi:10.1145/3055399.3055484}}.

\bibitem{BHPT20}
Amey Bhangale, Prahladh Harsha, Orr Paradise, and Avishay Tal.
\newblock Rigid {M}atrices {F}rom {R}ectangular {PCP}s.
\newblock {\em CoRR}, abs/2005.03123, 2020.
\newblock URL: \url{https://arxiv.org/abs/2005.03123}, \href
  {http://arxiv.org/abs/2005.03123} {\path{arXiv:2005.03123}}.

\bibitem{CLP17}
Ernie Croot, Vsevolod~F. Lev, and Péter~Pál Pach.
\newblock Progression-free sets in {${\mathrm{\mathbb{Z}}}_{4}^{\mathrm{n}}$}
  are exponentially small.
\newblock {\em Annals of Mathematics}, 185(1):331--337, 2017.
\newblock URL: \url{https://annals.math.princeton.edu/2017/185-1/p07}.

\bibitem{Des07}
Amit~Jayant Deshpande.
\newblock {Sampling-Based Algorithms for Dimension Reduction}.
\newblock {\em PhD Thesis}, 2007.
\newblock URL:
  \url{https://dspace.mit.edu/bitstream/handle/1721.1/38935/166267550-MIT.pdf;sequence=2}.

\bibitem{DFGS91}
Alicia Dickenstein, Noaï Fitchas, Marc Giusti, and Carmen Sessa.
\newblock The membership problem for unmixed polynomial ideals is solvable in
  single exponential time.
\newblock {\em Discrete Applied Mathematics}, 33(1):73 -- 94, 1991.
\newblock \href {https://doi.org/https://doi.org/10.1016/0166-218X(91)90109-A}
  {\path{doi:https://doi.org/10.1016/0166-218X(91)90109-A}}.

\bibitem{Dvir11}
Zeev Dvir.
\newblock {On Matrix Rigidity and Locally Self-correctable Codes}.
\newblock {\em Comput. Complex.}, 20(2):367--388, 2011.
\newblock \href {https://doi.org/10.1007/s00037-011-0009-1}
  {\path{doi:10.1007/s00037-011-0009-1}}.

\bibitem{Dvir16}
Zeev Dvir.
\newblock On the non-rigidity of generating matrices of good codes (written by
  oded goldreich).
\newblock 2016.

\bibitem{DE17}
Zeev Dvir and Benjamin~L. Edelman.
\newblock Matrix rigidity and the {C}root-{L}ev-{P}ach {L}emma.
\newblock {\em Theory of Computing}, 15(8):1--7, 2019.
\newblock URL: \url{http://www.theoryofcomputing.org/articles/v015a008}, \href
  {https://doi.org/10.4086/toc.2019.v015a008}
  {\path{doi:10.4086/toc.2019.v015a008}}.

\bibitem{DGW19}
Zeev Dvir, Alexander Golovnev, and Omri Weinstein.
\newblock Static data structure lower bounds imply rigidity.
\newblock In {\em Proceedings of the 51st Annual {ACM} {SIGACT} Symposium on
  Theory of Computing, {STOC} 2019, Phoenix, AZ, USA, June 23-26, 2019}, pages
  967--978, 2019.
\newblock \href {https://doi.org/10.1145/3313276.3316348}
  {\path{doi:10.1145/3313276.3316348}}.

\bibitem{DL19}
Zeev Dvir and Allen Liu.
\newblock {Fourier and Circulant Matrices Are Not Rigid}.
\newblock In {\em 34th Computational Complexity Conference, {CCC} 2019, July
  18-20, 2019, New Brunswick, NJ, {USA}}, pages 17:1--17:23, 2019.
\newblock \href {https://doi.org/10.4230/LIPIcs.CCC.2019.17}
  {\path{doi:10.4230/LIPIcs.CCC.2019.17}}.

\bibitem{DSW14}
Zeev Dvir, Shubhangi Saraf, and Avi Wigderson.
\newblock {Improved rank bounds for design matrices and a new proof of
  Kelly’s theorem}.
\newblock {\em Forum of Mathematics, Sigma}, 2, 2014.
\newblock URL:
  \url{https://www.math.ias.edu/~avi/PUBLICATIONS/DvirSaWi2015.pdf}.

\bibitem{EF75}
Peter Elias and Richard~A. Flower.
\newblock {The Complexity of Some Simple Retrieval Problems}.
\newblock {\em J. {ACM}}, 22(3):367--379, 1975.
\newblock \href {https://doi.org/10.1145/321892.321899}
  {\path{doi:10.1145/321892.321899}}.

\bibitem{Fri93}
Joel Friedman.
\newblock A note on matrix rigidity.
\newblock {\em Combinatorica}, 13(2):235--239, 1993.
\newblock \href {https://doi.org/10.1007/BF01303207}
  {\path{doi:10.1007/BF01303207}}.

\bibitem{GT18}
Oded Goldreich and Avishay Tal.
\newblock Matrix rigidity of random {T}oeplitz matrices.
\newblock {\em Computational Complexity}, 27(2):305--350, Jun 2018.
\newblock \href {https://doi.org/10.1007/s00037-016-0144-9}
  {\path{doi:10.1007/s00037-016-0144-9}}.

\bibitem{GPW18}
Mika G{\"{o}}{\"{o}}s, Toniann Pitassi, and Thomas Watson.
\newblock {The Landscape of Communication Complexity Classes}.
\newblock {\em Comput. Complex.}, 27(2):245--304, 2018.
\newblock \href {https://doi.org/10.1007/s00037-018-0166-6}
  {\path{doi:10.1007/s00037-018-0166-6}}.

\bibitem{Gowers17}
Gower.
\newblock Reflections on the recent solution of the cap-set problem i.
\newblock URL:
  \url{https://gowers.wordpress.com/2016/05/19/reflections-on-the-recent-solution-of-the-cap-set-problem-i/}.

\bibitem{Gri76}
D.~Yu Grigoriev.
\newblock {Using the notions of seperability and independence for proving the
  lower bounds on the circuit complexity(in Russian). Notes of the Leningrad
  branch of the Steklov Mathematical Institute, Nauka}.
\newblock 1976.

\bibitem{Pra07}
Prahladh Harsha.
\newblock {A Course on PCPs, codes and inapproximability}.
\newblock 2007.
\newblock URL: \url{http://www.tcs.tifr.res.in/~prahladh/teaching/07autumn/}.

\bibitem{TLP98}
Tom Høholdt, Jacobus~H. van Lint, and Ruud Pellikaan.
\newblock Algebraic geometry codes.
\newblock 1998.
\newblock URL:
  \url{https://people.csail.mit.edu/dmoshkov/courses/codes/lec7-AG-codes.pdf}.

\bibitem{KLPS14}
Abhinav Kumar, Satyanarayana~V. Lokam, Vijay~M. Patankar, and Jayalal Sarma.
\newblock {Using Elimination Theory to Construct Rigid Matrices}.
\newblock {\em Comput. Complex.}, 23(4):531--563, 2014.
\newblock \href {https://doi.org/10.1007/s00037-013-0061-0}
  {\path{doi:10.1007/s00037-013-0061-0}}.

\bibitem{KV19}
Mrinal Kumar and Ben~Lee Volk.
\newblock Lower bounds for matrix factorization.
\newblock In Shubhangi Saraf, editor, {\em 35th Computational Complexity
  Conference, {CCC} 2020, July 28-31, 2020, Saarbr{\"{u}}cken, Germany (Virtual
  Conference)}, volume 169 of {\em LIPIcs}, pages 5:1--5:20. Schloss Dagstuhl -
  Leibniz-Zentrum f{\"{u}}r Informatik, 2020.
\newblock \href {https://doi.org/10.4230/LIPIcs.CCC.2020.5}
  {\path{doi:10.4230/LIPIcs.CCC.2020.5}}.

\bibitem{Lar14}
Kasper~Green Larsen.
\newblock {On Range Searching in the Group Model and Combinatorial
  Discrepancy}.
\newblock {\em {SIAM} J. Comput.}, 43(2):673--686, 2014.
\newblock \href {https://doi.org/10.1137/120865240}
  {\path{doi:10.1137/120865240}}.

\bibitem{Lok00}
Satyanarayana~V. Lokam.
\newblock {On the rigidity of Vandermonde matrices}.
\newblock {\em Theor. Comput. Sci.}, 237(1-2):477--483, 2000.
\newblock \href {https://doi.org/10.1016/S0304-3975(00)00008-6}
  {\path{doi:10.1016/S0304-3975(00)00008-6}}.

\bibitem{Lok06}
Satyanarayana~V. Lokam.
\newblock {Quadratic Lower Bounds on Matrix Rigidity}.
\newblock In {\em Theory and Applications of Models of Computation, Third
  International Conference, {TAMC} 2006, Beijing, China, May 15-20, 2006,
  Proceedings}, volume 3959 of {\em Lecture Notes in Computer Science}, pages
  295--307. Springer, 2006.
\newblock \href {https://doi.org/10.1007/11750321\_28}
  {\path{doi:10.1007/11750321\_28}}.

\bibitem{Lok09}
Satyanarayana~V. Lokam.
\newblock {Complexity Lower Bounds using Linear Algebra}.
\newblock {\em Foundations and Trends in Theoretical Computer Science},
  4(1-2):1--155, 2009.
\newblock \href {https://doi.org/10.1561/0400000011}
  {\path{doi:10.1561/0400000011}}.

\bibitem{MS10}
Meena Mahajan and Jayalal Sarma.
\newblock {On the Complexity of Matrix Rank and Rigidity}.
\newblock {\em Theory Comput. Syst.}, 46(1):9--26, 2010.
\newblock \href {https://doi.org/10.1007/s00224-008-9136-8}
  {\path{doi:10.1007/s00224-008-9136-8}}.

\bibitem{PTW10}
Rina Panigrahy, Kunal Talwar, and Udi Wieder.
\newblock {Lower Bounds on Near Neighbor Search via Metric Expansion}.
\newblock In {\em 51th Annual {IEEE} Symposium on Foundations of Computer
  Science, {FOCS} 2010, October 23-26, 2010, Las Vegas, Nevada, {USA}}, pages
  805--814. {IEEE} Computer Society, 2010.
\newblock \href {https://doi.org/10.1109/FOCS.2010.82}
  {\path{doi:10.1109/FOCS.2010.82}}.

\bibitem{Raz89}
Alexander~A. Razborov.
\newblock {On rigid matrices (in Russian)}.
\newblock 1989.
\newblock URL: \url{http://people.cs.uchicago.edu/~razborov/files/rigid.pdf}.

\bibitem{SSS97}
Mohammad~Amin Shokrollahi, Daniel~A. Spielman, and Volker Stemann.
\newblock {A Remark on Matrix Rigidity}.
\newblock {\em Inf. Process. Lett.}, 64(6):283--285, 1997.
\newblock \href {https://doi.org/10.1016/S0020-0190(97)00190-7}
  {\path{doi:10.1016/S0020-0190(97)00190-7}}.

\bibitem{Sho88}
Victor Shoup.
\newblock {New Algorithms for Finding Irreducible Polynomials over Finite
  Fields}.
\newblock In {\em 29th Annual Symposium on Foundations of Computer Science,
  White Plains, New York, USA, 24-26 October 1988}, pages 283--290. {IEEE}
  Computer Society, 1988.
\newblock \href {https://doi.org/10.1109/SFCS.1988.21944}
  {\path{doi:10.1109/SFCS.1988.21944}}.

\bibitem{Tao17}
Terence Tao.
\newblock Open question: best bounds for cap sets.
\newblock URL:
  \url{https://terrytao.wordpress.com/2007/02/23/open-question-best-bounds-for-cap-sets/}.

\bibitem{Val77}
Leslie~G. Valiant.
\newblock {Graph-Theoretic Arguments in Low-Level Complexity}.
\newblock In Jozef Gruska, editor, {\em Mathematical Foundations of Computer
  Science 1977, 6th Symposium, Tatranska Lomnica, Czechoslovakia, September
  5-9, 1977, Proceedings}, volume~53 of {\em Lecture Notes in Computer
  Science}, pages 162--176. Springer, 1977.
\newblock \href {https://doi.org/10.1007/3-540-08353-7\_135}
  {\path{doi:10.1007/3-540-08353-7\_135}}.

\bibitem{Wun12}
Henning Wunderlich.
\newblock {On a Theorem of Razborov}.
\newblock {\em Computational Complexity}, 21(2):431--477, 2012.
\newblock \href {https://doi.org/10.1007/s00037-011-0021-5}
  {\path{doi:10.1007/s00037-011-0021-5}}.

\bibitem{Yao81}
Andrew~Chi{-}Chih Yao.
\newblock {Should Tables Be Sorted?}
\newblock {\em J. {ACM}}, 28(3):615--628, 1981.
\newblock \href {https://doi.org/10.1145/322261.322274}
  {\path{doi:10.1145/322261.322274}}.

\bibitem{Zak83}
Stanislav Žák.
\newblock A turing machine time hierarchy.
\newblock {\em Theoretical Computer Science}, 26(3):327 -- 333, 1983.
\newblock URL:
  \url{http://www.sciencedirect.com/science/article/pii/0304397583900154},
  \href {https://doi.org/https://doi.org/10.1016/0304-3975(83)90015-4}
  {\path{doi:https://doi.org/10.1016/0304-3975(83)90015-4}}.

\end{thebibliography}
